\documentclass[11pt,dvipsnames]{article}
\usepackage[margin = 1in, footskip=0.3in]{geometry}
\PassOptionsToPackage{toc,page,header}{appendix}
\usepackage{amsfonts, amssymb, amsthm, amsmath, mathtools}
\usepackage{comment, dsfont}
\usepackage{enumitem, multicol, multirow}
\RequirePackage{fix-cm}
\usepackage{graphicx, float}
\usepackage{caption, subcaption}
\DeclareCaptionFont{slightlysmall}{\fontsize{10pt}{12pt}\selectfont}
\usepackage{url}
\usepackage{hyperref}
\usepackage[dvipsnames,svgnames]{xcolor}
\definecolor{cobalt}{rgb}{0.0, 0.28, 0.67}
\definecolor{rblue}{rgb}{0.25, 0.35, 0.98}
\definecolor{babyblue}{rgb}{0.54, 0.81, 0.94}
\definecolor{teal}{rgb}{0.0, 0.5, 0.5}
\definecolor{crimson}{rgb}{0.86, 0.08, 0.24}

\usepackage[toc, page]{appendix}
\usepackage{cancel}
\usepackage{pdfpages}
\usepackage{booktabs}

\DeclareMathAlphabet{\mathmybb}{U}{bbold}{m}{n}
\newcommand{\1}{\text{ $\mathbf{1}$}}
\newcommand{\ind}[1]{\1{\left\{#1\right\}}} 
\renewcommand{\cal}[1]{\mathcal{#1}}

\newcommand{\eps}{\varepsilon}

\newcommand{\R}{\mathbb{R}}

\def\N{\mathbb{N}}

\newcommand{\wh}[1]{\widehat{#1}} 
\newcommand{\wt}[1]{\widetilde{#1}} 

\newcommand\numberthis{\addtocounter{equation}{1}\tag{\theequation}}
\renewcommand{\qed}{\hfill$\square$}
\def\E{\mathbb{E}} 
\def\P{\mathbb{P}} 
\newcommand{\iid}{\stackrel{i.i.d.}{\sim}}
\newcommand{\indep}{\perp \!\!\!\! \perp}
\newcommand{\var}{{\rm Var}} 
\newcommand{\cov}{{\rm Cov}} 
\newcommand{\normal}{\mathcal{N}}

\newcommand{\ber}{\mathsf{Ber}}

\newcommand{\dto}{\stackrel{d}{\longrightarrow}}
\newcommand{\Pto}{\stackrel{P}{\longrightarrow}}

\def\o{o_p}
\def\O{O_p}
\providecommand{\argmin}{\mathop{\rm argmin}}
\providecommand{\arginf}{\mathop{\rm arginf}} 

\newcommand{\dkl}{D_{\mathrm{KL}}}

\usepackage{algorithm}
\usepackage[noend]{algpseudocode}

\usepackage{soul}

\usepackage{fancyhdr}

\newcommand{\trt}{A}

\newcommand{\tauR}{\overline\tau^\mathtt{R}}
\setstcolor{red}

\newcommand{\good}[1]{\textcolor{DarkGreen}{\mathbf{#1}}}
\newcommand{\better}[1]{#1}
\newcommand{\bad}[1]{\textcolor{red}{\mathbf{#1}}}

\usepackage{tikz}
\usetikzlibrary{arrows.meta, positioning}
\usepackage[dvipsnames,svgnames]{xcolor}
\definecolor{darkgreen}{rgb}{0.0, 0.45, 0.0}
\definecolor{grayred}{rgb}{0.8, 0.3, 0.3}

\usetikzlibrary{positioning,calc}

\newcommand{\rc}{r}
\newcommand{\tauAIPW}[1]{\tau_{#1}^{\,\mathtt{AIPW}}(D;P)}
\newcommand{\tauAIPWobs}[1]{\wh\tau^{\,\mathtt{AIPW}}_{#1}(D)}
\newcommand{\nam}{\mathtt{SR}}
\newcommand{\namp}{\nam+}
\newcommand{\tauRp}{\overline\tau^{\mathtt{R}+}}
\newcommand{\htauRp}{\wh\tau^{\nam+}_{bc}}
\newenvironment{keywords}{\noindent\textit{\textbf{Keywords:}}}{}
\definecolor{myblue}{RGB}{0,80,255}

\def\sP{\mathcal{P}}

\newcommand{\CC}{\mathcal{C}}
\newcommand{\CL}{\mathcal{L}}

\newcommand{\haipw}[1]{\wh\tau^\mathtt{AIPW}_{#1}(Y, A, X)}
\newcommand{\aipw}[1]{\tau^\mathtt{AIPW}_{#1}(Y,A,X)}
\newcommand{\daipw}[1]{\delta^\mathtt{AIPW}_{#1}(Y,A,X)}
\newcommand{\Daipw}[1]{\Delta\tau^\mathtt{AIPW}_{#1}(Y,A,X)}
\usepackage[authoryear]{natbib}
\usepackage[capitalise,nameinlink]{cleveref}
\usepackage{thmtools}
\usepackage[bottom]{footmisc}
\usepackage{minitoc}

\hypersetup{colorlinks = true, citecolor = blue, linkcolor = blue, filecolor = blue,
urlcolor = blue}
\theoremstyle{plain}

\newtheorem{theorem}{Theorem}[section]
\newtheorem{lemma}[theorem]{Lemma}
\theoremstyle{definition}
\newtheorem{definition}[theorem]{Definition}
\newtheorem{example}{Example}
\newtheorem{assumption}{Assumption}
\crefname{lemma}{Lemma}{Lemmas}
\crefname{figure}{Figure}{Figures}
\crefname{assumption}{Assumption}{Assumptions}
\newtheorem{proposition}[theorem]{Proposition}
\crefname{proposition}{Proposition}{Propositions}
\newtheorem{remark}[theorem]{Remark}

\crefname{myexample}{Example}{Examples}
\newtheorem{corollary}[theorem]{Corollary}
\crefname{corollary}{Corollary}{Corollaries}

\begin{document}

\title{Which Covariates to Adjust for? Specification-robust Causal Inference in Observational Studies}

\author{Aditya Ghosh\\
  \texttt{ghoshadi@stanford.edu}
  \and
  Dominik Rothenhäusler\\
  \texttt{rdominik@stanford.edu}
}
\date{Department of Statistics, Stanford University\\[1em] \today}
\maketitle

\begin{abstract}
In observational causal inference, domain knowledge often leaves multiple covariate adjustments plausible, yet which of them satisfy ignorability is not testable. Different adjustments can yield conflicting conclusions about the average treatment effect, and standard remedies (e.g., adjusting for their union or intersection, or reporting the union or convex hull of the confidence intervals) can fail or produce confidence sets that do not shrink to a point with growing sample size. We propose a specification-robust procedure that returns a single point estimate and a confidence interval that is valid when at least one candidate adjustment satisfies ignorability, and has width shrinking at the parametric $\smash{1/\sqrt{n}}$ rate. Our approach mirrors how trimming and overlap weighting handle overlap violations:~We shift the target to a reweighted population, closest in KL divergence to the original population, for which credible, specification-robust inference is feasible. We also provide diagnostics to assess the population shift and an extension to protect any function of the covariates used for reweighting, similar to adding calipers in matching. Synthetic and real-data examples demonstrate that the proposed procedure provides substantially tighter confidence intervals than the convex hull while achieving nominal coverage for the new target population. 
\end{abstract}

\medskip

\begin{keywords}
    Causal inference, covariate adjustment, ignorability, model misspecification, multiverse analysis, observational studies, reweighted population, robustness.
\end{keywords}

\doparttoc
\faketableofcontents

\allowdisplaybreaks

\section{Introduction}\label{sec:intro}

In observational studies, covariate adjustment is a standard approach for removing confounding when estimating the average treatment effect (ATE). Valid inference typically relies on two identification conditions: Ignorability (no unmeasured confounding given the chosen covariates) and overlap (positivity). However, domain knowledge often admits several plausible adjustment sets, with little to no empirical basis to choose among them since ignorability is untestable \citep{Holland1986,Pearl2009,ImbensRubin2015}. Different covariate adjustments can yield conflicting estimates---a recurring theme in multiverse and stability analyses \citep{Steegen2016,Yu2013stability,Yu2020veridical}.

A prominent example is the debate on digital media and adolescent well-being: In a large re-analysis of more than 355,000 adolescents, \citet{orben2019association} report estimates ranging from harmful to near-zero depending on analytic choices, such as covariate adjustment, exposure and outcome definitions, and missing-data handling. One source of disagreement, for instance, concerns the causal role of the covariates grades and enjoyment of school, whose temporal ordering relative to the exposure is unclear. Different assumptions about their causal role lead to different adjustment sets, which produce conflicting estimates.

When ancestral relationships are known, the disjunctive cause criterion of \citet{VanderWeeleShpitser2011} resolves this ambiguity without requiring the full causal graph. In many applications, however, these ancestral roles can be uncertain, and which adjustment sets satisfy ignorability remains an open question. 
Our work addresses this complementary setting of genuine specification uncertainty, where practitioners have identified multiple candidate adjustment sets but cannot determine which lead to valid inference.

\begin{table}[t]
\centering\renewcommand{\arraystretch}{1}
\begin{tabular}{@{}lcccccccc@{}}
\toprule
\multirow{2}{*}{}  & \multicolumn{2}{c}{Adjust for $\{X_1\}$} & \multicolumn{2}{c}{Adjust for $\{X_1,X_2\}$} & \multicolumn{2}{c}{{Convex hull}} & \multicolumn{2}{c}{{Proposed method}}\\ 
\cmidrule(lr){2-3} \cmidrule(lr){4-5} \cmidrule(lr){6-7} \cmidrule(lr){8-9}
\textbf{ } & {coverage} & {width} & {coverage} & {width} & {coverage} & {width}& {coverage} & {width} \\ \midrule
\cref{example-1} & $\bad{0.0\%}$  & $0.645$ & $\good{95.8\%}$ & $0.391$ & $\good{98.0\%}$ & $2.851$ & $\good{95.0\%}$ & $\better{1.246}$ \\
\cref{example-2} & $\good{94.5\%}$ & $1.121$  & $\bad{0.0\%}$ & $0.722$  & $\good{96.7\%}$ & $3.086$ & $\good{95.3\%}$ & $\better{1.271}$ \\
\bottomrule
\end{tabular}
\caption{
Empirical coverage and average width of 95\% CIs 
in \cref{example-1,example-2}, based on linear regression with  $\{X_1\}$ and $\{X_1, X_2\}$ as the adjustment sets; see~\cref{sec:simulations} for details.
The first two columns illustrate that the choice of the adjustment set can critically affect coverage. The last two columns correspond to the convex hull of the two candidate CIs and the proposed specification-robust CI, respectively; both of which achieve nominal coverage for their respective targets. The proposed method targets the ATE of a reweighted population shown in \cref{fig:hist_rewt1}.
}
\label{tab:siml_results}
\end{table}

\begin{figure}[t]
\centering
\resizebox{.7\textwidth}{!}{
\begin{tikzpicture}[
    baseline={(X2.base)},
    ->,
    >=Stealth,
    node distance=1cm,
    every node/.style={
        draw,
        circle,
        minimum size=7mm,
        inner sep=1pt
    },
    thick
]
\node (X1) {$X_1$};
\node (X2) [below=of X1] {$X_2$};
\node (A) [right=of X1] {$A$};
\node (Y) [below=of A] {$Y$};
\draw (X1) -- (A);
\draw (X2) -- (A);
\draw (X1) -- (Y);
\draw (X2) -- (Y);
\draw (A) -- (Y);
\end{tikzpicture}
\centering
\hspace{2cm}
\centering
\begin{tikzpicture}[
    baseline={(Y.base)},
    ->,
    >=Stealth,
    node distance=1cm,
    every node/.style={
        draw,
        circle,
        minimum size=7mm,
        inner sep=1pt
    },
    latent/.style={dashed},
    thick
]
\node[latent] (U1) {$U_1$};
\node[latent] (U2) [below=of U1] {$U_2$};
\node (A) [right=of U1, xshift=2cm] {$A$};
\node (X1) [right=of A] {$X_1$};
\node (Y) [below=of A] {$Y$};

\coordinate (midL) at ($(U1)!0.5!(U2)$);
\coordinate (midR) at ($(A)!0.5!(Y)$);
\node (X2) at ($(midL)!0.5!(midR)$) {$X_2$};

\draw (U1) [out=10, in=170] to (A);
\draw (U1) -- (X2);
\draw (U2) -- (X2);
\draw (U2) [out=-10, in=-170] to (Y);
\draw (A) -- (Y);
\draw (X1) -- (A);
\draw (X1) -- (Y);
\end{tikzpicture}
}
\caption{Illustration of the data-generating mechanisms in \cref{example-1} (left) and \cref{example-2} (right). In \cref{example-1}, both $X_1$ and $X_2$ are confounders;  adjusting only for $\{X_1\}$ leaves residual confounding from $X_2$. In \cref{example-2}, $X_2$ is a pre-treatment collider (influenced by the latent variables $U_1$ and $U_2$). Adjusting for $\{X_1\}$ is valid, but adjusting for $\{X_1, X_2\}$ induces $M$-bias by opening the backdoor path $A\leftarrow U_1 \rightarrow X_2 \leftarrow U_2 \rightarrow Y$.}
\label{fig:three-dags}
\end{figure}

The existing remedies have limitations:~Adjusting for the union or intersection of candidate adjustment sets can produce biased estimates (\cref{tab:siml_results}). Averaging estimates, as in \citet{orben2019association}, conflates substantive disagreement with the absence of an effect. Another naive approach is to adjust for all the pre-treatment covariates; but this is also known to be problematic in certain settings \citep{DingMiratrix2015}; see also \cref{example-2}, where $X_2$ is a pre-treatment collider. Recent work by \citet{yang2024} proposes robust testing of a causal null hypothesis when at least one candidate causal model is valid. If one extends their construction to test \(H_0:\tau=t\) and inverts these tests, the resulting confidence set asymptotically retains all the candidate estimands, and hence generally does not shrink to a point when those estimands disagree. Similarly, the convex hull or the union of the confidence intervals maintains nominal coverage when at least one of the candidate adjustments is valid, but its width generally does not vanish as the sample size grows, leaving the ATE partially identified.

When the challenge is overlap rather than ignorability, the field already accepts a pragmatic solution: Redefine the target population to where inference is credible and precise.
For instance, trimming removes units with poor overlap \citep{Crump2009}; overlap weighting downweights such units, tilting the estimand toward regions of better comparability
\citep{Li2018}. This strategy also
 applies more broadly:~Matching focuses on comparable units \citep{Rubin1973,Rosenbaum1983}; the local average treatment effect (LATE) restricts attention to compliers \citep{ImbensAngrist1994}. In all of these cases, researchers accept a modified estimand in exchange for identification and precision. This approach of modifying the estimand, while popular in empirical practice, requires care in transparent reporting of the new target population and interpreting the modified estimand for that population \citep{rizk2025beautiful}.

We apply the same principle to handle  uncertainty in regression adjustment. Point identification on the original population requires resolving which adjustment set satisfies ignorability, which is often impossible. We shift the target to a reweighted population as close as possible to the original population (in KL-divergence) for which all candidate estimands agree; see \cref{fig:hist_rewt2} for an illustration. This reweighting, driven entirely by the measured covariates, makes the population shift transparent and enables valid, precise inference.

\begin{figure}[t]
    \centering
    \includegraphics[width=\linewidth]{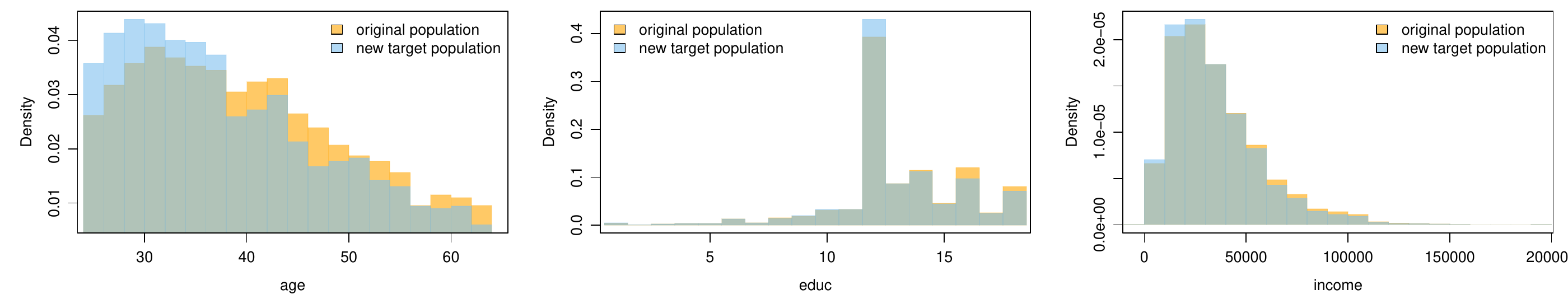}
    \caption{Distribution of the covariates age, education and income for the original population versus the new target population for the 401(k) dataset; see \cref{sec:401k} for details. We shift the target to a reweighted population chosen as close as possible to the original population such that all candidate covariate adjustments identify the ATE of the reweighted population, provided at least one of them satisfies ignorability.
    }
    \label{fig:hist_rewt2}
\end{figure}

\paragraph*{Our contributions}

\begin{itemize}
    \item  We develop a specification-robust procedure that returns a single point estimate and a confidence interval that is asymptotically valid as long as at least one candidate adjustment set satisfies ignorability, without requiring the knowledge of which one(s).

    \item The width of our confidence interval vanishes at the parametric $n^{-1/2}$ rate, unlike the convex hull interval whose width remains of constant order.

    \item We propose an extension to protect functions of the covariates used for reweighting, similar to adding calipers in matching \citep{Rosenbaum-Rubin-matching}. This ensures that the reweighting preserves key features of the original population deemed important by the researcher.
\end{itemize}

  Additionally, we recommend diagnostic plots such as \cref{fig:hist_rewt2} that compare the reweighted population with the original in terms of the distribution of the covariates, to assess whether the new target population is of scientific interest. In an empirical application studying the effect of 401(k) availability on net financial assets (\cref{sec:401k}), our proposed method produces a confidence interval  $\sim 77\%$ narrower than the convex hull interval (\cref{401(k)-example-m}), and targets the ATE of the reweighted population illustrated in \cref{fig:hist_rewt2}. The accompanying R package and replication files are available at: \href{https://github.com/ghoshadi/specrobust}{https://github.com/ghoshadi/specrobust}.

\subsection{Related work}

\paragraph*{Multiverse analysis and stability}
Multiverse analysis \citep{Steegen2016} and stability analysis \citep{Yu2013stability,Yu2020veridical} both address sensitivity to analytic choices---the former by reporting all plausible analyses, the latter by assessing whether conclusions are robust to appropriate perturbations. These principles have been influential across Bayesian statistics \citep{Box1980,Skene1986}, econometrics \citep{Leamer1983}, and causal inference \citep{Lalonde1986,Rosenbaum1987,Freedman1991statistical,Rosenbaum2010evidence,ImbensRubin2015,Karmakar2019integrating}. We address a similar multiplicity problem, but offer a different solution:~Rather than reporting the range of estimates across plausible adjustment sets, we collapse it into a single credible estimate.~\citet{Kang2016} study a parallel problem for instrumental variables---identifying the causal effect when a majority of candidate instruments are valid.

\paragraph*{Causal graphs} There are systematic approaches for selecting valid adjustment sets when the causal graph is known. For example, Pearl's backdoor criterion \citep{Pearl2009} is a sufficient condition, which \citet{ShpitserVanderWeeleRobins} generalize to a wider class. \citet{VanderWeeleShpitser2011} introduce the {disjunctive cause criterion}, which selects adjustment sets from ancestral information without requiring the complete graph structure.  \citet{guo2023confounderselectionobjectivesapproaches} provide an overview of such criteria, and \citet{GuoEmaRotnitzky2023} study efficiency when multiple {valid} sets are available. Another approach is to learn the graph via causal discovery and conduct post-selection inference \citep{chang2025postselectioninferencecausaleffects}. These tools can help guide adjustment choices, but uncertainty can remain when multiple plausible graphs are consistent with available knowledge.  \citet{bhattacharya2026robust} consider settings where additional causal assumptions yield testable restrictions. When at least one valid candidate satisfies these restrictions and the invalid ones do not, they propose using a weighted combination of the candidate estimators. 
Our work addresses the complementary case where specification uncertainty remains unresolved either by testable implications or by domain knowledge. Given a set of candidate adjustments---whether from graphical reasoning, causal discovery, or prior literature---our method aims to provide valid inference for a reweighted population when at least one candidate adjustment set is valid, without requiring us to identify which one(s).

\paragraph*{Sensitivity analysis} 
Sensitivity analysis methods assess the robustness of causal estimates to ignorability violations, typically through parameters that quantify the degree of unmeasured confounding. A variety of such methods have been introduced, including \citet{cornfield1959smoking,Rosenbaum1983sensitivity,Robins1999,Kenneth2000,Rosenbaum2002,Imbens2003sensitivity,brumback2004sensitivity,Imai2010,Hosman2010,vanderweele2011,Blackwell2014,dorie2016flexible,CinelliHazlett2020,Oster2019,Zhao2019,Franks2020,DornGuo2023,kang2024identification,DornGuoKallus2025}, among others. All of these methods take one adjustment set as given and probe how far ignorability can be violated before conclusions change. We address a different challenge:~Rather than perturbing a single specification, we tackle uncertainty about which adjustment set is valid in the first place.

\paragraph*{Reweighting methods}
Reweighting methods have been a cornerstone of causal inference, from inverse probability weighting \citep{HorvitzThompson1952} and its stabilizations via outcome modeling \citep{Robins1994}, calibration \citep{DevilleSarndal} and entropy balancing \citep{Hainmueller2012}, to overlap weighting \citep{Li2018} and weighted instrumental variable methods \citep{Small_et_al_2017,rakshit2024localeffectscontinuousinstruments}. While these methods reweight to improve estimation under a single, assumed-valid adjustment set, our approach reweights to reconcile estimation across multiple adjustment sets, assuming at least one is valid.

\section{Point identification via population reweighting}\label{new-sec:pop-level}

\subsection{Problem setup and notation}\label{sec:notation}

Suppose that we observe data $D_i=(Y_i,\, A_i,\, X_i)$, $i=1,\dots,n$, where $Y_i$ is the response, $A_i\in\{0,1\}$ is the binary treatment indicator, and $X_i\in\R^p$ are the covariates for the $i$-th unit. Denote by $Y_i(0)$ and $Y_i(1)$ the potential outcomes under control and treatment, respectively \citep{Neyman1923,Rubin1974,Rubin1977,ImbensRubin2015}, and assume that SUTVA holds, so that $Y_i=Y_i(A_i)$.  
\begin{assumption}\label{assump-iid}
$\{(Y_i(0),\,Y_i(1),\,\trt_i,\, X_i)\}_{1\le i\le n}$ are independent observations from an unknown distribution $P$. Denote by $(Y(0),\, Y(1),\, A,\, X)$ a generic observation from $P$.
\end{assumption}

Our parameter of interest is the average treatment effect (ATE) defined as
$$\overline\tau:=\E_P[Y(1)-Y(0)].$$ 
For any subset $S\subseteq \{1,2,\dots,p\}$ we use $X_S$ to denote the covariates $\{X_{s}:s\in S\}$.

\begin{definition}\label{def:valid-adj-set}
    Call an adjustment set $S$ \emph{valid} (under $P$) if it satisfies the ignorability condition:~The treatment is independent of the potential outcomes when we condition on the covariates in that set, i.e.,
    $${A\indep Y(a)\, \mid \, X_S,\quad a=0,1}.$$
We say that an adjustment set $S$ is \emph{invalid} if it does not satisfy the above. 
\end{definition}

An invalid adjustment set can compromise inference for two broad reasons. First, it can leave out an important confounder (\cref{example-1}). Second, including certain variables can induce conditional dependence between the treatment and potential outcomes (\cref{example-2}). Given a causal graph, \citet{ShpitserVanderWeeleRobins} provide 
a necessary and sufficient graphical criterion for the validity of a covariate adjustment; see \citet{Pearl2009} for a book-level treatment of graphical causal inference.
Our focus is on settings where such a specification is either unavailable or contested. Concretely, suppose that we have $K$ candidate adjustment sets $S_1,\dots,S_K$ that satisfy the following condition.

\begin{assumption}\label{assump-at-least-one-is-valid}
    At least one of the candidate adjustment sets is valid (we do not know which one(s)), and their intersection $S_{\cap}$ is non-empty. Assume also that each adjustment set $S_k$ satisfies strong overlap:~$\xi\le P(A=1\mid X_{S_k})\le 1-\xi$ almost surely, for some $0<\xi<1$. 
\end{assumption}

Even when ignorability is violated, many off-the-shelf CATE estimators can still be used to estimate the conditional mean contrasts:
\begin{equation}\label{def:contrasts}
 \tau(X_{S_k}):=\E_P[Y\mid \trt=1, X_{S_k}]-\E_P[Y\mid \trt=0,X_{S_k}].
\end{equation}
When the adjustment set $S_k$ is {valid}, the regression contrast $\smash{\tau(X_{S_k})}$ coincides with the conditional average treatment effect (CATE) and its expectation coincides with the ATE $\overline{\tau}$. However, for an invalid adjustment set $S_k$, the regression contrast $\smash{\tau(X_{S_k})}$ does not have a causal interpretation, and its expectation can differ arbitrarily from $\overline{\tau}$.

As discussed in \cref{sec:intro}, the naive approach of reporting the convex hull of the confidence intervals constructed using all candidate adjustment sets provides a valid confidence interval for $\overline{\tau}$ under \cref{assump-iid,assump-at-least-one-is-valid}. However, even when the sample size grows to infinity, this interval shrinks to the convex hull of the points $\{\E_P[\tau(X_{S_k})]\}_{k=1}^K$ which is of constant order. 
Indeed, the specification uncertainty not resolved by the observed data law leaves the ATE only partially identified whenever the candidate estimands disagree; see \cref{lem:minimax-reweighted-ATE} for a formal statement.

\subsection{Our specification-robust approach}\label{sec:intro-AR}

\begin{figure}[t]
    \centering
    \includegraphics[width=0.45\linewidth]{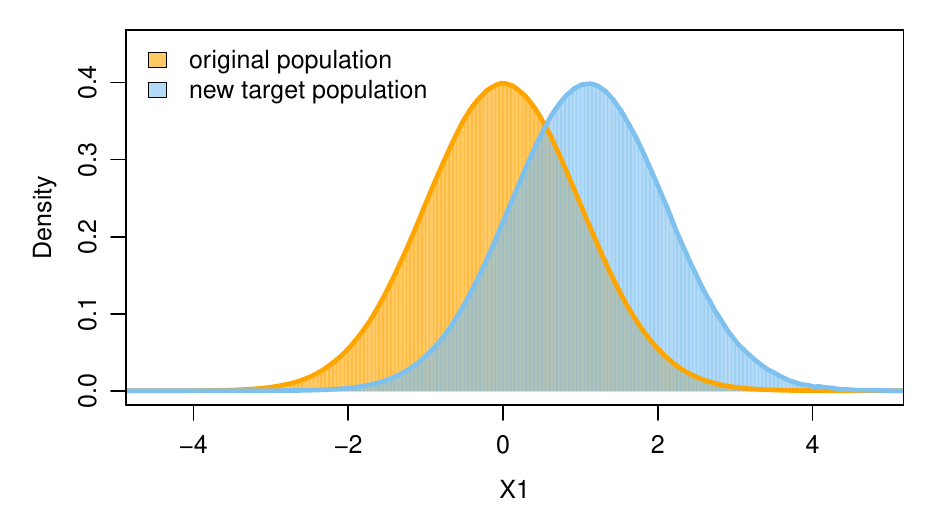}\hspace{1mm}
    \includegraphics[width=0.45\linewidth]{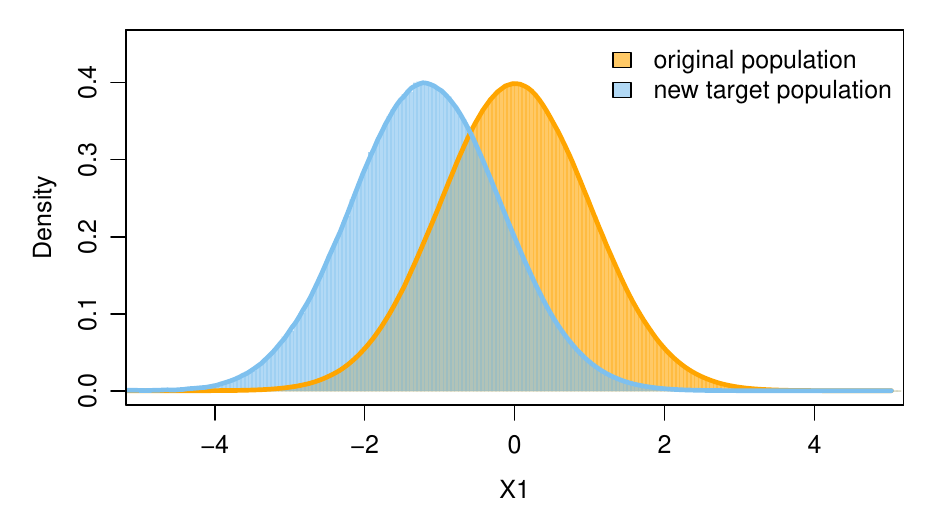}
    \caption{Distributions of the covariate $X_1$ under the original population ($P$) and the new target population ($w^*P$) for \cref{example-1} (left) and \cref{example-2} (right). We find the reweighted population as close as possible to the original population such that all the candidate adjustments identify the ATE of the population $w^* P$. In this example, $w^*$ is only a function of $X_1$, and $X_2$, being independent of $X_1$, retains the same distribution under the reweighting.}
    \label{fig:hist_rewt1} 
\end{figure}

We shift the target to a new population by reweighting the covariates while keeping the conditional distribution of potential outcomes and treatment  unchanged, i.e., consider reweighted populations of the form 
\begin{equation}
    \label{eq:w}
    wP =wP_X \cdot P_{Y(1),\,Y(0),\, A\,\mid\, X},\quad w\ge 0,\quad  \E_P[w]=1.
\end{equation}
Here $w(\cdot)$ denotes the likelihood ratio of the new target population with respect to the original. 
Our strategy is to find the population closest in Kullback--Leibler (KL) divergence\footnote{All twice continuously differentiable $f$-divergences share the same second-order Taylor expansion around unity (up to normalization), so for small perturbations the choice of divergence is inconsequential. We use the KL divergence since it yields tractable weight expressions that facilitate the bias-correction needed for parametric rate (\cref{sec:theory}).} to the original population for which all candidate adjustment sets target the same estimand. Formally, consider the optimization problem:
\begin{equation}\label{adj-sets-popln-opt}
\begin{split}    w^*=&\argmin_{w=w(X_{S_\cap})\ge 0,\,\E_P[w]=1}\dkl(wP\,\|\, P) \\ &\text{subject to}\ \ \E_P[w(X_{S_\cap})\tau(X_{S_1})]=\cdots=\E_P[w(X_{S_\cap})\tau(X_{S_K})].
\end{split}
\end{equation}
Throughout, we restrict to weights $w$ satisfying $\E_P[w(X_{S_\cap})|\tau(X_{S_k})|] < \infty$ for every $k$, so that the constraints in \eqref{adj-sets-popln-opt} are well defined. Two key design choices in \eqref{adj-sets-popln-opt} are that we restrict to weights that only depend on the covariates shared by all candidate adjustment sets, and we match the mean of the contrasts exactly on the reweighted population. We motivate these in \cref{sec:motivation}.
By construction, all candidate adjustment sets target a common estimand for the reweighted population:
\begin{equation}\label{eq:def-tauR}
    \tauR := \E_{w^*P}[\tau(X_{S_1})] =\E_{w^*P}[\tau(X_{S_2})]=\cdots=\E_{w^*P}[\tau(X_{S_K})].
\end{equation}
Unlike the original contrasts $\tau(X_{S_k})$ lacking a causal interpretation when $S_k$ is not valid, the
new target parameter $\tauR$ defined in \eqref{eq:def-tauR} has a causal interpretation under \cref{assump-at-least-one-is-valid}:
\begin{lemma}[Interpretation of the reweighted estimand]\label{propo:reweighted-estimand-makes-sense}
    Suppose that the adjustment sets $S_1,\dots,S_K$ satisfy \cref{assump-at-least-one-is-valid} and $w^*$ solves the optimization problem \eqref{adj-sets-popln-opt}. Then, for each $k=1,2,\dots,K$,  
    $$\tauR=\E_P[w^*(X_{S_\cap})\tau(X_{S_k})]=\E_{w^*P}[Y(1)-Y(0)].$$
\end{lemma}

\cref{propo:reweighted-estimand-makes-sense} emphasizes that by solving \eqref{adj-sets-popln-opt} we not only find a reweighted population close to the original population for which all candidate adjustments mechanically target the same estimand, but the common estimand also coincides with the ATE of the new target population. 
We refer the reader to \cref{fig:hist_rewt2,fig:hist_rewt1} for a visual comparison of the reweighted population and the original population for our empirical examples; see \cref{sec:applications} for details.

In the following section, we develop the supporting theory, from solving \eqref{adj-sets-popln-opt} empirically to constructing specification-robust confidence intervals using debiased machine learning. We relegate additional technical results to \cref{sec:dump}, and further elaborate on the interpretation and practical usefulness of the reweighted population for real-world applications in \cref{sec:practical}.

\section{Theory for estimation and inference}\label{sec:theory}

Given candidate adjustment sets $S_1,\dots,S_K$  satisfying \cref{assump-at-least-one-is-valid}, our goal is to estimate and draw inference on the ATE $\tauR$ of the reweighted population introduced in \cref{sec:intro-AR}. We proceed in two steps. First, in \cref{sec:transfer-weights-theory}, we characterize the solution to the population-level problem \eqref{adj-sets-popln-opt} and use this to lay out the construction of empirical transfer weights.
Then, in \cref{sec:est-and-infer}, we construct an estimator of $\tauR$
that aggregates the augmented inverse probability weighting (AIPW) estimators based on each adjustment set, together with an additional bias-correction to account for the uncertainty in estimating the transfer weights. The proposed estimator is asymptotically normal at the $n^{-1/2}$ rate (\cref{thm:matching.adj.sets}), using which we obtain specification-robust confidence intervals (\cref{cor:CI-from-Theorem-1}).

\subsection{Transfer weights via exponential tilting}\label{sec:transfer-weights-theory}

Solving the optimization problem \eqref{adj-sets-popln-opt} might initially appear as  challenging since it involves optimization over an infinite-dimensional space of reweighting functions. However, recognizing \eqref{adj-sets-popln-opt} as an I-projection onto a moment-constrained set, we can apply classical results from information theory \citep{Csiszar1975,CoverThomas2005MaximumEntropy} to reduce it to a finite dimensional convex optimization problem via exponential tilting. We state the relevant result tailored to our setting in \cref{propo:exist-n-unique}, which relies on the following sufficient conditions for the existence and uniqueness of a solution.

\begin{assumption}\label{assump:heterogeneity}
    Define differences-in-contrasts\footnote{The choice of $1$ as the reference index does not affect the reweighted population or the resulting estimand, since the constraint set $\{\E_{wP}[\tau(X_{S_\ell})-\tau(X_{S_k})]=0,\, 1\le k\le K\}$ spans the same affine subspace regardless of which index $\ell$ is designated as the reference.
    } $\Delta\tau_k(X):=\tau(X_{S_1})-\tau(X_{S_k})$ where $\tau(X_{S_k})$ is defined in \eqref{def:contrasts}. Write $\Delta\tau:=(\Delta\tau_2,\dots,\Delta\tau_K)^\top$, and define
    \begin{equation}
        \label{first-def-of-g}
    g(X_{S_\cap}):=\E_P[\Delta\tau(X)\mid X_{S_\cap}],\qquad \psi_P(\lambda):= \E_P [\exp(g(X_{S_\cap})^\top \lambda)] \, \in (0,\infty].
    \end{equation} Assume that  
    (i)\footnote{In fact, condition (i) is sufficient as well as necessary for the moment generating function $\psi_P(\cdot)$ to have a finite and unique minimizer; see \cref{app:proof-of-propo:exist-n-unique} for a proof.} there is no $v\neq 0$ such that $v^\top g(X_{S_\cap})\le 0$ almost surely under $P$, and (ii) the infimum of  $\lambda\mapsto\psi_P(\lambda)$ is not attained on the boundary of $\{\lambda:\psi_P(\lambda)<\infty\}$.
\end{assumption}

We note that \cref{assump:heterogeneity} is essentially a heterogeneity assumption on  how the contrasts vary with the shared covariates. In particular, it rules out cases where regardless of how we reweight based on the shared covariates, the contrasts cannot be matched as in \eqref{adj-sets-popln-opt} (e.g., when one or more of the differences-in-contrasts are always positive or always negative). We refer the reader to \cref{sec:feasibility} for the practical implications of infeasibility and possible remedies. The following result states how \cref{assump:heterogeneity} implies the existence and uniqueness of the solution to \eqref{adj-sets-popln-opt} and also reduces it to a finite-dimensional convex optimization problem.

\begin{lemma}[Existence and uniqueness of transfer weights]\label{propo:exist-n-unique}
Under \cref{assump:heterogeneity} (i), the moment generating function $\psi_P(\cdot)$ defined in \eqref{first-def-of-g} has a unique, finite minimizer $\lambda^*$, i.e., 
\begin{equation}\label{eq:def-lambda-star}
    \lambda^* :=\argmin_{\lambda\in \R^{K-1}}\ \psi_P(\lambda),
\end{equation}
 If, in addition, \cref{assump:heterogeneity} (ii) holds, then \eqref{adj-sets-popln-opt} has a unique ($P$-a.s.) solution given by the exponential tilt:
    $$w^*(X_{S_\cap})={\exp\big(g(X_{S_\cap})^\top\lambda^*\big)}\;/\;{\psi_P(\lambda^*)}.$$
\end{lemma}

\paragraph*{Estimating the transfer weights}
In light of \cref{propo:exist-n-unique}, we can empirically solve \eqref{adj-sets-popln-opt} using numerical convex optimization (see \citet{BoydVandenberghe2004} for a book-level treatment). Throughout, we use $\P_n [\psi(D)]$ to denote $n^{-1}\sum_{i=1}^n \psi(D_i)$, for any measurable function $\psi$ of a generic observation $D=(Y,A,X)$. We perform the following steps:
 \begin{enumerate}\itemsep0mm
     \item Use any off-the-shelf algorithm that estimates the regression contrasts $\tau(X_{S_k})$ regardless of whether ignorability holds (e.g., generalized random forests; \citealp{Wager2018,Athey2019grf}).

     \item Regress $\wh\tau(X_{S_k})$ further on the shared covariates $X_{S_\cap}$ and take their difference to obtain an estimator $\wh{g}(X_{S_\cap})$ of the function $g(X_{S_\cap})$ defined in \eqref{first-def-of-g}.

\item Use any standard convex optimization routine to numerically solve the following  problem:
\begin{equation}
    \label{eq:lambda-hat-def}
    \wh{\lambda}:=\argmin_{\lambda\in \R^{K-1}}\; \P_n \left[\exp\left(\wh{g}(X_{S_\cap})^\top \lambda\right)\right]=\argmin_{\lambda\in \R^{K-1}}\; \frac1n\sum_{i=1}^n \exp\left(\wh{g}(X_{i,S_\cap})^\top \lambda\right).
\end{equation}

\item Obtain transfer weights for each unit as $$\wh{w}(X_{i,S_\cap}):={\exp\big(\wh{g}(X_{i,{S_\cap}})^\top\, \wh{\lambda}\big)}\;\big/\;{\P_n \big[\exp\big(\wh{g}(X_{S_\cap})^\top\, \wh{\lambda}\big)\big]}.$$

 \end{enumerate}

For ease of exposition, we described above the main steps for estimating the transfer weights using the full sample. In practice, however, it is not desirable to evaluate the weight function on the same observations used to estimate it, since it can lead to overfitting. To avoid this, we split the sample into $L\ge 2$ folds and use cross-fitting (see, e.g., \citet{Chernozhukov2018DML}); \cref{algo:rewt-adj-sets} gives a complete description of our procedure. In the following section, we turn to estimation and inference for the ATE of the reweighted population.

\subsection{Specification-robust estimation and inference}
\label{sec:est-and-infer}

With the transfer weights $\wh{w}(X_{i,S_\cap})$ estimated as in \cref{sec:transfer-weights-theory}, we can estimate the ATE $\tauR$ of the reweighted population using any of the following estimators:
\begin{equation}
    \label{def-naive-rewt-est}
    \wh{\tau}_k^{\,\mathtt{R}} := \P_n \left[\wh{w}(X_{S_\cap})\,\wh\tau(X_{S_k})\right],\quad k=1,\dots,K.
\end{equation}
Since all of these estimators $\wh\tau^\mathtt{R}_k$ target the same estimand $\tauR$, we can aggregate them via an affine combination and report a single point estimate:\begin{equation}\label{AR-with-generic-nu}
    \wh\tau^{\nam}(\nu)=\sum_{k=1}^K\nu_k \wh\tau_k^\mathtt{R},\quad\text{where}\quad \sum_{k=1}^K\nu_k=1.
\end{equation}
In the special case where the baseline estimators $\wh{\tau}(X_{S_k})$ are based on linear regressions with treatment-covariate interactions (at least one correctly specified), the above estimator achieves the fast $\sqrt{n}$-rate of convergence for any choice of the affine weights $\nu$. We analyze this special case in \cref{sec:lm-based}.

More generally, modern machine learning methods offer more flexible ways to estimate the conditional contrasts.
The literature on semiparametric inference and double/debiased machine learning  \citep{Robins1994,Robins1995semiparametric,Hahn1998,Scharfstein1999adjusting,Chernozhukov2018DML,chernozhukov2021automatic} provides a general framework for using such flexible estimators, together with augmented inverse probability weighting (AIPW) that eliminates first-order sensitivity to nuisance estimation errors. 
With this motivation, we replace the generic nonparametric estimators $\wh\tau(X_{S_k})$ in \eqref{def-naive-rewt-est} with their AIPW versions. 
 The resulting reweighted estimators, however, still inherit a first-order error from estimating the transfer weights. This has two components, namely from estimating the function $g$ defined in \eqref{first-def-of-g} and the exponential tilt parameter  $\lambda^*$ defined in \eqref{eq:def-lambda-star}. We handle these sources of nuisance error as follows. First, we carefully choose the affine weights $\nu$  in \eqref{AR-with-generic-nu} that mitigate the first-order error from estimating $\lambda^*$. Then, we handle the remaining first-order error by adding a bias-correction term. \cref{algo:rewt-adj-sets} provides a more explicit description of this debiasing procedure, and \cref{propo:Neyman-orthogonality} provides further insights through Neyman orthogonality.

The following result states the asymptotic distribution of the bias-corrected estimator.
To keep the exposition simple, we assume here that the nonparametric  estimators taken as input in \cref{algo:rewt-adj-sets} are all fit externally, i.e., independent of the data they are evaluated on. In our software implementation we split the sample into $L\ge 2$ folds and use cross-fitting.
A standard conditioning argument then implies that the same asymptotic distribution holds for the cross-fitted estimator computed in \cref{algo:rewt-adj-sets}.

\begin{theorem}[Asymptotic normality of the specification-robust estimator]\label{thm:matching.adj.sets}
Assume that the propensity score estimators $\wh{e}(X_{S_k})$, conditional response estimators $\wh\mu_a(X_{S_k})$ and nonparametric regressors $\wh{\E}_P[\wh\tau(X_{S_k})\mid X_{S_\cap}]$ used in \cref{algo:rewt-adj-sets} are fit externally, i.e., independent of $(Y_i,A_i,X_i)_{i=1}^n$, and that \cref{assump-iid,assump-at-least-one-is-valid,assump:heterogeneity} hold true. Assume further that,
    \begin{enumerate}[label=(\roman*)]
    \itemsep0mm
    \item (rates of nonparametric regressors) for some exponent $q>2$ and for every $1\le k\le K$, 
   \begin{align*}
      & \|\wh{\mu}_a(X_{S_k})-\mu_a(X_{S_k})\|_{L^{q}(\P_n)}=\o(n^{-\alpha_\mu}), \quad 
      \|\wh{e}\,(X_{S_k})-e(X_{S_k})\|_{L^{q}(\P_n)}=\o(n^{-\alpha_e}),
    \end{align*}
   where $\alpha_\mu\ge 0$, $\alpha_e\ge 0$ and $\alpha_\mu+\alpha_e\ge{1}/{2}$, i.e., the product of the above is $\o(n^{-1/2})$. Moreover, 
   \begin{equation}
       \label{eq:fourth-root}
       \|\wh{\E}_P[\wh{\tau}(X_{S_k})\,|\, X_{S_\cap}]-\E_P[\tau(X_{S_k})\,|\, X_{S_\cap}]\|_{L^{q}(\P_n)}=\o(n^{-1/4}).
   \end{equation}

   \item (strong overlap) The estimated propensity scores $\wh{e}\,(X_{S_k})$ are bounded between $[\xi,1-\xi]$ for some $0<\xi<1$.
    \item (bounded moments) The moment generating functions of $\E_P[\tau(X_{S_1})\mid X_{S_\cap}]$ and $g(X_{S_\cap})$ are finite in a radius $\rc> \frac{2q}{q-2}\|\lambda^*\|$ around the origin, and the same holds for $\wh{g}$ conditionally:~$\smash{\E_P[\exp(r\|\wh{g}(X_{S_\cap})\|)\mid \wh{g}(\cdot)]=\O(1)}$. Also, $\var_P[Y\mid A=a,X_{S_\cap}]\le C$ and $\var_P[Y\mid A=a,X_{S_k}]\le C$ for some constant $C\in (0,\infty)$.
    \end{enumerate}  Then, the specification-robust estimator $\wh\tau^\nam_{bc}$ we compute in \cref{algo:rewt-adj-sets} satisfies
$$\sqrt{n}\left(\wh\tau^\nam_{bc}-\tauR\right)\dto \normal(0, \mathrm{V}_\nam).$$
An explicit expression for the asymptotic variance $\mathrm{V}_\nam$ is deferred to \cref{sec:neyman-orthogonality}.
\end{theorem}

\begin{algorithm}[p]
    \caption{\;Bias-corrected specification-robust estimator (with cross-fitting)}
    \label{algo:rewt-adj-sets} 
\small
    \begin{algorithmic}[1] \itemsep2mm
    \vspace{1mm}
    \Require Data $(Y_i,\trt_i,X_i)_{i=1}^n$; $K\ge 2$ adjustment sets $S_1,\dots,S_K$ with non-empty intersection $S_\cap$; black-box machine learning algorithms to estimate the propensity scores $\smash{e(X_{S_k})=P(\trt=1\mid X_{S_k})}$, conditional response functions $\smash{\mu_{a}(X_{S_k})={\E}_P\left[Y\mid \trt=a,X_{S_k}\right]}$, for $a=0,1$, $k=1,\dots,K$, and a regression algorithm for $\smash{\E_P[\,\cdot \mid X_{S_\cap}]}$; number of folds $L\ge 2$ for cross-fitting.

    \Ensure Estimated transfer weights (out-of-fold) and the proposed specification-robust estimator $\wh\tau^\nam_{bc}$ for the ATE of the reweighted population.

\State Split $[n]=\{1,\dots,n\}$ into $L$ disjoint, nearly equal-sized folds $I_1,\dots,I_L$, and write $I_{-\ell}:=[n]\setminus I_\ell$.

\For{$\ell=1,\dots,L$}

\State Using data from $I_{-\ell}$, train estimators $\wh{\mu}_a^{-\ell}(X_{S_k})$ of $\mu_a(X_{S_k})$ and $\wh{e}^{-\ell}(X_{S_k})$ of $e(X_{S_k})$, for $a=0,1$ and $k=1,\dots,K$. Define the AIPW functional
\begin{multline*}
        \haipw{k,\,-\ell}:=\wh{\mu}_1^{-\ell}(X_{S_k})-\wh{\mu}_0^{-\ell}(X_{S_k})+\frac{\trt}{\wh{e}^{-\ell}(X_{S_k})}\big(Y-\wh\mu_{1}^{-\ell}(X_{S_k})\big)-\frac{1-\trt}{1-\wh{e}^{-\ell}(X_{S_k})}\big(Y-\wh\mu_{0}^{-\ell}(X_{S_k})\big).
\end{multline*}

\State Using data from $I_{-\ell}$, regress the out-of-bag fitted (or, internally cross-fitted) values of $\smash{\wh{\mu}_1^{-\ell}(X_{S_k})}$ and  $\smash{\wh{\mu}_0^{-\ell}(X_{S_k})}$ on $X_{S_\cap}$, and take the difference of the two fits to obtain $\wh{m}_k^{-\ell}(X_{S_\cap})$ estimating the conditional contrast ${m_k(X_{S_\cap})=\E_P[\tau(X_{S_k})\mid X_{S_\cap}]}$.

\State With $\smash{\wh{g}_k^{-\ell}:=\wh{m}_1^{-\ell}-\wh{m}_k^{-\ell}}$ for each $k=2,\dots,K$, set $\smash{\wh{g}^{-\ell}:=(\wh{g}_2^{-\ell},\dots,\wh{g}_K^{-\ell})^\top}$.

\State Use numerical convex optimization to solve the problem \eqref{eq:lambda-hat-def} on $I_{-\ell}$, with $\wh{g}^{-\ell}(X_{i,S_\cap})$ given by out-of-bag (or, internally cross-fitted) predictions for $i\in I_{-\ell}$, and obtain
\[\wh{\lambda}^{-\ell}=\argmin_{\lambda\in\R^{K-1}}\frac{1}{|I_{-\ell}|}\sum_{i\in I_{-\ell}}\exp\left(\lambda^\top\wh{g}^{-\ell}(X_{i,S_\cap})\right).\]

\State Denoting by $\P_{n,\ell}$ the empirical mean over $I_\ell$, compute transfer weights $$\wh{w}^{-\ell}(X_{S_\cap})= \exp\big(\wh{g}^{-\ell}(X_{S_\cap})^\top \wh{\lambda}^{-\ell}\big)\;\big/\;\P_{n,\ell}\big[\exp\big(\wh{g}^{-\ell}(X_{S_\cap})^\top \wh{\lambda}^{-\ell}\big)\big],$$ normalized on $I_\ell$ so that $\P_{n,\ell}\,[\wh{w}^{-\ell}]=1$.

\State Evaluate the nuisance estimators and the weighting function (all learned on $I_{-\ell}$) on $I_\ell$ to compute the reweighted (AIPW) estimators $$\wh{\tau}_{k,\,\ell}^{\,\mathtt{R}} := \P_{n,\ell} \left[\wh{w}^{-\ell}(X_{S_\cap})\,\haipw{k,\,-\ell}\right],\qquad k=1,\dots,K.$$

\State Compute affine weights $\wh{\nu}_\ell$ as follows. Set $\wh\nu_{1,\,\ell} := 1-\sum_{k=2}^K \wh\nu_{k,\,\ell}$ and 
\begin{multline*}
    \wh\nu_{2:K,\,\ell} := \left(\P_{n,\ell}\left[ \wh{w}^{-\ell}(X_{S_\cap})\wh{g}^{-\ell}(X_{S_\cap})\wh{g}^{-\ell}(X_{S_\cap})^\top\right]\right)^{-1} \P_{n,\ell}\left[ \wh{w}^{-\ell}(X_{S_\cap})\wh{g}^{-\ell}(X_{S_\cap})\left(\wh{m}_{1}^{-\ell}(X_{S_\cap})-\wh\tau_{1,\,\ell}^\mathtt{R} \right)\right].
\end{multline*}

\State With $\Delta\haipw{k,\,-\ell}:=\haipw{1,\,-\ell}-\haipw{k,\,-\ell}$, compute a bias-correction term:
\[
\wh{B}_{\ell}:=\P_{n,\ell}\left[\wh{w}^{-\ell}(X_{S_\cap})\left(\Delta\haipw{2:K,\,-\ell}-\wh{g}^{-\ell}(X_{S_\cap})\right)^{\!\top}\wh{\lambda}^{-\ell}\left(\sum_{k=1}^K \wh\nu_{k,\,\ell}\, \wh{m}_{k}^{-\ell}(X_{S_\cap})-\wh\tau_{1,\,\ell}^\mathtt{R} \right)\right].
\]
\EndFor

\vspace{-2mm}

\State Aggregate over the folds to compute the bias-corrected specification-robust estimator as $$\wh\tau^\nam_{bc}=\sum_{\ell=1}^L\frac{|I_\ell|}{n}\left(\sum_{k=1}^K\wh{\nu}_{k,\,\ell}\, \wh\tau_{k,\,\ell}^\mathtt{R}+\wh{B}_{\ell}\right).$$

\vspace{0mm}

\end{algorithmic}
\end{algorithm}

\begin{remark}\label{rem:fourth-root}
    For the outcome regressions and propensity score estimators, \cref{thm:matching.adj.sets} only requires consistency and the usual product-rate condition. However, for the iterated regressions (the regression contrasts further regressed on the shared covariates), we impose the fourth-root rate as in \eqref{eq:fourth-root}, because the remainder in the first-order Taylor expansion of the transfer weights is quadratic in the estimation error of those iterated regressions. 
   Fourth-root conditions of this form are common in the semiparametric inference literature when the relevant remainder is quadratic in the nuisance estimation error; see, e.g., \citet{newey1994,chen2003,Chernozhukov2018DML,nie2021}. 
   \end{remark}

\paragraph*{Inference}
For any consistent estimator $\smash{\wh{\mathrm{V}}_\nam}$ of the asymptotic variance $\mathrm{V}_\nam$ in \cref{thm:matching.adj.sets}, we can construct an asymptotically valid $(1-\alpha)$-confidence interval for $\tauR$ as follows:
\begin{equation}
    \label{eq:CI-from-AR}
    \mathcal{I}_\alpha:=\left[\wh\tau^\nam_{bc}-\frac{z_{\alpha/2}}{\sqrt{n}}\sqrt{\wh{\mathrm{V}}_\nam},\ \wh\tau^\nam_{bc}+\frac{z_{\alpha/2}}{\sqrt{n}}\sqrt{\wh{\mathrm{V}}_\nam}\,\right],
\end{equation}
where $z_\alpha$ is the $(1-\alpha)$-th quantile of standard normal distribution. We propose a plug-in estimator in \cref{sec:nonparametric-baseline}
and empirically demonstrate the resulting confidence interval in \cref{sec:401k}.

\begin{corollary}[Specification-robust confidence intervals]\label{cor:CI-from-Theorem-1}
    For any consistent estimator $\wh{\mathrm{V}}_\nam$ of the asymptotic variance $\mathrm{V}_\nam$ in \cref{thm:matching.adj.sets}, the $(1-\alpha)$ confidence interval $\mathcal{I}_\alpha$ defined in \eqref{eq:CI-from-AR} is (pointwise) asymptotically valid:
       $\lim_{n\to\infty} P(\tauR\in \mathcal{I}_\alpha)=1-\alpha.$
\end{corollary}

\paragraph*{Uniform validity} 
The pointwise validity in \cref{cor:CI-from-Theorem-1} already carries some robustness to the unknown data-generating process:~The proposed procedure is valid over the specification uncertainty class $\overline\CC(P)$ defined in \cref{sec:motivation}. We show in \cref{sec:uniform-validity} that, under uniform analogues of \cref{assump-iid,assump-at-least-one-is-valid,assump:heterogeneity}, this can be strengthened to uniform validity over a class $\cal{P}$, in the sense that $\lim_{n\to\infty}\sup_{P\in\sP}|P(\tauR\in \mathcal{I}_\alpha)-(1-\alpha)|= 0.$

\section{Additional technical results}\label{sec:dump}

\subsection{Motivation for restricting to shared covariates and exact agreement}\label{sec:motivation} One rationale behind restricting the transfer weights $w$ to depend only on the shared covariates is that it respects \cref{assump-at-least-one-is-valid}. The following lemma formalizes this, and also shows that solving the optimization problem \eqref{adj-sets-popln-opt} for arbitrary weights might not lead to a causal estimand.

\begin{lemma}\label{lem:why-shared-covs-only}
    When $w=w(X_{S_\cap})$, any candidate adjustment set that is valid under $P$ is also valid under $wP$. On the other hand, there exist examples where an arbitrary weighting $w=w(X)$ yields no candidate adjustment set valid for the reweighted population $wP$, and solving the optimization problem \eqref{adj-sets-popln-opt} for arbitrary weights $w=w(X)$ leads to a common estimand, as in \eqref{eq:def-tauR}, that has no causal interpretation.
\end{lemma}

Another key design choice in \eqref{adj-sets-popln-opt} is that we require the mean of the contrasts to exactly match on the reweighted population. To motivate this choice, we derive a lower bound on the minimax risk of estimating the ATE of the reweighted population $wP$ under the specification uncertainty not resolved by the observed data law under \cref{assump-iid,assump-at-least-one-is-valid}. To formalize this, define for each $1\le k\le K$ the collection\footnote{A simple construction shows that each $\CC_k(P)$ is non-empty:~Draw $(A,X)$ from their marginal under $P$. If $A=a$, draw $Y(a)$ from the observed conditional law of $Y$ given $(A=a,X)$, and independently draw $Y(1-a)$ from the observed conditional law of $Y$ given $\smash{(A=1-a,X_{S_k})}$. With $Y=Y(A)$, the resulting $(Y,A,X)$ has the same law as under $P$, and satisfies $(Y(0), Y(1))\indep A$ conditional on $\smash{X_{S_k}}$.
} $\CC_k(P)$ of all probability distributions of $(Y(0), Y(1), A, X)$ that induce the same observed data law as $P$, and imply ignorability for the adjustment set $S_k$. Precisely,
\[\CC_k(P):=\{Q: \CL_Q(Y(A), A, X)=\CL_P(Y(A), A, X)\text{ and }S_k \text{ is a valid adjustment under Q}\},\]
where $\CL_P(U)$ denotes the law of $U$ under $P$.
The union $\overline\CC(P):=\bigcup_{k=1}^K \CC_k(P)$ then collects all distributions of $(Y(0), Y(1), A, X)$ that induce the same observed data-law as $P$ and satisfy \cref{assump-at-least-one-is-valid}. The following result shows that the minimax risk over the specification-uncertainty class $\overline{\CC}(P)$ is bounded away from zero unless we exactly match the estimands.

\begin{lemma}\label{lem:minimax-reweighted-ATE}
Fix any $w=w(X_{S_\cap})\ge 0$, $\E_{P}[w]=1$. Under \cref{assump-at-least-one-is-valid}, it holds that
\[\inf_{a\in\R}\,\sup_{Q\in\overline\CC(P)}\left(a-\E_{wQ}[Y(1)-Y(0)]\right)^2 \ge \frac{1}{4}\left(\max_k \overline\tau_k(wP)-\min_k \overline\tau_k(wP)\right)^2,\]
where $\overline\tau_k(wP):=\E_{P}[w(X_{S_\cap})\tau(X_{S_k})]$, with $\tau(X_S)$ as defined in \eqref{def:contrasts}. 
Thus, the minimax risk is bounded away from zero unless $\max_k \overline\tau_k(wP)=\min_k\overline\tau_k(wP)$, in which case the minimax risk is zero and the ATE of the reweighted population is the unique minimizer.
\end{lemma}

In particular, \cref{lem:minimax-reweighted-ATE} implies that the ATE of a reweighted population defined in \eqref{eq:w} is point identified if and only if the reweighting makes all candidate estimands the same. Among all such populations, we find the one closest (in KL divergence) to the original population. 
    \subsection{Neyman orthogonality}\label{sec:neyman-orthogonality}
    Following the double/debiased machine learning literature \citep{Chernozhukov2018DML,chernozhukov2021automatic}, the debiasing in \cref{thm:matching.adj.sets} can be motivated through Neyman orthogonality of the influence function. Let $\eta:=(\mu, e, m_1, g, \lambda, \nu)$ collect all the nuisance parameters, namely: 
    \begin{enumerate}\itemsep0mm
        \item the conditional response functions $\mu_a(X_{S_k})=\E_P[Y\mid A=a, X_{S_k}]$,
        \item the propensity score functions $e(X_{S_k})=P(A=1\mid X_{S_k})$,
        \item the conditional contrasts $\E_P[\tau(X_{S_k})\mid X_{S_\cap}]$, reparametrized as $(m_1, g)$ where $g$ is as  in \eqref{first-def-of-g} and $m_1(X_{S_\cap}):=\E_P[\tau(X_{S_1})\mid X_{S_\cap}]$, 
        \item the exponential tilt parameter $\lambda$ that, together with $g$, defines the transfer weights as $$w_{g,\lambda}(X_{S_\cap})\propto \exp( g(X_{S_\cap})^\top \lambda),$$
        \item the aggregation weights $\nu$ summing to $1$, as used in \eqref{AR-with-generic-nu}.
    \end{enumerate}
Our proof of \cref{thm:matching.adj.sets} given in \cref{proof-of-main-CLT-with-bias-corr} establishes the asymptotic linearity $$\wh\tau_{bc}^\mathtt{SR}-\tauR=\P_n[\psi(D; \tauR, \eta)]+\o(n^{-1/2}),$$ where $D=(Y,A,X)$,  and the influence function $\psi(D; \tauR, \eta)$ is given by    
    \begin{multline}   \label{eq:inf-fun}
   \psi(D;\tau,\eta):= w_{g,\lambda}(X_{S_\cap})\left(\nu^\top \aipw{\mu,e}-\tau\right.\\
+\left.\lambda^\top (\Delta\aipw{\mu,e}-g(X_{S_\cap}))\left(m_1(X_{S_\cap})-g(X_{S_\cap})^\top \nu_{2:K}-\tau\right)\right),
    \end{multline} 
 where
 \(
    \tau^\mathtt{AIPW}_{\mu,e}(X):=\mu_1(X)-\mu_0(X)+\frac{\trt}{e(X)}(Y-\mu_1(X))-\frac{1-\trt}{1-e(X)}(Y-\mu_0(X))\),
 and $\Delta\aipw{k} :=\aipw{1}-\aipw{k}$, $k=2,\dots,K$. Consequently, the asymptotic variance $\mathrm{V}_\nam$ in \cref{thm:matching.adj.sets} is given by
 \begin{equation}
     \label{eq:VSR}
     \mathrm{V}_\nam = \var_P(\psi(D;\tauR, \mu, e, m_1, g, \lambda^*, \nu^*)),
 \end{equation}
with $\lambda^*$ as in \cref{propo:exist-n-unique} and the affine weights $\nu^*$ defined as: $\nu^*_1=1-\sum_{k=2}^K\nu^*_k$ and
$$\nu^*_{2:K} := \left(\E_P\left[ w^*(X_{S_\cap})\,g(X_{S_\cap})\,g(X_{S_\cap})^\top\right]\right)^{-1}\E_P\left[ w^*(X_{S_\cap})\,g(X_{S_\cap})\left(m_1(X_{S_\cap})-\tauR \right)\right].$$
Next, we introduce the perturbation paths for the nuisance parameters and establish Neyman orthogonality in \cref{propo:Neyman-orthogonality}. Assume that for some $\eps>0$,
\begin{equation}
    \label{integ-neyortho}
    \E_P\left[e^{( \|\lambda^*\|+\eps)\|g(X_{S_\cap})\|}(1+|m_1(X_{S_\cap})|)\left\{1+|Y|+\sum_{k=1}^K\sum_{a=0}^1|\mu_a(X_{S_k})|\right\}\right]<\infty.
\end{equation}
Fix measurable perturbation functions $h_{\mu, a, k}\left(X_{S_k}\right)$, $h_{e, k}\left(X_{S_k}\right)$, $h_{m_1}\left(X_{S_{\cap}}\right)$ and $h_g\left(X_{S_{\cap}}\right)$, all bounded in absolute value (or norm) by a constant $H\in(0,\infty)$, and vectors $h_\lambda \in \mathbb{R}^{K-1}$ and $h_\nu \in \mathbb{R}^K$ satisfying $\sum_{k=1}^K h_{\nu, k}=0$. For $|t|<t_0:=\min\{1,\xi/(2H),\eps/(2+2\|h_\lambda\|)\}$ (with $\xi$ as in \cref{assump-at-least-one-is-valid}), define
\begin{equation}\label{nuisance-paths}
    \begin{aligned}
        \mu_{a, t}\left(X_{S_k}\right) & =\mu_a\left(X_{S_k}\right)+t\cdot h_{\mu, a, k}\left(X_{S_k}\right), & e_t\left(X_{S_k}\right) & =e\left(X_{S_k}\right)+t\cdot h_{e, k}\left(X_{S_k}\right), \\
m_{1, t}\left(X_{S_{\cap}}\right) & =m_1\left(X_{S_{\cap}}\right)+t\cdot h_{m_1}\left(X_{S_{\cap}}\right), & g_t\left(X_{S_{\cap}}\right) & =g\left(X_{S_{\cap}}\right)+t\cdot h_g\left(X_{S_{\cap}}\right), \\
\lambda_t & =\lambda^*+t\cdot h_\lambda, & \nu_t & =\nu^*+t\cdot h_\nu .
    \end{aligned}
\end{equation}

\begin{proposition}\label{propo:Neyman-orthogonality}
Suppose that \cref{assump-iid,assump-at-least-one-is-valid,assump:heterogeneity} hold true, and that the integrability condition \eqref{integ-neyortho} holds. 
    Then, with the nuisance paths as defined in \eqref{nuisance-paths}, we have
$$\frac{d}{dt}\E_P\left[\psi(D;\, \tauR,\, \mu_{t}, \, e_t,\, m_{1,t},\, g_t,\, \lambda_t,\, \nu_t)\right]\bigg|_{t=0}=0.$$
\end{proposition}
The above result complements \cref{thm:matching.adj.sets} by providing further insights into the debiasing. It
shows that the resulting influence function in \eqref{eq:inf-fun} is robust to first-order perturbations of all the nuisance components. The proof of this result clarifies how the above choice of \(\nu^*\) and the bias-correction term jointly account for first-order errors from estimating \(\lambda^*\) and \(g\).

\subsection{Bounds on the ATE shift}\label{sec:ATE-bounds}

While the proposed method provides sharp inference for the reweighted ATE, it is natural to ask how we can bound the gap $|\tauR-\overline\tau|$ between the ATE of the reweighted population and the original. This difference is $O(1)$ in general, since we cannot test which adjustment sets satisfy ignorability. The following lemma gives an upper bound on this ATE shift.
Throughout this section, we denote the conditional contrasts by ${m_k}(X_{S_\cap}):=\E_P[\tau(X_{S_k})\mid X_{S_\cap}]$ and their averages by $\overline\tau_k:=\E_P[{m_k}(X_{S_\cap})]=\E_P[\tau(X_{S_k})]$.

\begin{lemma}\label{global-bound}
    Suppose that \cref{assump-at-least-one-is-valid,assump-iid,assump:heterogeneity} hold, and that the conditional contrasts ${m_k}(X_{S_\cap})$ are $\sigma^2$ sub-Gaussian, i.e., for each $k=1,\dots,K$ and every $t\in\R$, it holds that $$\E_P\left[\exp\left(t\left({m_k}(X_{S_\cap})-\overline\tau_k\right)\right)\right]\le \exp\left(\sigma^2 t^2/2\right).$$ Then, the difference between the ATE of the new target population and the original satisfies $$\left|\tauR-\overline\tau\right|\le \sigma\sqrt{2\dkl(w^*P\,\|\,P)}.$$
\end{lemma}

The above bound is sharp, as shown in \cref{proof-global-bound} using an example with jointly Gaussian conditional contrasts.
Another natural question is how $|\tauR-\overline\tau|$ scales with the disagreements in the estimands $\overline\tau_k=\E_P[\tau(X_{S_k})]$. Define the vector of disagreements as
\begin{equation}
    \label{vec-disagree}
    \Delta:=\E_P[g(X_{S_\cap})]=(\overline\tau_1-\overline\tau_2,\dots,\overline\tau_1-\overline\tau_K)^\top.
\end{equation}
The following result provides a local upper bound on the difference $|\tauR-\overline\tau|$ in terms of the Mahalanobis distance of $\Delta$ from zero.

\begin{lemma}\label{local-bound}
    Suppose that \cref{assump-at-least-one-is-valid,assump-iid,assump:heterogeneity} hold along a sequence  \(\smash{\{P_t\}}\). Assume that the variances of the
conditional contrasts  $\smash{m_k(X_{S_\cap})}$ under $\smash{\{P_t\}}$ are uniformly bounded and define $\Sigma_g := \var(g(X_{S_\cap}))$. Assume also that the moment generating functions of $\smash{\Sigma_g^{-1/2}g(X_{S_\cap})}$ under $\{P_t\}$ are uniformly bounded in a radius $\rc>0$ (independent of $t$). Then, 
    \[|\tauR-\overline\tau| = O(\|\Sigma_g^{-1/2}\Delta\|), \quad \text{as}\quad \|\Sigma_g^{-1/2}\Delta\|\to 0.\]
In the special case where the conditional contrasts $\{{m_k}(X_{S_\cap})\}$ are jointly Gaussian, we have the exact identity:~\(\tauR-\overline\tau=-\Delta^\top \Sigma_g^{-1}c_j\) where $c_j:=\cov_P(g(X_{S_\cap}), {m_j}(X_{S_\cap}))$ and $S_j$ is a valid adjustment set (which exists under \cref{assump-at-least-one-is-valid}).
\end{lemma}

\section{Practical considerations}\label{sec:practical}

\subsection{Diagnostic plots}\label{sec:diagnostics}

Our procedure shifts the target from the original population to a reweighted population. The resulting estimate should therefore be interpreted with respect to the new target population, which requires transparent reporting of how that population differs from the original.
To do this in practice, we recommend comparing the covariate distributions under the original and reweighted populations using histogram or density plots as in \cref{fig:hist_rewt1,fig:hist_rewt2} and summary statistics (means, quantiles). With such diagnostic plots, one can assess whether the redefined population remains substantively meaningful.

There are two possible scenarios. When the reweighted population closely resembles the original, as in \cref{fig:hist_rewt2} for our real-data example, the interpretive cost of shifting the target is modest and the reweighted ATE might be of direct scientific interest. When the distribution shift is more pronounced, one should assess whether the reweighted population remains a meaningful target for the question at hand. In this case, our framework allows imposing additional constraints on the reweighting---requiring, for instance, that some pre-specified functions of the covariates retain the same mean as in the original population---so that the new target preserves features of the original population deemed important by the researcher. We describe this extension in the following section.

\subsection{Protecting covariates}\label{sec:protect-covariates}

    We present here an extension of the proposed method that enables the researcher to protect certain features of the original population. More precisely, consider any pre-specified function $f(X_{S_\cap})$ of the covariates that define the reweighting; we aim to preserve its mean under the reweighting.
   Formally, we solve the following problem:
\begin{multline}\label{eq:KL-protect}
    w^*=\argmin_{w=w(X_{S_\cap})\ge0,\, \E_P[w]=1}\dkl(wP\,\|\, P) \quad\;\text{subject to}\\  \quad \E_{wP}[\tau(X_{S_1})]=\cdots=\E_{wP}[\tau(X_{S_K})],
   \quad \text{and}\quad  \E_{wP}[f(X_{S_\cap})]=\E_P[f(X_{S_\cap})].
\end{multline}
    Our specification-robust approach developed in \cref{new-sec:pop-level,sec:theory} extends naturally to the above problem, simply by appending $f-\E_P[f]$ to the constraint vector $\Delta\tau(X)$ defined in \cref{assump:heterogeneity}. That is, replace $\Delta\tau(X)$ with the following:
    \begin{equation}
        \label{eq:new-delta}
        \delta(X):=(\tau(X_{S_1})-\tau(X_{S_2}),\dots,\tau(X_{S_1})-\tau(X_{S_K}), f(X_{S_\cap})-\E_P[f(X_{S_\cap})]).
    \end{equation}
    Similarly, we replace the function $g$ defined in \cref{assump:heterogeneity} with the following:
    \begin{equation}
        \label{eq:new-g}
        g(X_{S_\cap}):=\E_P\left[\delta(X)\mid X_{S_\cap}\right].
    \end{equation}
We can now rewrite the optimization problem \eqref{eq:KL-protect} more compactly as:
$$w^*=\argmin_{w=w(X_{S_\cap})\ge0,\, \E_P[w]=1}\dkl(wP\,\|\, P)\quad\; \text{subject to}\quad\E_{wP}[g(X_{S_\cap})]=0,$$
which bears enough similarity to our previous formulation \eqref{adj-sets-popln-opt} that
we can empirically solve the above problem in the same manner as in \cref{sec:transfer-weights-theory}. The bias-correction, however, requires some modification to incorporate the additional constraints; we refer the reader to \cref{supp:constraints}
for a complete description. We denote by $\htauRp$ the resulting bias-corrected estimator, and the following result provides its asymptotic distribution.

\begin{figure}[t]
    \centering
\includegraphics[width=0.47\linewidth]{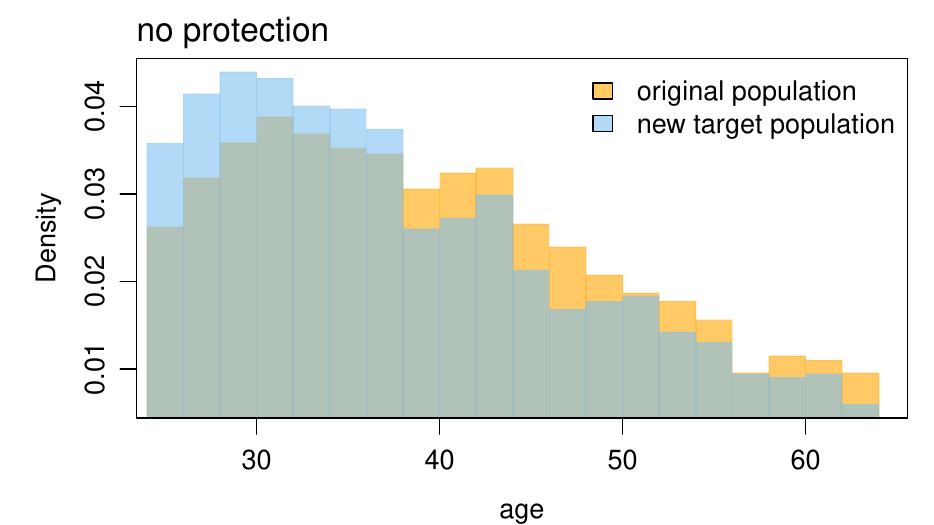}
\includegraphics[width=0.47\linewidth]{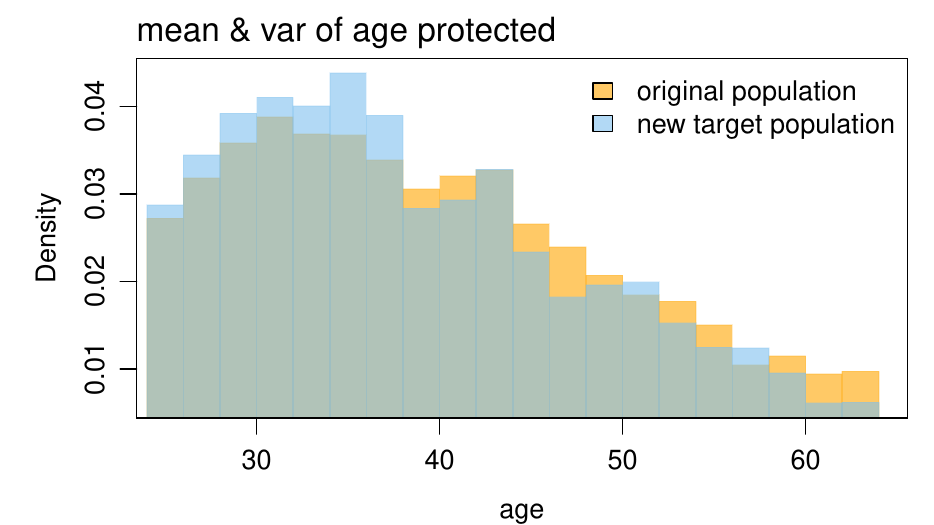}
    \caption{An illustration of covariate protection in the 401(k) application (\cref{sec:401k}). Left panel:~The marginal distributions of age without protection. Right panel:~The reweighted marginal is visibly closer to the original when we additionally protect the mean and variance of age; see \cref{sec:401k} for a quantitative comparison.}
    \label{fig:protect-example}
\end{figure}

\begin{theorem}[Asymptotic normality under additional covariate constraints]
\label{thm:matching.adj.sets-protect-covariates}    
Assume that the conditions of \cref{thm:matching.adj.sets} hold when $g$ is re-defined as in \eqref{eq:new-g} and the differences-in-contrasts vector $\Delta\tau$ is replaced with the constraint vector $\delta$ in \eqref{eq:new-delta}. Then,
$$\sqrt{n}\left(\htauRp-\tauRp\right)\dto \normal(0, \mathrm{V}_{\namp}),$$
where $\tauRp$ is the ATE of the new target population defined as the solution to \eqref{eq:KL-protect}. An explicit expression for the asymptotic variance $\mathrm{V}_{\namp}$ is deferred to \cref{supp:constraints}.
\end{theorem}

\cref{thm:matching.adj.sets-protect-covariates} yields asymptotically valid confidence intervals for the ATE $\tauRp$ in the same manner as in \cref{cor:CI-from-Theorem-1} of \cref{thm:matching.adj.sets}; we omit the analogous result here for space.

The protection extension developed above provides a way to act on the diagnostics from \cref{sec:diagnostics}. For example, the diagnostic plot (\cref{fig:hist_rewt2}) in our real-data analysis shows that the marginal distribution of age shifts toward a younger population under the reweighting. For covariates exhibiting such distributional shifts, we recommend protecting the first few moments as a starting point. We illustrate this in \cref{fig:protect-example}, where protecting the first two moments of age brings its reweighted marginal distribution visibly closer to the original; see \cref{sec:401k} for a quantitative comparison. 

Protecting a richer collection of features  naturally imposes more restrictions, which may eventually lead to infeasibility. We illustrate this in \cref{fig:moment-ladder} where we successively protect all mixed moments  up to order \(k\). As \(k\) grows, the proposed confidence interval gradually widens and the KL divergence increases, and the problem becomes infeasible at $k=12$. This highlights a tradeoff between preserving additional features of the original population and retaining enough flexibility for the reweighting to reconcile the candidate estimands.

\subsection{Infeasibility is informative}\label{sec:feasibility}

    When some adjustment sets produce estimates that are drastically different from the rest, the optimization problem \eqref{adj-sets-popln-opt} might have no solution---that is, no reweighted population exists for which all candidate estimators agree. This infeasibility is, in fact, a meaningful finding:~It highlights that the assumptions underlying the candidate adjustment sets might be fundamentally incompatible. In this scenario, our procedure can still be used to identify clusters of adjustment sets that are internally compatible, in the sense that for each cluster we can find a target population close to the original population such that all the adjustment sets \emph{in that cluster} target the ATE of this new target population. This can serve as a starting point for the practitioner to deliberate on which group of assumptions is most credible for the application at hand.

\begin{figure}[t]
    \centering
\includegraphics[width=\linewidth]{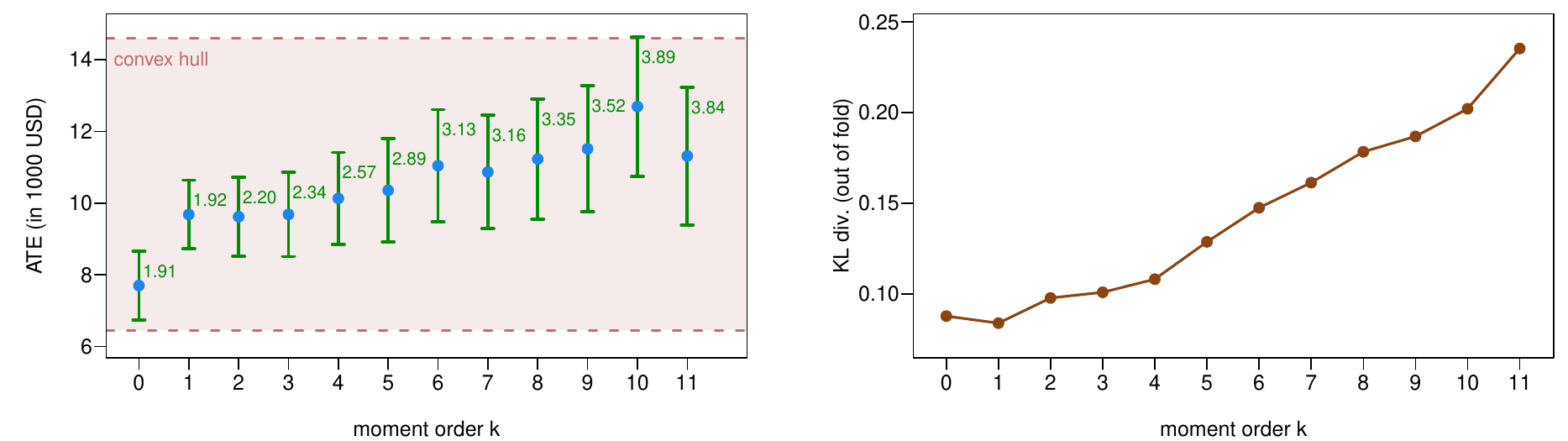}
    \caption{Specification-robust inference under increasingly rich covariate protection in the 401(k) application. We successively protect all mixed moments of the common covariates up to order \(k\). The left panel shows the 95\% confidence intervals (numerical labels are widths, blue dots show point estimates), and the right panel reports the (out-of-fold) KL divergence from the original population. The optimization problem becomes infeasible at \(k=12\).}
    \label{fig:moment-ladder}
\end{figure}

\section{Empirical illustrations}\label{sec:applications}

\subsection{Simulation examples}\label{sec:simulations}

Here we use simulation examples to empirically demonstrate the proposed method. These examples also illustrate the failure of naive remedies:~\cref{example-1} shows that the intersection of candidate adjustment sets need not be valid, while \cref{example-2} shows that the union need not be valid either, and that even a pre-treatment covariate can be a non-confounder whose inclusion leads to invalid inference.

For each of these simulation examples, we consider two candidate adjustment sets:~$S_1=\{X_1\}$ and $S_2=\{X_1,X_2\}$. We draw $n=1000$ observations and compute 95\% confidence intervals for the ATE obtained using the interacted linear regression models for both adjustment sets; see \cref{example-1,example-2} below for the results from one replication. In particular, we note that in each of these examples the confidence intervals for the two adjustment sets lead to \emph{conflicting inference} about the ATE.

\begin{example}[Two confounders]
\label{example-1} 
Draw $X_1, X_2\iid\normal(0,1)$, $A\mid X_1,X_2\sim\ber(1/(1+e^{5X_1+5X_2}))$, $Y=A\cdot(1 + X_1 - 5X_2) + 4X_2+ \eps_Y$ where $\eps_{Y}\sim\normal(0,1)$. Here $\{X_1\}$ is an invalid adjustment set while $\{X_1,X_2\}$ is valid. The true ATE is $\overline\tau=1$. The baseline method is the linear regression with treatment-covariate interactions.

For a single replication, the 95\% bootstrap confidence intervals are $[-1.85, -1.14]$ for the adjustment set $\{X_1\}$ and $[0.81,\, 1.21]$ for the adjustment set $\{X_1,X_2\}$. Applying our specification-robust procedure (the linear regression variant described in \cref{sec:lm-based}) results in a location shift of the distribution of $X_1$, namely, $\normal(c,1)$ where $c\approx 1.097$ (see \cref{sec:sim-shift} for details). The conditional distribution of $(X_2, A, Y(0), Y(1))$ given $X_1$ remains unchanged under the reweighting, since $X_2\indep X_1$ and $w=w(X_1)$. The reweighted population (\cref{fig:hist_rewt1}) has ATE $\tauR\approx 2.097$ and the proposed specification-robust CI is $[0.99,\,  2.13]$.
\end{example} 

\begin{example}[M-bias]\label{example-2} Draw $U_1,U_2,X_1\iid\normal(0,1)$, $A=\ind{U_1+X_1>0}$, $X_2 =  U_1+U_2$, $Y=A\cdot(1 + 0.5 X_1 - 6U_2)+5U_2+ \eps_Y$ where $\eps_{Y}\sim\normal(0,1)$. Here $\{X_1\}$ is a valid adjustment set while $\{X_1,X_2\}$ is invalid (adjusting for $X_2$ opens a backdoor path and introduces M-bias).  The true ATE is $\overline\tau=1$. The baseline method is the linear regression with treatment-covariate interactions. 

For a single replication, the 95\% bootstrap confidence intervals are $[0.90,\, 2.02]$ for the adjustment set $\{X_1\}$ and $[-1.30, -0.50]$ for the adjustment set $\{X_1,X_2\}$. Applying our specification-robust procedure (the linear regression variant described in \cref{sec:lm-based}) results in a location shift of the distribution of $X_1$, namely, $\normal(c,1)$ where $c=-\frac{2}{3}\sqrt{\pi}\approx -1.182$ (see \cref{sec:sim-shift} for details). The conditional distribution of $(X_2, A, Y(0), Y(1))$ given $X_1$ remains unchanged under the reweighting, since $X_2\indep X_1$ and $w=w(X_1)$. The reweighted population (\cref{fig:hist_rewt1}) has ATE $\tauR\approx 0.409$ and the proposed specification-robust CI is $[-0.38,\,  0.83]$.
\end{example}

We run the above experiments for 1000 replications and report in \cref{tab:siml_results} the empirical coverage (proportion of times the CIs contain the true target) along with the average width (averaged across replications). Note that in both of these examples, the CIs constructed using an invalid adjustment set (which differs across the two examples) provide zero coverage, which highlights the consequences of a misspecified covariate adjustment.~\cref{tab:siml_results} also verifies that both the convex hull and the proposed method provide nominal coverage for their respective target estimands. The proposed method leads to much tighter confidence intervals than the convex hull while changing the target to a shifted population as shown in \cref{fig:hist_rewt1}.

\subsection{Application:~Effect of 401(k) availability on financial assets}\label{sec:401k}

Here we empirically demonstrate the proposed procedure using the \texttt{pension} dataset from the R package \texttt{hdm} \citep{chernozhukov2016hdm}, also used as an application by \citet{Chernozhukov2018DML}. The sample consists of 9,915 households drawn from the fourth round of interviews in the 1990 panel of the Survey of Income and Program Participation \citep{chernozhukov2004}. Our goal is to study the causal effect of 401(k) availability, defined as whether any household member works for an employer offering a 401(k) plan, on  the net financial assets of that household.

The dataset includes a variety of covariates, including age, income, education, family size, marital status, home ownership, and whether the individual is covered by other pension or IRA plans. There is some ambiguity regarding which covariates should be included in the adjustment set. For instance, participation in the IRA is \emph{plausibly} affected by 401(k) availability, making it potentially a post-treatment variable rather than a confounder. 

A related ambiguity concerns income.
Following \citet{poterba1994,poterba1995}, \citet{Chernozhukov2018DML} motivate adjusting for income by arguing that workers primarily based their job choices on income and other job characteristics rather than 401(k) availability. However, \citet{ippolito1997} emphasizes that workers with stronger preferences for saving may select into jobs that offer 401(k) plans.
If 401(k) availability influences employment choices and the associated earnings, contemporaneous income could be downstream of that availability. In this dataset, income and 401(k) availability are recorded from the same round of interviews, leaving their causal ordering uncertain. Given the usual caution against adjusting for post-treatment variables, a researcher concerned about this possibility may exclude income, while another may regard it as an essential confounder.

 To demonstrate our method, we consider an example with four adjustment sets $S_1,\dots,S_4$ constructed in a nested manner: $S_1=$ \{age, education\}, $S_2 = S_1\,\cup$ \{family size, income\}, $S_3 = S_2\, \cup$ \{marital status, two-earner household\}, and $S_4 = S_3 \,\cup $ \{home-ownership, defined pension, participation in IRA\}.
 As the baselines, we construct AIPW confidence intervals based on each of the candidate adjustment sets. Our implementation
 of \cref{algo:rewt-adj-sets} uses generalized random forests (\texttt{grf} package; \citealp{Wager2018,Athey2019grf}) and 2-fold cross-fitting to learn the nuisance functions.

\begin{figure}[t]
 \centering
 \includegraphics[width=.89\linewidth]{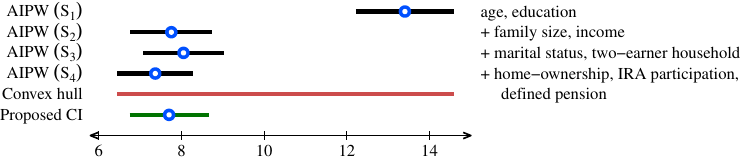}
\caption{95\% CIs for the ATE of 401(k) availability on net financial assets (in thousands of US dollars). The black lines show AIPW intervals for four nested adjustment sets, the red line shows the convex hull of the AIPW CIs, and the green line shows the proposed specification-robust CI. The circles indicate corresponding point estimates.}
\label{401(k)-example-m}
\end{figure}

\cref{401(k)-example-m} shows the 95\% confidence intervals obtained using each of these adjustment sets (in black), the naive confidence interval formed by the convex hull of the different AIPW confidence intervals (in red) and the confidence interval constructed using our specification-robust approach (in green). The proposed specification-robust confidence interval in this example is about \emph{77\% narrower} compared to the convex hull. {This demonstrates how our specification-robust procedure can be a useful alternative to committing to one adjustment set without guaranteed coverage or simply reporting the whole range.}

\cref{fig:hist_rewt2} provides a visualization of the reweighted population in this example using the marginal distributions of the covariates age, income, and education. We note that the marginal histograms of education and income for the reweighted population are visibly close to those of the original population, while the reweighted age distribution shows a shift toward younger individuals. 
We further illustrate in \cref{fig:protect-example} that protecting the mean and variance of age brings its marginal histogram visibly closer to that of the original population than in the unprotected case. Quantitatively, this reduces the KL divergence of the marginal distribution of age (in the reweighted population vs the original) by 37\% compared to the unprotected case, while paying a modest 6.6\% increase in the total divergence.

\section{Discussion}\label{sec:conclusion}

In observational studies, researchers often face multiple plausible adjustment sets, none of which can be confirmed as valid using the data alone. This poses a practical challenge, since different adjustment sets may yield conflicting conclusions and reporting the whole range of estimates does not resolve the ambiguity. Simple remedies such as adjusting for the union or intersection of the candidate adjustment sets, or restricting to pre-treatment covariates, can also fail, as we illustrate in \cref{example-2,example-1}.

We introduce a specification-robust procedure that returns a single point estimate and one confidence interval for the average treatment effect, which are valid when at least one candidate adjustment set satisfies ignorability---without requiring knowledge of which one.
Similar to how overlap weighting or trimming handle overlap violations, we resolve specification uncertainty by changing the target to a population where credible, precise inference is feasible. 
We focus on a reweighted population where all candidate estimands agree, chosen as close as possible to the original population in KL-divergence. \cref{sec:diagnostics} recommends diagnostics for assessing this population shift, and \cref{sec:protect-covariates} provides an extension that can preserve pre-specified features of the original population.

The proposed method has a few limitations worth noting. First, it requires that at least one adjustment set satisfies ignorability; it does not protect against the case where all candidate sets are invalid. Second, there could be settings where the candidate adjustment sets have no covariate in common, in which case our proposed method is not applicable.
Third, the candidate adjustment sets may be fundamentally incompatible---in the sense that no reweighted population exists for which all estimands agree. We discuss in \cref{sec:feasibility} how to handle this infeasibility issue by considering smaller clusters of the candidate adjustment sets.
    Another limitation is the interpretability of the transfer weights for individual units. In practice, binary weights can be easier to interpret and can yield a more transparent description of the resulting target population. However, extending our framework to binary weights is both analytically and computationally challenging, and beyond our current scope.

    Another scope limitation concerns heavy-tailed contrasts, which may result in extreme transfer weights. This can be diagnosed by plotting the fitted contrasts and the fitted transfer weights. In such cases, trimming or winsorizing are practical remedies to stabilize the procedure. More fundamentally, the exponential-moment restriction arises from the exponential tilting induced by the KL divergence, and extending our procedure to other divergences may permit weaker tail conditions and thus accommodate heavy-tailed contrasts.

It would also be worthwhile to extend the proposed procedure to settings with many candidate adjustment sets, where requiring exact agreement of the estimands may be overly restrictive. In such scenarios, it might be beneficial to relax our constraints by allowing approximate alignment of the estimands, or requiring only a qualitative agreement (e.g., same sign of the causal effect). We leave these directions for future research.

\section{Acknowledgments}

The authors thank Guido Imbens, Sam Pimentel, Peng Ding, and Stefan Wager for insightful discussions, and the participants of the Berkeley-Stanford Joint Colloquium and Stanford Statistics Department retreat for their valuable feedback. Rothenh\"ausler acknowledges financial support from the Dieter Schwarz Foundation, the Dudley Chamber Fund, and the David Huntington Foundation.

\bibliographystyle{apalike}
\bibliography{refs_specrobust}

@article{Crump2009,
  author   = {Crump, Richard K. and Hotz, V. Joseph and Imbens, Guido W. and Mitnik, Oscar A.},
  title    = {Dealing with Limited Overlap in Estimation of Average Treatment Effects},
  journal  = {Biometrika},
  volume   = {96},
  number   = {1},
  pages    = {187--199},
  year     = {2009}
}

@article{newey1994,
  author  = {Newey, Whitney K.},
  title   = {The Asymptotic Variance of Semiparametric Estimators},
  journal = {Econometrica},
  year    = {1994},
  volume  = {62},
  number  = {6},
  pages   = {1349--1382},
  doi     = {10.2307/2951752}
}

@article{chen2003,
  author  = {Chen, Xiaohong and Linton, Oliver and Van Keilegom, Ingrid},
  title   = {Estimation of Semiparametric Models when the Criterion Function Is Not Smooth},
  journal = {Econometrica},
  year    = {2003},
  volume  = {71},
  number  = {5},
  pages   = {1591--1608},
  doi     = {10.1111/1468-0262.00461}
}

@article{nie2021,
    author = {Nie, X and Wager, S},
    title = {Quasi-oracle estimation of heterogeneous treatment effects},
    journal = {Biometrika},
    volume = {108},
    number = {2},
    pages = {299-319},
    year = {2021},
    issn = {0006-3444},
    doi = {10.1093/biomet/asaa076},
    url = {https://doi.org/10.1093/biomet/asaa076},
    eprint = {https://academic.oup.com/biomet/article-pdf/108/2/299/37938939/asaa076.pdf},
}

@article{chernozhukov2004,
  author  = {Chernozhukov, Victor and Hansen, Christian},
  title   = {The Effects of {401(k)} Participation on the Wealth
             Distribution: An Instrumental Quantile Regression Analysis},
  journal = {The Review of Economics and Statistics},
  year    = {2004},
  volume  = {86},
  number  = {3},
  pages   = {735--751},
  doi     = {10.1162/0034653041811734},
  url     = {https://doi.org/10.1162/0034653041811734}
}

@book{ippolito1997,
  author    = {Ippolito, Richard A.},
  title     = {Pension Plans and Employee Performance:
               Evidence, Analysis, and Policy},
  publisher = {University of Chicago Press},
  address   = {Chicago},
  year      = {1997},
  url       = {https://press.uchicago.edu/ucp/books/book/chicago/P/bo3642600.html}
}

@article{Holland1986,
  author   = {Holland, Paul W.},
  title    = {Statistics and {C}ausal {I}nference},
  ajournal = {J. Amer. Statist. Assoc.},
  journal  = {Journal of the American Statistical Association},
  volume   = {81},
  year     = {1986},
  number   = {396},
  pages    = {945--970},
  issn     = {0162-1459,1537-274X},
  mrclass  = {62A99},
  mrnumber = {867618}
}

@inproceedings{ShpitserVanderWeeleRobins,
author = {Shpitser, Ilya and VanderWeele, Tyler and Robins, James M.},
title = {On the validity of covariate adjustment for estimating causal effects},
year = {2010},
isbn = {9780974903965},
publisher = {AUAI Press},
address = {Arlington, Virginia, USA},
booktitle = {Proceedings of the Twenty-Sixth Conference on Uncertainty in Artificial Intelligence},
pages = {527–536},
numpages = {10},
location = {Catalina Island, CA},
series = {UAI'10}
}

@misc{bhattacharya2026robust,
  title={Robust Weighted Triangulation of Causal Effects Under Model Uncertainty},
  author={Bhattacharya, Rohit and Ocelli, Ina and Westling, Ted},
  howpublished={arXiv:~2603.01119},
  year={2026}
}

@article{rizk2025beautiful,
  title={Beautiful weights, misinterpreted effects: The use and misuse of overlap weighting in major medical journals, 2020-2025},
  author={Rizk, John G and Lippi, Giuseppe and Lavie, Carl J},
  journal={Journal of Clinical Epidemiology},
  volume={191},
  pages={112113},
  year={2025},
  publisher={Elsevier}
}

@InProceedings{yang2024,
  title = 	 {Statistical and Causal Robustness for Causal Null Hypothesis Tests},
  author =       {Yang, Junhui and Bhattacharya, Rohit and Lee, Youjin and Westling, Ted},
  booktitle = 	 {Proceedings of the Fortieth Conference on Uncertainty in Artificial Intelligence},
  pages = 	 {3956--3978},
  year = 	 {2024},
  volume = 	 {244}
}

@article{Rosenbaum-Rubin-matching,
 ISSN = {00031305, 15372731},
 URL = {http://www.jstor.org/stable/2683903},
 author = {Paul R. Rosenbaum and Donald B. Rubin},
 journal = {The American Statistician},
 number = {1},
 pages = {33--38},
 publisher = {[American Statistical Association, Taylor & Francis, Ltd.]},
 title = {Constructing a Control Group Using Multivariate Matched Sampling Methods That Incorporate the Propensity Score},
 urldate = {2026-02-25},
 volume = {39},
 year = {1985}
}

@article{kang2024identification,
  title     = {Identification and inference with invalid instruments},
  author    = {Kang, Hyunseung and Guo, Zijian and Liu, Zhonghua and Small, Dylan},
  journal   = {Annual Review of Statistics and Its Application},
  volume    = {12},
  year      = {2024},
  pages = {385--405},
  publisher = {Annual Reviews}
}

@article{DingMiratrix2015,
  title   = {To Adjust or Not to Adjust? Sensitivity Analysis of M-Bias and Butterfly-Bias},
  author  = {Peng Ding and Luke W. Miratrix},
  pages   = {41--57},
  volume  = {3},
  number  = {1},
  journal = {Journal of Causal Inference},
  year    = {2015}
}

@article{Small_et_al_2017,
  author    = {Dylan S. Small and Zhiqiang Tan and Roland R. Ramsahai and Scott A. Lorch and M. Alan Brookhart},
  title     = {{Instrumental Variable Estimation with a Stochastic Monotonicity Assumption}},
  volume    = {32},
  journal   = {Statistical Science},
  number    = {4},
  publisher = {Institute of Mathematical Statistics},
  pages     = {561 -- 579},
  year      = {2017},
  doi       = {10.1214/17-STS623},
  url       = {https://doi.org/10.1214/17-STS623}
}

@article{Wager2018,
  author    = {Stefan Wager and Susan Athey},
  title     = {Estimation and Inference of Heterogeneous Treatment Effects using Random Forests},
  journal   = {Journal of the American Statistical Association},
  volume    = {113},
  number    = {523},
  pages     = {1228--1242},
  year      = {2018},
  publisher = {ASA Website},
  doi       = {10.1080/01621459.2017.1319839},
  url       = { 
               https://doi.org/10.1080/01621459.2017.1319839
               },
  eprint    = { 
               https://doi.org/10.1080/01621459.2017.1319839
               
               }
}

@article{Athey2019grf,
  author    = {Susan Athey and Julie Tibshirani and Stefan Wager},
  title     = {{Generalized random forests}},
  volume    = {47},
  journal   = {The Annals of Statistics},
  number    = {2},
  publisher = {Institute of Mathematical Statistics},
  pages     = {1148 -- 1178},
  year      = {2019},
  doi       = {10.1214/18-AOS1709},
  url       = {https://doi.org/10.1214/18-AOS1709}
}

@book{ImbensRubin2015,
  author     = {Imbens, Guido W. and Rubin, Donald B.},
  title      = {Causal inference---for statistics, social, and biomedical
                sciences},
  publisher  = {Cambridge University Press, New York},
  year       = {2015},
  pages      = {xx+625},
  isbn       = {978-0-521-88588-1},
  mrclass    = {62A01 (62E15 62F15 62G10 62P10 62P25)},
  mrnumber   = {3309951},
  mrreviewer = {Michael\ J.\ Evans},
  doi        = {10.1017/CBO9781139025751},
  url        = {https://doi.org/10.1017/CBO9781139025751}
}

@book{Pearl2009,
  author    = {Pearl, Judea},
  title     = {Causality{:} Models, reasoning, and inference},
  edition   = {Second},
  publisher = {Cambridge University Press, Cambridge},
  year      = {2009},
  pages     = {xx+464},
  isbn      = {978-0-521-89560-6; 0-521-77362-8},
  mrclass   = {68T37 (03B48 68-02)},
  mrnumber  = {2548166},
  doi       = {10.1017/CBO9780511803161},
  url       = {https://doi.org/10.1017/CBO9780511803161}
}

@article{Neyman1923,
  author   = {Splawa-Neyman, Jerzy},
  journal  = {Statistical Science},
  title    = {On the application of probability theory to agricultural experiments. {E}ssay on principles. {S}ection 9},
  year     = {1990},
  issn     = {0883-4237},
  number   = {4},
  pages    = {465--472},
  volume   = {5},
  fjournal = {Statistical Science. A Review Journal of the Institute of Mathematical Statistics},
  mrclass  = {01A75},
  mrnumber = {1092986}
}

@article{Rubin1974,
  title     = {Estimating causal effects of treatments in randomized and nonrandomized studies.},
  author    = {Rubin, Donald B},
  journal   = {Journal of Educational Psychology},
  volume    = {66},
  number    = {5},
  pages     = {688-701},
  year      = {1974},
  publisher = {American Psychological Association}
}

@article{Rubin1977,
  title     = {Assignment to treatment group on the basis of a covariate},
  author    = {Rubin, Donald B},
  journal   = {Journal of Educational Statistics},
  volume    = {2},
  number    = {1},
  pages     = {1-26},
  year      = {1977},
  publisher = {Sage Publications Sage CA:~Thousand Oaks, CA}
}

@article{chernozhukov2016hdm,
  author  = {Chernozhukov, Victor and Hansen, Christian and Spindler, Martin},
  title   = {\texttt{hdm}: High-Dimensional Metrics},
  journal = {The R Journal},
  year    = {2016},
  volume  = {8},
  number  = {2},
  pages   = {185--199},
  doi     = {10.32614/RJ-2016-040}
}

@article{Chernozhukov2018DML,
  author     = {Chernozhukov, Victor and Chetverikov, Denis and Demirer, Mert
                and Duflo, Esther and Hansen, Christian and Newey, Whitney and
                Robins, James},
  title      = {Double/debiased machine learning for treatment and structural
                parameters},
  ajournal   = {Econom. J.},
  journal    = {The Econometrics Journal},
  volume     = {21},
  year       = {2018},
  number     = {1},
  pages      = {C1--C68},
  issn       = {1368-4221,1368-423X},
  mrclass    = {62F10 (62F12 62G05 62G20 62P20)},
  mrnumber   = {3769544},
  mrreviewer = {Simos\ G.\ Meintanis},
  doi        = {10.1111/ectj.12097},
  url        = {https://doi.org/10.1111/ectj.12097}
}

@book{BoydVandenberghe2004,
  author     = {Boyd, Stephen and Vandenberghe, Lieven},
  title      = {Convex {O}ptimization},
  publisher  = {Cambridge University Press, Cambridge},
  year       = {2004},
  pages      = {xiv+716},
  isbn       = {0-521-83378-7},
  mrclass    = {90-01 (90C05 90C25 90C46 90C51)},
  mrnumber   = {2061575},
  mrreviewer = {Dan\ Butnariu},
  doi        = {10.1017/CBO9780511804441},
  url        = {https://doi.org/10.1017/CBO9780511804441}
}

@article{Buja2019,
  author   = {Buja, Andreas and Brown, Lawrence and Berk, Richard and
              George, Edward and Pitkin, Emil and Traskin, Mikhail and
              Zhang, Kai and Zhao, Linda},
  title    = {Models as approximations {I}: consequences illustrated with
              linear regression},
  ajournal = {Statist. Sci.},
  journal  = {Statistical Science. A Review Journal of the Institute of
              Mathematical Statistics},
  volume   = {34},
  year     = {2019},
  number   = {4},
  pages    = {523--544},
  issn     = {0883-4237,2168-8745},
  mrclass  = {62J05 (62F12 62F35 62H12 62J02)},
  mrnumber = {4048582},
  doi      = {10.1214/18-STS693},
  url      = {https://doi.org/10.1214/18-STS693}
}

@article{Anoke2019,
  author   = {Anoke, Sarah C. and Normand, Sharon-Lise and Zigler, Corwin
              M.},
  title    = {Approaches to treatment effect heterogeneity in the presence
              of confounding},
  ajournal = {Stat. Med.},
  journal  = {Statistics in Medicine},
  volume   = {38},
  year     = {2019},
  number   = {15},
  pages    = {2797--2815},
  issn     = {0277-6715,1097-0258},
  mrclass  = {99-01},
  mrnumber = {3962143},
  doi      = {10.1002/sim.8143},
  url      = {https://doi.org/10.1002/sim.8143}
}

@article{Hainmueller_Mummolo_Xu_2019,
  title   = {How Much Should We Trust Estimates from Multiplicative Interaction Models? Simple Tools to Improve Empirical Practice},
  volume  = {27},
  doi     = {10.1017/pan.2018.46},
  number  = {2},
  journal = {Political Analysis},
  author  = {Hainmueller, Jens and Mummolo, Jonathan and Xu, Yiqing},
  year    = {2019},
  pages   = {163–192}
}

@article{Lin2013,
  author     = {Lin, Winston},
  title      = {Agnostic notes on regression adjustments to experimental data:{~R}eexamining {F}reedman's critique},
  ajournal   = {Ann. Appl. Stat.},
  journal    = {The Annals of Applied Statistics},
  volume     = {7},
  year       = {2013},
  number     = {1},
  pages      = {295--318},
  issn       = {1932-6157,1941-7330},
  mrclass    = {62B15 (62J10 62P15)},
  mrnumber   = {3086420},
  mrreviewer = {Miguel\ Fonseca},
  doi        = {10.1214/12-AOAS583},
  url        = {https://doi.org/10.1214/12-AOAS583}
}

@article{DonskerVaradhan1976,
  author     = {Donsker, M. D. and Varadhan, S. R. S.},
  title      = {Asymptotic evaluation of certain {M}arkov process expectations
                for large time. {III}},
  ajournal   = {Comm. Pure Appl. Math.},
  journal    = {Communications on Pure and Applied Mathematics},
  volume     = {29},
  year       = {1976},
  number     = {4},
  pages      = {389--461},
  issn       = {0010-3640,1097-0312},
  mrclass    = {60J25 (60G15)},
  mrnumber   = {428471},
  mrreviewer = {D.\ W.\ Stroock},
  doi        = {10.1002/cpa.3160290405},
  url        = {https://doi.org/10.1002/cpa.3160290405}
}

@article{ImbensAngrist1994,
  issn      = {00129682, 14680262},
  url       = {http://www.jstor.org/stable/2951620},
  author    = {Guido W. Imbens and Joshua D. Angrist},
  journal   = {Econometrica},
  number    = {2},
  pages     = {467--475},
  publisher = {[Wiley, Econometric Society]},
  title     = {Identification and Estimation of Local Average Treatment Effects},
  urldate   = {2025-04-28},
  volume    = {62},
  year      = {1994}
}

@article{HorvitzThompson1952,
  author    = {D. G. Horvitz and D. J. Thompson},
  title     = {A Generalization of Sampling Without Replacement from a Finite Universe},
  journal   = {Journal of the American Statistical Association},
  ajournal  = {J. Amer. Statist. Assoc.},
  volume    = {47},
  number    = {260},
  pages     = {663--685},
  year      = {1952},
  publisher = {ASA Website},
  doi       = {10.1080/01621459.1952.10483446},
  url       = { 
               https://www.tandfonline.com/doi/abs/10.1080/01621459.1952.10483446
               },
  eprint    = {https://www.tandfonline.com/doi/pdf/10.1080/01621459.1952.10483446
               }
}

@article{DevilleSarndal,
  author    = {Jean-Claude Deville and Carl-Erik Särndal},
  title     = {Calibration Estimators in Survey Sampling},
  journal   = {Journal of the American Statistical Association},
  volume    = {87},
  number    = {418},
  pages     = {376--382},
  year      = {1992},
  publisher = {ASA Website},
  doi       = {10.1080/01621459.1992.10475217},
  url       = { 
               https://www.tandfonline.com/doi/abs/10.1080/01621459.1992.10475217
               },
  eprint    = { 
               https://www.tandfonline.com/doi/pdf/10.1080/01621459.1992.10475217
               }
}

@article{Hainmueller2012,
  title   = {Entropy Balancing for Causal Effects: A Multivariate Reweighting Method to Produce Balanced Samples in Observational Studies},
  volume  = {20},
  doi     = {10.1093/pan/mpr025},
  number  = {1},
  journal = {Political Analysis},
  author  = {Hainmueller, Jens},
  year    = {2012},
  pages   = {25–46}
}

@article{Robins1994,
  author    = {James M. Robins and Andrea Rotnitzky and Lue Ping Zhao},
  title     = {Estimation of Regression Coefficients When Some Regressors are not Always Observed},
  journal   = {Journal of the American Statistical Association},
  volume    = {89},
  number    = {427},
  pages     = {846--866},
  year      = {1994},
  publisher = {ASA Website},
  doi       = {10.1080/01621459.1994.10476818},
  url       = { 
               https://doi.org/10.1080/01621459.1994.10476818
               },
  eprint    = { 
               https://doi.org/10.1080/01621459.1994.10476818
               }
}

@article{Li2018,
  author    = {Fan Li and Kari Lock Morgan and Alan M. Zaslavsky},
  title     = {Balancing Covariates via Propensity Score Weighting},
  journal   = {Journal of the American Statistical Association},
  volume    = {113},
  number    = {521},
  pages     = {390--400},
  year      = {2018},
  publisher = {ASA Website},
  doi       = {10.1080/01621459.2016.1260466},
  url       = { 
               https://doi.org/10.1080/01621459.2016.1260466
               },
  eprint    = {   https://doi.org/10.1080/01621459.2016.1260466
               }
}

@article{Steegen2016,
  issn      = {17456916, 17456924},
  url       = {http://www.jstor.org/stable/26358709},
  author    = {Sara Steegen and Francis Tuerlinckx and Andrew Gelman and Wolf Vanpaemel},
  journal   = {Perspectives on Psychological Science},
  number    = {5},
  pages     = {702--712},
  publisher = {[Association for Psychological Science, Sage Publications, Inc.]},
  title     = {Increasing Transparency Through a Multiverse Analysis},
  urldate   = {2025-04-28},
  volume    = {11},
  year      = {2016}
}

@misc{chernozhukov2021automatic,
  title        = {Automatic debiased machine learning via Riesz regression},
  author       = {Chernozhukov, Victor and Newey, Whitney K and Quintas-Martinez, Victor and Syrgkanis, Vasilis},
  howpublished = {arXiv:~2104.14737},
  year         = {2021}
}

@article{Hahn1998,
  author   = {Hahn, Jinyong},
  title    = {On the role of the propensity score in efficient
              semiparametric estimation of average treatment effects},
  journal  = {Econometrica},
  fjournal = {Econometrica. Journal of the Econometric Society},
  volume   = {66},
  year     = {1998},
  number   = {2},
  pages    = {315--331},
  issn     = {0012-9682,1468-0262},
  mrclass  = {62P20 (62G05)},
  mrnumber = {1612242},
  doi      = {10.2307/2998560},
  url      = {https://doi.org/10.2307/2998560}
}

@article{Box1980,
  title     = {Sampling and Bayes’ inference in scientific modelling and robustness},
  author    = {Box, George EP},
  journal   = {Journal of the Royal Statistical Society Series A: Statistics in Society},
  volume    = {143},
  number    = {4},
  pages     = {383--404},
  year      = {1980},
  publisher = {Oxford University Press}
}

@article{Skene1986,
  title     = {Bayesian modelling and sensitivity analysis},
  author    = {Skene, AM and Shaw, JEH and Lee, TD},
  journal   = {Journal of the Royal Statistical Society: Series D (The Statistician)},
  volume    = {35},
  number    = {2},
  pages     = {281--288},
  year      = {1986},
  publisher = {Wiley Online Library}
}

@article{Leamer1983,
  title     = {Let's take the con out of econometrics},
  author    = {Leamer, Edward E},
  journal   = {The American Economic Review},
  volume    = {73},
  number    = {1},
  pages     = {31--43},
  year      = {1983},
  publisher = {JSTOR}
}

@incollection{poterba1994,
  author    = {Poterba, James M. and Venti, Steven F. and Wise, David A.},
  title     = {{401(k)} Plans and Tax-Deferred Saving},
  booktitle = {Studies in the Economics of Aging},
  editor    = {Wise, David A.},
  pages     = {105--142},
  year      = {1994},
  publisher = {University of Chicago Press},
  address   = {Chicago}
}

@article{poterba1995,
  author  = {Poterba, James M. and Venti, Steven F. and Wise, David A.},
  title   = {Do {401(k)} Contributions Crowd Out Other Personal Saving?},
  journal = {Journal of Public Economics},
  volume  = {58},
  number  = {1},
  pages   = {1--32},
  year    = {1995},
  doi     = {10.1016/0047-2727(94)01462-W}
}

@article{Lalonde1986,
  issn      = {00028282},
  url       = {http://www.jstor.org/stable/1806062},
  author    = {Robert J. LaLonde},
  journal   = {The American Economic Review},
  number    = {4},
  pages     = {604--620},
  publisher = {American Economic Association},
  title     = {Evaluating the Econometric Evaluations of Training Programs with Experimental Data},
  urldate   = {2025-05-12},
  volume    = {76},
  year      = {1986}
}

@article{Rosenbaum1987,
  title     = {The role of a second control group in an observational study},
  author    = {Rosenbaum, Paul R},
  journal   = {Statistical Science},
  volume    = {2},
  number    = {3},
  pages     = {292--306},
  year      = {1987},
  publisher = {Institute of Mathematical Statistics}
}

@article{Yu2013stability,
  author    = {Bin Yu},
  title     = {{Stability}},
  volume    = {19},
  journal   = {Bernoulli},
  number    = {4},
  publisher = {Bernoulli Society for Mathematical Statistics and Probability},
  pages     = {1484 -- 1500},
  year      = {2013},
  doi       = {10.3150/13-BEJSP14},
  url       = {https://doi.org/10.3150/13-BEJSP14}
}

@article{Yu2020veridical,
  title   = {Veridical data science},
  author  = {Yu, Bin and Kumbier, Karl},
  journal = {Proceedings of the National Academy of Sciences},
  volume  = {117},
  number  = {8},
  pages   = {3920--3929},
  year    = {2020}
}

@article{Karmakar2019integrating,
  title     = {Integrating the evidence from evidence factors in observational studies},
  author    = {Karmakar, Bikram and French, Benjamin and Small, Dylan S},
  journal   = {Biometrika},
  volume    = {106},
  number    = {2},
  pages     = {353--367},
  year      = {2019},
  publisher = {Oxford University Press}
}

@article{Rosenbaum2010evidence,
  title     = {Evidence factors in observational studies},
  author    = {Rosenbaum, Paul R},
  journal   = {Biometrika},
  volume    = {97},
  number    = {2},
  pages     = {333--345},
  year      = {2010},
  publisher = {Oxford University Press}
}

@article{Freedman1991statistical,
  issn      = {00811750, 14679531},
  url       = {http://www.jstor.org/stable/270939},
  author    = {David A. Freedman},
  journal   = {Sociological Methodology},
  pages     = {291--313},
  publisher = {[American Sociological Association, Wiley, Sage Publications, Inc.]},
  title     = {Statistical Models and Shoe Leather},
  urldate   = {2025-05-12},
  volume    = {21},
  year      = {1991}
}

@article{Imbens2003sensitivity,
  title     = {Sensitivity to exogeneity assumptions in program evaluation},
  author    = {Imbens, Guido W},
  journal   = {American Economic Review},
  volume    = {93},
  number    = {2},
  pages     = {126--132},
  year      = {2003},
  publisher = {American Economic Association}
}

@book{Rosenbaum2002,
  author    = {Rosenbaum, Paul R.},
  title     = {Observational studies},
  series    = {Springer Series in Statistics},
  edition   = {Second},
  publisher = {Springer-Verlag, New York},
  year      = {2002},
  pages     = {xiv+375},
  isbn      = {0-387-98967-6},
  mrclass   = {62-02 (62-07)},
  mrnumber  = {1899138},
  doi       = {10.1007/978-1-4757-3692-2},
  url       = {https://doi.org/10.1007/978-1-4757-3692-2}
}

@article{Anderson1982,
  issn      = {00905364, 21688966},
  url       = {http://www.jstor.org/stable/2240714},
  author    = {P. K. Andersen and R. D. Gill},
  journal   = {The Annals of Statistics},
  number    = {4},
  pages     = {1100--1120},
  publisher = {Institute of Mathematical Statistics},
  title     = {Cox's Regression Model for Counting Processes: A Large Sample Study},
  urldate   = {2025-05-05},
  volume    = {10},
  year      = {1982}
}

@article{Scharfstein1999adjusting,
  title     = {Adjusting for nonignorable drop-out using semiparametric nonresponse models},
  author    = {Scharfstein, Daniel O and Rotnitzky, Andrea and Robins, James M},
  journal   = {Journal of the American Statistical Association},
  volume    = {94},
  number    = {448},
  pages     = {1096--1120},
  year      = {1999},
  publisher = {Taylor \& Francis}
}

@article{Robins1995semiparametric,
  title     = {Semiparametric efficiency in multivariate regression models with missing data},
  author    = {Robins, James M and Rotnitzky, Andrea},
  journal   = {Journal of the American Statistical Association},
  volume    = {90},
  number    = {429},
  pages     = {122--129},
  year      = {1995},
  publisher = {Taylor \& Francis}
}

@article{DornGuo2023,
  author    = {Jacob Dorn and Kevin Guo},
  title     = {Sharp Sensitivity Analysis for Inverse Propensity Weighting via Quantile Balancing},
  journal   = {Journal of the American Statistical Association},
  volume    = {118},
  number    = {544},
  pages     = {2645--2657},
  year      = {2023},
  publisher = {ASA Website},
  doi       = {10.1080/01621459.2022.2069572},
  url       = { 
               https://doi.org/10.1080/01621459.2022.2069572
               },
  eprint    = { https://doi.org/10.1080/01621459.2022.2069572
               }
}

@article{cornfield1959smoking,
  author   = {Cornfield, Jerome and Haenszel, William and Hammond, E. Cuyler and Lilienfeld, Abraham M. and Shimkin, Michael B. and Wynder, Ernst L.},
  title    = {Smoking and Lung Cancer: Recent Evidence and a Discussion of Some Questions},
  journal  = {JNCI: Journal of the National Cancer Institute},
  volume   = {22},
  number   = {1},
  pages    = {173-203},
  year     = {1959},
  issn     = {0027-8874},
  doi      = {10.1093/jnci/22.1.173},
  url      = {https://doi.org/10.1093/jnci/22.1.173},
  eprint   = {https://academic.oup.com/jnci/article-pdf/22/1/173/2704718/22-1-173.pdf}
}

@article{DornGuoKallus2025,
  author    = {Jacob Dorn and Kevin Guo and Nathan Kallus},
  title     = {Doubly-Valid/Doubly-Sharp Sensitivity Analysis for Causal Inference with Unmeasured Confounding},
  journal   = {Journal of the American Statistical Association},
  volume    = {120},
  number    = {549},
  pages     = {331--342},
  year      = {2025},
  publisher = {ASA Website},
  doi       = {10.1080/01621459.2024.2335588},
  url       = {         https://doi.org/10.1080/01621459.2024.2335588
               },
  eprint    = {        https://doi.org/10.1080/01621459.2024.2335588
               }
}

@article{CinelliHazlett2020,
    author = {Cinelli, Carlos and Hazlett, Chad},
    title = {Making Sense of Sensitivity: Extending Omitted Variable Bias},
    journal = {Journal of the Royal Statistical Society Series B: Statistical Methodology},
    volume = {82},
    number = {1},
    pages = {39-67},
    year = {2020},
    issn = {1369-7412},
    doi = {10.1111/rssb.12348},
    url = {https://doi.org/10.1111/rssb.12348},
    eprint = {https://academic.oup.com/jrsssb/article-pdf/82/1/39/49320681/jrsssb_82_1_39.pdf},
}

@article{Zhao2019,
  author   = {Zhao, Qingyuan and Small, Dylan S. and Bhattacharya, Bhaswar B.},
  title    = {Sensitivity Analysis for Inverse Probability Weighting Estimators via the Percentile Bootstrap},
  journal  = {Journal of the Royal Statistical Society Series B: Statistical Methodology},
  volume   = {81},
  number   = {4},
  pages    = {735-761},
  year     = {2019},
  issn     = {1369-7412},
  doi      = {10.1111/rssb.12327},
  url      = {https://doi.org/10.1111/rssb.12327},
  eprint   = {https://academic.oup.com/jrsssb/article-pdf/81/4/735/49269653/jrsssb\_81\_4\_735.pdf}
}

@article{Rosenbaum1983sensitivity,
  issn      = {00359246},
  url       = {http://www.jstor.org/stable/2345524},
  author    = {P. R. Rosenbaum and D. B. Rubin},
  journal   = {Journal of the Royal Statistical Society. Series B (Methodological)},
  number    = {2},
  pages     = {212--218},
  publisher = {[Royal Statistical Society, Oxford University Press]},
  title     = {Assessing Sensitivity to an Unobserved Binary Covariate in an Observational Study with Binary Outcome},
  urldate   = {2025-05-12},
  volume    = {45},
  year      = {1983}
}

@article{Rosenbaum1983,
    author    = {P. R. Rosenbaum and D. B. Rubin},
    title = {The central role of the propensity score in observational studies for causal effects},
    journal = {Biometrika},
    volume = {70},
    number = {1},
    pages = {41-55},
    year = {1983},
    month = {04},
    issn = {0006-3444},
    doi = {10.1093/biomet/70.1.41},
    url = {https://doi.org/10.1093/biomet/70.1.41},
    eprint = {https://academic.oup.com/biomet/article-pdf/70/1/41/662954/70-1-41.pdf},
}

@article{brumback2004sensitivity,
  title     = {Sensitivity analyses for unmeasured confounding assuming a marginal structural model for repeated measures},
  author    = {Brumback, Babette A and Hern{\'a}n, Miguel A and Haneuse, Sebastien JPA and Robins, James M},
  journal   = {Statistics in {M}edicine},
  volume    = {23},
  number    = {5},
  pages     = {749--767},
  year      = {2004},
  publisher = {Wiley Online Library}
}

@article{Imai2010,
  author    = {Kosuke Imai and Luke Keele and Teppei Yamamoto},
  title     = {{Identification, Inference and Sensitivity Analysis for Causal Mediation Effects}},
  volume    = {25},
  journal   = {Statistical Science},
  number    = {1},
  publisher = {Institute of Mathematical Statistics},
  pages     = {51 -- 71},
  year      = {2010},
  doi       = {10.1214/10-STS321},
  url       = {https://doi.org/10.1214/10-STS321}
}

@article{Oster2019,
  author    = {Emily Oster},
  title     = {Unobservable Selection and Coefficient Stability: Theory and Evidence},
  journal   = {Journal of Business \& Economic Statistics},
  volume    = {37},
  number    = {2},
  pages     = {187--204},
  year      = {2019},
  publisher = {ASA Website},
  doi       = {10.1080/07350015.2016.1227711},
  url       = { 
               https://doi.org/10.1080/07350015.2016.1227711
               },
  eprint    = {      https://doi.org/10.1080/07350015.2016.1227711
               }
}

@article{Robins1999,
  issn      = {00397857, 15730964},
  url       = {http://www.jstor.org/stable/20118224},
  author    = {James M. Robins},
  journal   = {Synthese},
  number    = {1/2},
  pages     = {151--179},
  publisher = {Springer},
  title     = {Association, Causation, and Marginal Structural Models},
  urldate   = {2025-05-12},
  volume    = {121},
  year      = {1999}
}

@article{Kenneth2000,
  author   = {Kenneth A. Frank},
  title    = {Impact of a Confounding Variable on a Regression Coefficient},
  journal  = {Sociological Methods \& Research},
  volume   = {29},
  number   = {2},
  pages    = {147-194},
  year     = {2000},
  doi      = {10.1177/0049124100029002001},
  url      = { 
              
              https://doi.org/10.1177/0049124100029002001
              
              
              
              },
  eprint   = { 
              
              https://doi.org/10.1177/0049124100029002001
              
              
              
              }
}

@article{Blackwell2014,
  title   = {A Selection Bias Approach to Sensitivity Analysis for Causal Effects},
  volume  = {22},
  doi     = {10.1093/pan/mpt006},
  number  = {2},
  journal = {Political Analysis},
  author  = {Blackwell, Matthew},
  year    = {2014},
  pages   = {169–182}
}

@article{vanderweele2011,
  title     = {Bias formulas for sensitivity analysis of unmeasured confounding for general outcomes, treatments, and confounders},
  author    = {VanderWeele, Tyler J and Arah, Onyebuchi A},
  journal   = {Epidemiology},
  volume    = {22},
  number    = {1},
  pages     = {42--52},
  year      = {2011},
  publisher = {LWW}
}

@article{VanderWeeleShpitser2011,
  author   = {VanderWeele, Tyler J. and Shpitser, Ilya},
  title    = {A New Criterion for Confounder Selection},
  journal  = {Biometrics},
  volume   = {67},
  number   = {4},
  pages    = {1406-1413},
  year     = {2011},
  issn     = {0006-341X},
  doi      = {10.1111/j.1541-0420.2011.01619.x},
  url      = {https://doi.org/10.1111/j.1541-0420.2011.01619.x},
  eprint   = {https://academic.oup.com/biometrics/article-pdf/67/4/1406/53164609/biometrics_67_4_1406.pdf}
}

@article{Csiszar1975,
  issn      = {00911798, 2168894X},
  url       = {http://www.jstor.org/stable/2959270},
  author    = {Csiszár, Imre.},
  journal   = {The Annals of Probability},
  number    = {1},
  pages     = {146--158},
  publisher = {Institute of Mathematical Statistics},
  title     = {{I}-Divergence Geometry of Probability Distributions and Minimization Problems},
  urldate   = {2026-01-23},
  volume    = {3},
  year      = {1975}
}

@book{CoverThomas2005MaximumEntropy,
  author    = {Cover, Thomas M. and Thomas, Joy A.},
  title     = {Elements of Information Theory},
  chapter   = {12},
  pages     = {409--425},
  publisher = {John Wiley \& Sons, Ltd},
  edition   = {2nd},
  year      = {2005},
  doi       = {10.1002/047174882X.ch12},
  url       = {https://onlinelibrary.wiley.com/doi/abs/10.1002/047174882X.ch12}
}

@misc{chang2025postselectioninferencecausaleffects,
  title        = {Post-selection inference for causal effects after causal discovery},
  author       = {Ting-Hsuan Chang and Zijian Guo and Daniel Malinsky},
  year         = {2025},
  eprint       = {2405.06763},
  howpublished = {arXiv:~2405.06763},
  primaryclass = {stat.ME}
}

@misc{guo2023confounderselectionobjectivesapproaches,
  title        = {Confounder Selection: Objectives and Approaches},
  author       = {F. Richard Guo and Anton Rask Lundborg and Qingyuan Zhao},
  year         = {2023},
  eprint       = {2208.13871},
  howpublished = {arXiv:~2208.13871},
  primaryclass = {stat.ME}
}

@article{GuoEmaRotnitzky2023,
  author   = {Guo, F Richard and Perković, Emilija and Rotnitzky, Andrea},
  title    = {Variable elimination, graph reduction and the efficient g-formula},
  journal  = {Biometrika},
  volume   = {110},
  number   = {3},
  pages    = {739-761},
  year     = {2023},
  doi      = {10.1093/biomet/asac062}
}

@article{Kang2016,
  author    = {Hyunseung Kang and Anru Zhang and T. Tony Cai and Dylan S. Small},
  title     = {Instrumental Variables Estimation With Some Invalid Instruments and its Application to Mendelian Randomization},
  journal   = {Journal of the American Statistical Association},
  volume    = {111},
  number    = {513},
  pages     = {132--144},
  year      = {2016},
  publisher = {Taylor \& Francis},
  doi       = {10.1080/01621459.2014.994705},
  url       = { https://doi.org/10.1080/01621459.2014.994705
               },
  eprint    = {https://doi.org/10.1080/01621459.2014.994705
               }
}

@misc{rakshit2024localeffectscontinuousinstruments,
  title        = {Local Effects of Continuous Instruments without Positivity},
  author       = {Prabrisha Rakshit and Alexander Levis and Luke Keele},
  year         = {2024},
  eprint       = {2409.07350},
  howpublished = {arXiv:~2409.07350},
  primaryclass = {stat.ME}
}

@book{petrov,
  title={Sums of independent random variables},
  author={Petrov, Valentin V},
  year={1975},
  publisher={Springer-Verlag, Berlin}
}

@article{Hosman2010,
  author    = {Carrie A. Hosman and Ben B. Hansen and Paul W. Holland},
  title     = {{The sensitivity of linear regression coefficients’ confidence limits to the omission of a confounder}},
  volume    = {4},
  journal   = {The Annals of Applied Statistics},
  number    = {2},
  publisher = {Institute of Mathematical Statistics},
  pages     = {849 -- 870},
  year      = {2010},
  doi       = {10.1214/09-AOAS315},
  url       = {https://doi.org/10.1214/09-AOAS315}
}

@article{dorie2016flexible,
  title     = {A flexible, interpretable framework for assessing sensitivity to unmeasured confounding},
  author    = {Dorie, Vincent and Harada, Masataka and Carnegie, Nicole Bohme and Hill, Jennifer},
  journal   = {Statistics in {M}edicine},
  volume    = {35},
  number    = {20},
  pages     = {3453--3470},
  year      = {2016},
  publisher = {Wiley Online Library}
}

@article{Franks2020,
  author    = {AlexanderM. Franks and Alexander D{’}Amour and Avi Feller},
  title     = {Flexible Sensitivity Analysis for Observational Studies Without Observable Implications},
  journal   = {Journal of the American Statistical Association},
  volume    = {115},
  number    = {532},
  pages     = {1730--1746},
  year      = {2020},
  publisher = {ASA Website},
  doi       = {10.1080/01621459.2019.1604369},
  url       = { 
               
               https://doi.org/10.1080/01621459.2019.1604369
               
               
               
               },
  eprint    = { 
               
               https://doi.org/10.1080/01621459.2019.1604369
               
               
               
               }
}

@article{orben2019association,
  title     = {The association between adolescent well-being and digital technology use},
  author    = {Orben, Amy and Przybylski, Andrew K},
  journal   = {Nature Human Behaviour},
  volume    = {3},
  number    = {2},
  pages     = {173--182},
  year      = {2019},
  publisher = {Nature Publishing Group}
}

@book{vdV, place={Cambridge}, series={Cambridge Series in Statistical and Probabilistic Mathematics}, title={Asymptotic Statistics}, publisher={Cambridge University Press}, author={van der Vaart, A. W.}, year={1998}, collection={Cambridge Series in Statistical and Probabilistic Mathematics}}

@article{neweyMcFadden1994,
  title={Large sample estimation and hypothesis testing},
  author={Newey, Whitney K and McFadden, Daniel},
  journal={Handbook of econometrics},
  volume={4},
  pages={2111--2245},
  year={1994},
  publisher={Elsevier}
}

@article{Rubin1973,
 ISSN = {0006341X, 15410420},
 URL = {http://www.jstor.org/stable/2529684},
 author = {Donald B. Rubin},
 journal = {Biometrics},
 number = {1},
 pages = {159--183},
 publisher = {International Biometric Society},
 title = {Matching to Remove Bias in Observational Studies},
 urldate = {2026-10-16},
 volume = {29},
 year = {1973}
}
\normalsize
\newpage

\appendixpageoff
\appendixtitleoff
\appendix

\numberwithin{equation}{section}
\numberwithin{figure}{section}
\numberwithin{table}{section}
\numberwithin{remark}{section}
\numberwithin{lemma}{section}
\numberwithin{proposition}{section}
\numberwithin{theorem}{section}
\numberwithin{corollary}{section}
\numberwithin{algorithm}{section}
\numberwithin{assumption}{section}

\part{Appendices}

\numberwithin{equation}{section}
\counterwithin{figure}{section}
\counterwithin{table}{section}
\counterwithin{remark}{section}
\counterwithin{lemma}{section}
\counterwithin{proposition}{section}
\counterwithin{theorem}{section}
\counterwithin{corollary}{section}
\counterwithin{algorithm}{section}
\counterwithin{assumption}{section}
\normalsize

\parttoc

\newpage

\section{Specification-robust inference with parametric baseline}\label{sec:lm-based}

The linear regression model with treatment-covariate interactions \citep{Lin2013} is one of the most commonly used parametric approaches for drawing inference on the ATE in applied research, often preferred for its straightforward implementation or transparent interpretability of the coefficients; see discussions in \citet[Section 1]{Hainmueller_Mummolo_Xu_2019} and \citet[Section 3.1]{Anoke2019}. 

In this section, we consider the problem of drawing inference with multiple adjustment sets using the linear model with treatment-covariate interactions as the baseline method for estimating the contrasts $\tau(X_{S_k})=\E_P[Y\mid A=1, X_{S_k}]-\E_P[Y\mid A=0, X_{S_k}]$.~Throughout this section, we work under the assumption that such an interacted linear model is well-specified for at least one of the adjustment sets.

\begin{assumption}\label{assump:linear-model}
   Assume that $j\in\{1,2,\dots,K\}$ is such that the adjustment set $S_j$ is valid and the data $(Y(0),Y(1),A,X)$ generated from $P$ satisfy the following:  $$\E_P[Y(a)\mid X_{S_j}] = \overline\delta +  X_{S_j}^\top\overline\beta + a\cdot(\overline{\tau}+  X_{S_j}^\top \overline\gamma), \quad a=0,1.$$
   Assume further that (i) the covariates are centered, i.e., $\E_P[X]=0$, (ii) the adjustment sets have non-empty intersection:~$S_\cap=\cap_{k=1}^K S_k\neq \emptyset$, (iii) the design matrix for each regression is nonsingular, and (iv) $\E_P[|Y|^4+\|X\|^4]<\infty$.
\end{assumption}

Empirically, we obtain $\wh{\tau}(X_{S_k})$ by fitting a linear regression with interaction between the treatment $\trt$ and the covariates $X_{S_k}$. Thus, 
\begin{equation}\label{def:contrasts-lm-empirical}
\begin{split}
        \wh{\tau}(X_{S_k})&:=\wh{\tau}_k+X_{S_k}^\top \wh{\gamma}_k,\\
   (\wh\delta_k,\,\wh\beta_k,\,\wh\tau_k,\,\wh\gamma_k) &:= \argmin_{\delta,\,\beta,\,\tau,\,\gamma}\ \sum_{i=1}^n \left(Y_i-\delta -X_{i,S_k}^\top \beta-\trt_i(\tau+X_{i,S_k}^\top \gamma)\right)^2.
\end{split}
\end{equation}
Even when the adjustment set is possibly invalid, it is well-known (see, e.g.,~\citet[Proposition 7.1]{Buja2019}) that $\wh{\tau}(X_{S_k})$ targets the following estimand: 
\begin{equation}\label{def:contrasts-lm}
\begin{split}
        \tau^*(X_{S_k})&:=\tau_k^*+X_{S_k}^\top \gamma_k^*,\\
        (\delta_k^*,\,\beta_k^*,\,\tau_k^*,\,\gamma_k^*) &:= \argmin_{\delta,\,\beta,\,\tau,\,\gamma}\ \E_P\left(Y-\delta -X_{S_k}^\top \beta-\trt(\tau+X_{S_k}^\top \gamma)\right)^2.
\end{split}
\end{equation}
Under \cref{assump:linear-model}, $\tau_j^*$ is the same as the ATE $\overline\tau$ (and $\gamma_j^*=\overline{\gamma}$). However, $\tau_k^*$ does not have a causal interpretation when the adjustment set $S_k$ is invalid. Our core approach remains the same as in \cref{sec:theory}, except that here we aim to solve the population optimization problem \eqref{adj-sets-popln-opt} with $\tau^*(X_{S_k})$ instead of $\tau(X_{S_k})$. Thus, we find a reweighted population $w^* P$ closest in KL-divergence to the original population $P$ such that $$\E_P[w^*(X_{S_\cap})\tau^*(X_{S_1})]=\cdots =\E_P[w^*(X_{S_\cap})\tau^*(X_{S_K})].$$ 
The following lemma mirrors \cref{propo:reweighted-estimand-makes-sense} for the setting of this section and shows that estimand in the above display has a causal interpretation---it is the ATE of the reweighted population; see
\cref{proof-of-propo3} 
for a proof.
\begin{lemma}[Interpretation of the new estimand]\label{propo:reweighted-estimand-makes-sense-lm-based}
    Suppose that the adjustment sets $S_1,\dots,S_K$ satisfy \cref{assump:linear-model} and $w^*=w^*(X_{S_\cap})$ solves the optimization problem
    \begin{equation}\label{adj-sets-popln-opt-lm-based}
         w^*=\argmin_{w=w(X_{S_\cap})\ge 0,\, \E_P[w]=1}\dkl(wP\,\|\, P)\quad  \text{subject to}\quad \E_{wP}[\tau^*(X_{S_1})]=\cdots=\E_{wP}[\tau^*(X_{S_K})].
    \end{equation}
    Then, for each $k=1,2,\dots,K$, it holds that $$\E_P[w^*(X_{S_\cap})\tau^*(X_{S_k})]=\E_{w^*P}[Y(1)-Y(0)]=\tauR.$$
\end{lemma}

\subsection{Existence and uniqueness of the transfer weights}

Here we derive a result analogous to \cref{propo:exist-n-unique} for the setting of this section by replacing $\tau(X_{S_k})$ with $\tau^*(X_{S_k})$ in \cref{propo:exist-n-unique}. We state this result as \cref{propo:exist-n-unique-lm-based} below, which shows the existence and uniqueness of the weighting function $w^*$ in \eqref{adj-sets-popln-opt-lm-based} under the following heterogeneity assumption.

\begin{assumption}\label{assump:heterogeneity-lm-based}
    Define differences-in-contrasts $\Delta\tau_k^*(X):=\tau^*(X_{S_1})-\tau^*(X_{S_k})$ for $k=2,\dots,K$, where $\tau^*(X_{S_k})$ are as defined in \eqref{def:contrasts-lm}, and denote $\Delta\tau^*:=(\Delta\tau_2^*,\dots,\Delta\tau_K^*)^\top$. Define regression on the shared covariates as 
    $$g^*(X_{S_\cap}):=\E_P[\Delta\tau^*(X)\mid X_{S_\cap}],\quad \psi_P(\lambda):=\E_P[\exp(\lambda^\top g^*(X_{S_\cap}))]\in (0,\infty].$$
    Assume that (i) there is no $v\neq 0$ such that $v^\top g^*(X_{S_\cap})\le 0$ almost surely under $P$, and (ii) the infimum of $\lambda\mapsto\psi_P(\lambda)$ is not attained on the boundary of $\{\lambda:\psi_P(\lambda)<\infty\}$. 
\end{assumption}

\begin{lemma}[Existence and uniqueness of the transfer  weights]\label{propo:exist-n-unique-lm-based}
Under \cref{assump:heterogeneity-lm-based}, there exists a unique, finite $\lambda^*\in\R^{K-1}$ such that 
\begin{equation}
    \label{eq:def-lambda-lm}
    \lambda^* =\argmin_{\lambda\in \R^{K-1}}\ \psi_P(\lambda).
\end{equation}
 Furthermore, the problem \eqref{adj-sets-popln-opt-lm-based} has a unique ($P$-a.s.) solution given by the exponential tilt:
    $$w^*(X_{S_\cap})={\exp\big(g^*(X_{S_\cap})^\top\lambda^*\big)}\;\big/\;\psi_P(\lambda^*).$$
\end{lemma}

\subsection{Estimation and inference}

In light of \cref{propo:exist-n-unique-lm-based}, solving the problem \eqref{adj-sets-popln-opt-lm-based} reduces to solving the convex problem \eqref{eq:def-lambda-lm}, for which we need to estimate the conditional means $g_k^*(X_{S_\cap})=\E_P[\tau^*(X_{S_1})-\tau^*(X_{S_k})\mid X_{S_\cap}]$. It turns out that \cref{assump:heterogeneity-lm-based,assump:linear-model} are not enough to enable us to estimate $g(X_{S_\cap})$ at parametric rate, which is why we further impose the following modeling assumption.

\begin{assumption}\label{assump:linear-model-further}
    For each $k=1,\dots,K$, there exist constants $a_k$ and $b_k$ such that $$\E_P[\tau^*(X_{S_k}) \mid X_{S_\cap}]=a_k +  X_{S_\cap}^\top b_k.$$
\end{assumption}

We describe in \cref{algo:rewt-adj-sets-lm-based} below a simplified version of our general procedure (\cref{algo:rewt-adj-sets}) proposed in \cref{sec:est-and-infer} that operates under \cref{assump-iid,assump:heterogeneity-lm-based,assump:linear-model,assump:linear-model-further}. We retain the name \emph{specification-robust} for this model-based version of our estimator, as it remains robust without us knowing which regression adjustment set $S_j$ satisfies \cref{assump:linear-model}.

The following theorem states the asymptotic distribution of the specification-robust linear model-based estimator $\wh{\tau}^\mathtt{SR,lm}$ we compute in \cref{algo:rewt-adj-sets-lm-based}; see
\cref{app:proof-of-thm-lm-based} 
for a proof. It is important to note that this variant of our specification-robust procedure does not require any bias-correction. However, it relies on the stronger assumption that at least one of the linear models is correctly specified (\cref{assump:linear-model}) and that the conditional contrasts satisfy \cref{assump:linear-model-further}.

\begin{theorem}[Asymptotic normality of the  specification-robust estimator, parametric variant]\label{thm:matching.adj.sets-lm-based}
        Suppose that \cref{assump-iid,assump:heterogeneity-lm-based,assump:linear-model,assump:linear-model-further} hold, and assume that $\E_P[\exp(\rc\|X_{S_\cap}\|)]<\infty$ for some $\rc > 2\|\sum_{k=2}^K\lambda_k^*(b_1-b_k)\|$ (i.e., twice the slope of the affine map $X_{S_\cap}\mapsto g^*(X_{S_\cap})^\top \lambda^*$). Then,
        $$\sqrt{n}\left(\wh{\tau}^{\,\mathtt{SR,lm}}-\tauR\right)\dto \normal(0, \mathrm{V}_\mathtt{SR,lm}),$$ for some $\mathrm{V}_\mathtt{SR,lm}<\infty$. An explicit formula of $\mathrm{V}_\mathtt{SR,lm}$ is deferred to \eqref{eqn:VARlm}. 
\end{theorem}

\begin{algorithm}[t]
    \caption{Parametric variant of the proposed specification-robust procedure}
    \label{algo:rewt-adj-sets-lm-based}
    \vspace{2mm}

    \begin{algorithmic}[1] 
    \Require Data $(Y_i,\trt_i,X_i)_{i=1}^n$; $K\ge 2$ adjustment sets $S_1,\dots,S_K$ with intersection $S_\cap\neq \emptyset$.\vspace{1mm}

    \Ensure An estimator $\wh{\tau}^\mathtt{SR,lm}$ of the ATE of the reweighted population defined in \cref{propo:reweighted-estimand-makes-sense-lm-based}.
        \vspace{2mm}

\item Estimate $\tau^*(X_{S_k})$ using the linear regression of $Y$ on treatment $A$, covariates $X_{S_k}$ and the treatment-covariate interactions.

    \item Run linear regression of $\wh{\tau}^*(X_{S_k})$ on $X_{S_\cap}$ and obtain $\wh{\E}_P[\wh{\tau}^*(X_{S_k})\mid X_{S_\cap}]=\wh{a}_k+X_{S_\cap}^\top \wh{b}_k$.
    \item Compute $\wh{g}_k(X_{S_\cap})=\wh{a}_1-\wh{a}_k + X_{S_\cap}^\top (\wh{b}_1-\wh{b}_k)$, $k=2,\dots,K$, and solve the convex  problem: $$\wh{\lambda}:=\argmin_{\lambda\in\R^{K-1}}\, \P_n \left[\exp\left(\wh{g}(X_{S_\cap})^\top \lambda \right)\right].$$
    \item Compute the transfer weight for the $i$-th unit as $$\wh{w}(X_{i,S_\cap})={\exp\left(\wh{g}(X_{i,S_\cap})^\top \wh{\lambda}\right)}\,\Big/\,{\P_n\left[ \exp\left(\wh{g}(X_{S_\cap})^\top \wh\lambda\right)\right]}.$$
    \item Compute the specification-robust estimator:
    $$\wh{\tau}^\mathtt{SR,lm} = \wh{a}_1 + \wh{b}_1^{\,\top} \P_n \left[\wh{w}(X_{S_\cap})X_{S_\cap} \right].$$

\end{algorithmic}
\end{algorithm}

\begin{corollary}[specification-robust confidence intervals based on linear regression]\label{corollary-lm-based}
    Given any consistent estimator $\wh{\mathrm{V}}_\mathtt{SR,lm}$ of the asymptotic variance
    $\mathrm{V}_\mathtt{SR,lm}$ in \cref{thm:matching.adj.sets-lm-based},
    $$\lim_{n\to\infty} P\left(\tauR\in \left[\wh{\tau}^\mathtt{SR,lm}-\frac{z_{\alpha/2}}{\sqrt{n}}\sqrt{\wh{\mathrm{V}}_\mathtt{SR,lm}},\ \wh{\tau}^\mathtt{SR,lm}+\frac{z_{\alpha/2}}{\sqrt{n}}\sqrt{\wh{\mathrm{V}}_\mathtt{SR,lm}}\,\right]\right)=1-\alpha,$$ 
    where $z_{\alpha}$ is the $(1-\alpha)$-th quantile of the standard normal distribution.
\end{corollary}

\begin{remark}[Variance estimation]\label{remark:AR-boot-lm-case}
    We recommend estimating $\mathrm{V}_\mathtt{SR,lm}$ either using the plug-in estimates in the explicit formula \eqref{eqn:VARlm} (replacing $\theta^*$ by $\wh{\theta}$ and expectations by
empirical averages), 
    or using the following bootstrap procedure. 
    \begin{enumerate}
        \item Choose a large number $B=B(n)$ and for $b=1,\dots,B$,  sample $D_1^b,\dots,D_n^b$ with replacement from the data $\{D_1,\dots,D_n\}$ where $D_i=(Y_i,\, X_i,\, A_i)$.
        \item Compute the specification-robust linear model-based estimator $\wh{\tau}^{\mathtt{SR,lm}}_{b}$ using the resampled data and compute $$\wh{\mathrm{V}}_\mathtt{SR,lm}=\frac{n}{B}\sum_{b=1}^B \left(\wh{\tau}^{\mathtt{SR,lm}}_{b}-\wh{\tau}^{\mathtt{SR,lm}}\right)^2.$$
    \end{enumerate}
    Since our proof of \cref{thm:matching.adj.sets-lm-based} is based on asymptotic linear expansions, this bootstrap procedure can be theoretically justified under suitable conditions. We demonstrate the empirical performance of the bootstrap confidence intervals constructed using \cref{remark:AR-boot-lm-case} for our simulation examples in \cref{sec:simulations}.
\end{remark}

\begin{remark}[Affine weights are redundant]
    It follows from the empirical estimating equations that the estimators $(\wh{a}_k+\wh{b}_k^{\,\top} \P_n[\wh{w}(X_{S_\cap})X_{S_\cap}])$ computed in \cref{algo:rewt-adj-sets-lm-based} are numerically the same for each $k$. This is why the final step in \cref{algo:rewt-adj-sets-lm-based} computes the specification-robust estimate only using the reference adjustment set, instead of an affine aggregation as in \cref{algo:rewt-adj-sets}. 
\end{remark}

\section{Deferred results and details}

\subsection{Distribution shift in the simulation examples}\label{sec:sim-shift}

In this section, we provide additional numerical details on the simulation examples discussed in \cref{sec:simulations}, including the explicit transfer weights and the reweighted ATE. In both examples, we have:
\begin{itemize}
    \item $K=2$ candidate adjustments, namely $\{X_1\}$ and $\{X_1,X_2\}$.
    \item $X_1\sim \normal(0,1)$, and $X_2\indep X_1$ and $X_2$ is also Gaussian with mean zero.
    \item The conditional distribution of $(X_2, A, Y(0), Y(1))$ given $X_1$ remains unchanged under the reweighting, since $w=w(X_1)$ and $X_2\indep X_1$.
    \item Both \cref{assump:heterogeneity-lm-based,assump:linear-model-further} hold true. Write $g^*(X_1)=\alpha+\beta X_1$ where $\alpha=a_1-a_2$ and $\beta=b_1-b_2\neq 0$, with $a_1,a_2,b_1,b_2$ as defined in \cref{assump:linear-model-further}.
\end{itemize}

The following result shows that the reweighting in both \cref{example-1,example-2} results in a location shift for $X_1$; see \cref{proof:lemma:simplify} for a proof.

\begin{lemma}\label{lemma:simplify}
    Suppose that $K=2$ and $g^*(X_1)=\alpha+\beta X_1$ with $\beta\neq 0$. Then, for $X_1\sim\normal(0,1)$,
    \[\lambda^*=-\frac{\alpha}{\beta^2},\quad w^*(X_1)=\exp\left(cX_1-\frac{c^2}{2}\right),\quad c :=\lambda^*\beta = -\frac{\alpha}{\beta}.\]
    Thus, $X_1\sim\normal(c,1)$ under $w^*P$. Moreover, the reweighted ATE is given by $\tauR=a_1+b_1c=a_2+b_2c$ where $a_k,b_k$ are as defined in \cref{assump:linear-model-further}.
\end{lemma}

In light of \cref{lemma:simplify}, what remains is the computation of the regression coefficients $a_1,a_2,b_1,b_2$. Doing so leads to explicit formula for the transfer weights and reweighted ATE in \cref{example-1,example-2} which we provide in the following results, respectively; see \cref{proof:eg1-details,proof:eg2-details} for proofs.

\begin{lemma}\label{eg1-details}
    In \cref{example-1}, it holds that
    \( g^*(X_1)=\frac{3p}{1-p^2} + \frac{5p^2}{1-p^2}X_1\),
    where $p=\E_P[X_1\mid A=1]=\E_P[X_2\mid A=1]$. Consequently,
    \begin{align*}
        \lambda^* = -\frac{3(1-p^2)}{25p^3},\quad
        w^*(X_1)=\exp\left(-\frac{3}{5p}X_1-\frac{9}{50p^2}\right), \quad c = -\frac{3}{5p},\qquad \tauR = 1-\frac{3}{5p}.
    \end{align*}
\end{lemma}

The constant $p$ in \cref{eg1-details}, equivalently expressed as $p=\E_P\big[2X_1/(1+e^{5X_1+5X_2})\big]$, has no closed form. Numerical approximations give $p\approx -0.547$, and thus $c\approx 1.097$ and $\tauR\approx 2.097$.

\begin{lemma}\label{eg2-details}
    In \cref{example-2}, we have 
    \(g^*(X_1)=(4\sqrt{\pi}+6X_1)/(2\pi-3)\).
    Consequently,
    \begin{align*}
        \lambda^* = -\frac{\sqrt{\pi}(2\pi-3)}{9},\quad
        w^*(X_1)=\exp\left(-\frac{2\sqrt{\pi}}{3}X_1-\frac{2\pi}{9}\right), \quad c = -\frac{2\sqrt{\pi}}{3},\qquad \tauR = 1-\frac{\sqrt{\pi}}{3}.
    \end{align*}
\end{lemma}

\subsection{Further details on implementation}\label{sec:implementation}

\subsubsection{Parametric baseline}
For the simulation examples discussed in \cref{sec:simulations}, we estimate the conditional response functions using linear regression with treatment-covariate interactions, and construct specification-robust confidence intervals via the nonparametric bootstrap, as outlined in \cref{corollary-lm-based} and \cref{remark:AR-boot-lm-case}. 

As outlined in \cref{sec:lm-based}, this procedure does not require additional bias-correction for the transfer weights, but it does rely on the stronger assumption that at least one of the candidate linear models is correctly specified. It is also worthwhile to note that \cref{algo:rewt-adj-sets-lm-based} and \cref{thm:matching.adj.sets-lm-based} do not require cross-fitting.

\subsubsection{Nonparametric baseline}\label{sec:nonparametric-baseline} 
For the 401(k) application discussed in \cref{sec:401k}, we estimate the conditional response functions and propensity scores using generalized random forests (\texttt{grf} package; \citealp{Wager2018,Athey2019grf}), and use the plug-in variance estimator of \cref{rem:plug-in-est} below to construct specification-robust confidence intervals. 

\begin{remark}[Plug-in estimator of variance]\label{rem:plug-in-est}
    In practice, we recommend using a plug-in estimator of the asymptotic variance $\mathrm{V}_\nam$ in \cref{thm:matching.adj.sets}, as follows.
    \begin{multline*}
\wh{\mathrm{V}}_\nam:=\wh\var\left[\wh{w}(X_{S_\cap})\left(\sum_{k=1}^K\wh{\nu}_k\haipw{k}-\wh\tau^{\nam}_{bc}\right)\right.\\
+\left.\wh\lambda^\top \wh{w}(X_{S_\cap})(\Delta\haipw{}-\wh{g}(X_{S_\cap}))\left(\sum_{k=1}^K \wh{\nu}_k \wh{\E}_P\left[\wh\tau(X_{S_k})\mid X_{S_\cap}\right]-\wh\tau^\nam_{bc}\right)\right],
\end{multline*}
where $\wh{\var}$ denotes empirical variance.
\end{remark}

This complements our construction of an asymptotically valid confidence interval in \cref{cor:CI-from-Theorem-1}, and the resulting confidence interval demonstrates robust empirical performance, providing nominal coverage in our simulation examples in \cref{sec:applications}.

\subsubsection{Protection extension}\label{supp:constraints}

In \cref{sec:protect-covariates}, we discuss an extension of our specification-robust procedure which enables the researcher to protect (known) functions of the common covariates, with the population-level optimization problem as described in \eqref{eq:KL-protect}.
In this section, we provide the relevant algorithm (a modification of \cref{algo:rewt-adj-sets}) that solves the optimization problem \eqref{eq:KL-protect} empirically, and computes the bias-corrected specification-robust estimator $\htauRp$ incorporating the additional constraints.  \cref{thm:matching.adj.sets-protect-covariates} states the asymptotic normality of this estimator but defers an explicit expression for the asymptotic variance, which we provide in the following remark; see \cref{proof:thm:matching.adj.sets-protect-covariates} for a proof.

\begin{algorithm}[t]
    \caption{Specification-robust estimator for the ATE with additional covariate constraints (using externally fit nonparametric regressors)}
    \label{algo:rewt-adj-sets-protect-covariates} 
    \vspace{2mm}
    \begin{algorithmic}[1] 
    \Require Data $(Y_i,\trt_i,X_i)_{i=1}^n$; $K\ge 2$ adjustment sets $S_1,\dots,S_K$ with non-empty intersection $S_\cap$;
    an $\R^d$-valued function $f$ of the  covariates that we want to protect as in \eqref{eq:KL-protect}; propensity score estimators $\wh{e}(X_{S_k})$ of $e(X_{S_k})=P(\trt=1\mid X_{S_k})$, conditional response function estimators $\wh\mu_{a}(X_{S_k})$ of $\mu_{a}(X_{S_k})={\E}_P[Y\mid \trt=a,X_{S_k}]$, and  regressors $\wh{\E}[\wh\tau(X_{S_k})\mid X_{S_\cap}]$ estimating $\E_P[\tau(X_{S_k})\mid X_{S_\cap}]$ for $a=0,1$, and $k=1,\dots,K$.
    \vspace{2mm}
    \Ensure An estimator $\htauRp$ of the ATE $\tauRp$ of the reweighted population $w^*P$ defined via \eqref{eq:KL-protect}.
        \vspace{2mm}
\item Set $\wh{\tau}(X_{S_k})=\wh\mu_1(X_{S_k})-\wh\mu_0(X_{S_k})$ and estimate the constraint vector 
$\delta(X)$ defined in \eqref{eq:new-delta} using $$\wh{\delta}(X):=\left(\wh{\tau}(X_{S_1})-\wh{\tau}(X_{S_2}),\,\dots,\,\wh{\tau}(X_{S_1})-\wh{\tau}(X_{S_K}),\,f(X_{S_\cap})-\P_n[f(X_{S_\cap})]\right).$$ 

\item Use nonparametric regressors to compute $\wh{g}_k(X_{S_\cap}):=\wh{\E}[\wh\delta_{k}(X)\mid X_{S_\cap}]$, for $ k=1,\dots,K-1$; set $\wh{g}_{K:(K-1+d)}(X_{S_\cap}):=f(X_{S_\cap})-\P_n[f(X_{S_\cap})]$.
\item Find transfer weights using the above choice of $\wh{g}$, obtain $\wh{\lambda}$ as a by-product, and compute reweighted AIPW estimators $\wh{\tau}_k^{\,\mathtt{R}}$ as in \cref{algo:rewt-adj-sets}.

\item Find weights $\wh{\nu}$ as in \cref{algo:rewt-adj-sets} using the above $\wh{g}$, i.e., set $\wh\nu_1 := 1-\sum_{k=2}^K \wh\nu_{k}$ and $$\wh\nu_{2:(K+d)} := \left(\P_n\left[ \wh{w}(X_{S_\cap})\wh{g}(X_{S_\cap})\wh{g}(X_{S_\cap})^\top\right]\right)^{-1}\P_n\left[\wh{w}(X_{S_\cap})\wh{g}(X_{S_\cap})\left(\wh{\E}_P[\wh\tau(X_{S_1})\mid X_{S_\cap}]-\wh\tau_1^\mathtt{R} \right)\right].$$

\item With $\wh\delta^\mathtt{AIPW}_k:=\haipw{1}-\haipw{k+1}$ for $1\le k\le K-1$, and $\wh\delta^\mathtt{AIPW}_k:=\wh{g}_k(X_{S_\cap})$ for $K\le k\le K-1+d$, compute a bias-correction term as
\begin{equation*}
    \begin{split}
B_n:=\wh{\lambda}^\top\,\P_n\left[\wh{w}(X_{S_\cap})\left(
\wh\delta^\mathtt{AIPW}(Y,A,X)-\wh{g}(X_{S_\cap})\right)\left(\wh{\E}\left[\wh\tau(X_{S_1})\mid X_{S_\cap}\right]-\wh\tau_1^\mathtt{R} - \wh{g}(X_{S_\cap})^\top \wh{\nu}_{2:(K+d)} \right)\right]\\
\,-\, \P_n\left[\wh{w}(X_{S_\cap})\left(f(X_{S_\cap})-\P_n [f(X_{S_\cap})]\right)^\top\wh\nu_{(K+1):(K+d)}\right].
    \end{split}
\end{equation*}

\item With $\wh\nu$ and $B_n$ as defined above, compute $\htauRp:=\sum_{k=1}^K \wh\nu_k \wh{\tau}_k^\mathtt{\,R}  + B_n$. 
\end{algorithmic}
\end{algorithm}

\begin{remark}\label{rem:protect-covariates-variance}
The asymptotic variance $\mathrm{V}_{\namp}$ in \cref{thm:matching.adj.sets-protect-covariates} is given by
    \begin{multline*}
\mathrm{V}_{\namp}:=\var\Bigg[
 w^*(X_{S_\cap})\Bigg(\sum_{k=1}^K\nu^*_k\,\aipw{k}-\tauRp\\
 \qquad\qquad\qquad\qquad+\lambda^{*,\top} (\daipw{}-g(X_{S_\cap}))\left( \E_P\left[\tau(X_{S_1})\mid X_{S_\cap}\right]-\tauRp-g(X_{S_\cap})^\top \nu_{2:(K+d)}^{*}\right)\Bigg)\\
-\left(w^*(X_{S_\cap})-1\right)(f(X_{S_\cap}) - \E_P[f(X_{S_\cap})])^\top \nu_{(K+1):(K+d)}^*\Bigg],
\end{multline*}
where the vector $\daipw{}$ is the same as the constraint vector $\delta(X)$ defined in \eqref{eq:new-delta} except with contrasts $\tau(X_{S_k})$ replaced with their AIPW versions $\aipw{k}$ (cf.~\cref{thm:matching.adj.sets}), and the weights $\nu^*$ are such that $\nu^*_1=1-\sum_{k=2}^K\nu^*_k$ and
$$\nu^*_{2:(K+d)} := \left(\E_P\left[ w^*(X_{S_\cap})\,g(X_{S_\cap})\,g(X_{S_\cap})^\top\right]\right)^{-1}\E_P\left[ w^*(X_{S_\cap})\,g(X_{S_\cap})\left(\tau(X_{S_1})-\tauRp \right)\right].$$
\end{remark}

\subsection{Uniform validity}\label{sec:uniform-validity}

In this section, we establish that the specification-robust confidence intervals from \cref{sec:est-and-infer} are uniformly valid over a suitable class of distributions, provided the regularity conditions hold uniformly.
First, we set a general notation for the model class where we do not know which adjustment set is valid:
\begin{equation*}
    \sP_\nam := \bigcup_{k=1}^K \sP_{k},\quad\text{where}\quad\sP_k:= \{P: {S_k} \text{ is a valid adjustment set under }P\}.
\end{equation*}
In other words, $\sP_k$ is the set of all distributions $P$ of $(Y_i(0),Y_i(1), A_i,X_i)$ (cf.~\cref{assump-iid}) such that the ignorability condition $A\indep Y(a)\mid X_{S_k}$ $(a=0,1)$ holds under $P$. The substantive difference from the classes introduced in \cref{sec:motivation} is that here the observed data law is allowed to change. We  view the average treatment effect (ATE) as a functional: 
\begin{equation}
    \label{eq:ATE-functional}
    \overline\tau(P):=\E_P[Y(1)-Y(0)].
\end{equation}
On $\sP_\nam$, the original ATE $\overline{\tau}(P)$ is not point-identified unless the researcher knows which of the $K$ adjustment sets is (are) valid. We shift the goalpost and set the following as the target estimand: $$\tauR(P):=\overline\tau(w^*P),$$ where the reweighted population $w^*P$ is the closest---in KL divergence---to the original population where all candidate adjustment sets agree (see \cref{sec:theory} for details). Thus, even the `pointwise' results derived in \cref{sec:theory} have a sense of uniform validity:~While the standard ATE estimators are only valid on a single $\sP_k$, our approach is valid on the union class $\sP_\nam$ (modulo our regularity conditions). In the rest of this section, we derive
results that are uniformly valid over a restricted class $\sP\subset \sP_\nam$, under suitable conditions, as follows.

\begin{assumption}
    [Uniform version of \cref{assump-iid}]\label{assump-iid-uniform}
     $(Y_i(0),Y_i(1),\trt_i, X_i)$, $i=1,\dots,n$, are independent observations from an unknown distribution $P$ that belongs to the class $\sP$.
\end{assumption}

\begin{assumption}
    [Uniform version of \cref{assump-at-least-one-is-valid}]\label{assump-at-least-one-is-valid-uniform}  Assume that $\sP\subseteq \sP_\nam$ so that under every $P\in\sP$, at least one of the candidate adjustment sets is valid. Assume further that the intersection $S_{\cap}$ is non-empty and that we have strong overlap:~For every $P\in\sP$, and every $1\le k\le K$, it holds that $\xi\le P(A=1\mid X_{S_k})\le 1-\xi $ almost surely, for some $\xi\in (0,1)$ independent of $P$. 
\end{assumption}

Similar to \eqref{eq:ATE-functional}, denote contrasts as $\tau(X_{S_k};P):=\E_P[Y\mid X_{S_k}, A=1]-\E_P[Y\mid X_{S_k}, A=0]$, and differences-in-contrasts as $\Delta\tau_k(X;P):=\tau(X_{S_1};P)-\tau(X_{S_k};P)$ for $k=2,\dots,K$, and write $\Delta\tau:=(\Delta\tau_2,\dots,\Delta\tau_K)^\top$. Also define    \begin{equation}
        \label{last-def-of-g}
    g(X_{S_\cap};P):=\E_P[\Delta\tau(X;P)\mid X_{S_\cap}],\quad \psi_P(\lambda):=\E_P[\exp(\lambda^\top g(X_{S_\cap}; P))]\in (0,\infty].
    \end{equation} Furthermore, denote by $Q_{\lambda,P}$ the exponentially tilted distribution with density proportional to $\exp(\lambda^\top g)$ w.r.t.~$P$, i.e.,
    \begin{equation}
        \label{eq:exp-tilt}
\frac{dQ_{\lambda,P}}{dP}:=\exp(\lambda^\top g(X_{S_\cap};P))\;\Big/\;\E_P\left[\exp(\lambda^\top g(X_{S_\cap};P))\right],
    \end{equation}
    and recall from \cref{propo:exist-n-unique} that the solution to the problem \eqref{adj-sets-popln-opt}, unique $P$-a.s., is given by
    $w^*(X_{S_\cap};P)=dQ_{\lambda^*(P),\, P}/dP$ where $\lambda^*(P)$ is defined in \eqref{def:lambda-uniform} below.
\begin{assumption}
    [Uniform version of \cref{assump:heterogeneity}]\label{assump:heterogeneity-uniform} 
    There exist constants $\zeta>0$ and $c_0\in(0,1]$ such that, for every $P\in\sP$ and every unit vector $v\in\mathbb R^{K-1}$,
\[
P\big(v^\top g(X_{S_\cap};P)\ge\zeta\big)\ge c_0.
\]

\end{assumption}

 The following lemma mirrors \cref{propo:exist-n-unique}. It states the existence and uniqueness of a finite minimizer $\lambda^*(P)$ of the moment generating function $\psi_P$ defined in \eqref{last-def-of-g}, and also shows that $\|\lambda^*(P)\|$ remains uniformly bounded; see \cref{proof:propo:exist-n-unique-uniform} for a proof. 

\begin{lemma}[Uniform extension of \cref{propo:exist-n-unique}]\label{propo:exist-n-unique-uniform}
Under \cref{assump:heterogeneity-uniform}, for every $P\in\sP$ the moment generating function $\psi_P$ defined in \eqref{last-def-of-g} has a unique, finite minimizer $\lambda^*(P)$, i.e.,   
\begin{equation}
    \label{def:lambda-uniform}
    \lambda^*(P) =\argmin_{\lambda\,\in\, \R^{K-1}} \psi_P(\lambda).
\end{equation}
 Moreover, $\sup_{P\in\sP} \|\lambda^*(P)\|\le \log (1/c_0)/\zeta <\infty$ where $c_0$ and $\zeta$ are constants as in \cref{assump:heterogeneity-uniform}.
\end{lemma}

It is worth noting that \cref{assump:heterogeneity-uniform} does not include an analogue of \cref{assump:heterogeneity} (ii) (that $\lambda^*(P)$ lies in the interior of the effective domain $\{\lambda:\psi_P(\lambda)<\infty\}$), since  \cref{propo:exist-n-unique} already makes the connection clear, and the main purpose of this section is to derive uniform validity---which is the content of \cref{thm:matching.adj.sets-uniform} below; see \cref{proof:thm:matching.adj.sets-uniform} for a proof. Indeed, condition (4) of \cref{thm:matching.adj.sets-uniform} states that $\psi_P$ is finite in a radius $r>\sup_{P\in\sP}\|\lambda^*(P)\|$ around the origin, which implies that $\lambda^*(P)$ lies in the interior of $\{\lambda:\psi_P(\lambda)<\infty\}$ for every $P\in\sP$ and thus the reweighting $w^*(X_{S_\cap};P)$ as defined above is the unique ($P$-a.s.) solution to \eqref{adj-sets-popln-opt} for every $P\in\sP$, due to \cref{propo:exist-n-unique}.

\begin{theorem}
       [Uniform version of \cref{thm:matching.adj.sets}]\label{thm:matching.adj.sets-uniform}
       Assume that the propensity score estimators $\wh{e}(X_{S_k})$, the conditional response function estimators $\wh\mu_a(X_{S_k})$ and the iterated regression estimators $\wh{\E}_P[\wh\tau(X_{S_k})\mid X_{S_\cap}]$ used in \cref{algo:rewt-adj-sets} are fit independent of the data $(Y_i,A_i,X_i)_{i=1}^n$, and that \cref{assump-iid-uniform,assump-at-least-one-is-valid-uniform,assump:heterogeneity-uniform} hold true. Assume further that,
    \begin{enumerate}[label=(\arabic*)]

    \item (uniform rates of the nonparametric regressors) for $a=0,1$ and $k=1,\dots,K$,
    \begin{align*}
       & \sup_{P\in\sP}\|\wh{\mu}_a(X_{S_k})-\mu_a(X_{S_k})\|_{L^{q}(\P_n)}=\o(n^{-\alpha_\mu}), \\
       &\sup_{P\in\sP}\|\wh{e}\,(X_{S_k})-e(X_{S_k})\|_{L^{q}(\P_n)}=\o(n^{-\alpha_e}),\\
       & \sup_{P\in\sP}\|\wh{\E}_P[\wh{\tau}(X_{S_k})\,|\, X_{S_\cap}]-\E_P[{\tau}(X_{S_k})\,|\, X_{S_\cap}]\|_{L^{q}(\P_n)}=\o(n^{-1/4}),
    \end{align*}
   for some $q>2$, $\alpha_\mu\ge 0$, $\alpha_e\ge 0$ such that $\alpha_\mu+\alpha_e\ge \frac{1}{2}$.

\item (strong overlap)  The estimated propensity scores $\wh{e}\,(X_{S_k})$ are bounded between $[\xi,1-\xi]$ for some $0<\xi<1$, almost surely under every $P\in\sP$.

    \item (uniformly bounded moments) $\sup_{P\in\sP}\E_P|Y|^q<\infty$, and $\sup_{P\in\sP}\var_P[Y\mid A=a,X_{S}]\le C$ a.s.~for $S=S_\cap,S_1,\dots,S_K$, for some $C\in (0,\infty)$ that does not depend on $P$,

        \item (finite exponential moments) The moment generating functions of $\E_P[\tau(X_{S_1};P)\mid X_{S_\cap}]$ and $g(X_{S_\cap})$ are uniformly bounded in a radius $\rc > \frac{2q}{q-2}\sup_{P\in\sP}\|\lambda^*(P)\|$ around the origin, and the same holds for $\wh{g}$ in the conditional sense:~$\{\E_P[\exp(r\|\wh{g}(X_{S_\cap})\|)\mid \wh{g}(\cdot)]\}_{P\in\sP}$ is tight.
    \end{enumerate}  Then, the specification-robust estimator $\wh\tau^\nam_{bc}$ we compute in \cref{algo:rewt-adj-sets}  satisfies
    \begin{equation}\label{eq:ALE-first}
  \wh\tau^{\,\nam}_{bc} - \tauR(P)
  \;=\;
  \P_n[\psi^{\,\nam}(D;P)]
  \;+\;
  r_n(P),
\end{equation}
where the remainder term is uniformly asymptotically negligible: $\sup_{P\in\sP}\left|r_n(P)\right|=\o(n^{-1/2})$, and the influence function is given by 
    \begin{equation}\label{eq:IF-first}
\begin{split}
     \psi^{\,\nam}(D;P)
  &:=
  w^*(X_{S_\cap};P)
  \,\Bigg[
    \sum_{k=1}^K \nu^*_k
    \bigl(\tauAIPW{k} - \tauR(P)\bigr)\\
    &\qquad\qquad\qquad\qquad\qquad+
    (\lambda^*)^\top
    \bigl(\Delta\tauAIPW{} - g(X_{S_\cap};P)\bigr)
    \phi(X_{S_\cap};\nu^*(P),P)
  \Bigg],
\end{split}
\end{equation}
where $\phi(X_{S_\cap};\nu,P):=\sum_{k=1}^K\nu_k\,\E_P[\tau(X_{S_k};P)\mid X_{S_\cap}]-\tauR(P)$, with affine weights $\nu^*(P)$ as defined in \cref{thm:matching.adj.sets}.
Moreover, if the asymptotic variance is uniformly bounded from below, i.e., if $\inf_{P\in\sP}\var_P(\psi^\nam(D;P))\ge v_0$ for some constant $v_0>0$, then 
$$\lim_{n\to\infty}\sup_{P\in\sP}\,\sup_{t\in\R}\left| P\left(\frac{\sqrt{n}\left(\wh\tau^\nam_{bc}-\tauR(P)\right)}{\sqrt{\mathrm{V}_\nam(P)}}\le t\right) -\Phi(t)\right|=0,$$
where the asymptotic variance $\mathrm{V}_\nam(P)=\var_P(\psi^\nam(D;P))$ is the same as in \cref{thm:matching.adj.sets}.
\end{theorem}

Recall that we construct an $(1-\alpha)$-confidence interval for $\tauR$ based on $\wh\tau^\nam_{bc}$ as follows.
\begin{equation}
    \label{eq:CI-from-AR-recall}
    \mathcal{I}_\alpha:=\left[\wh\tau^\nam_{bc}-\frac{z_{\alpha/2}}{\sqrt{n}}\sqrt{\wh{\mathrm{V}}_\nam},\ \wh\tau^\nam_{bc}+\frac{z_{\alpha/2}}{\sqrt{n}}\sqrt{\wh{\mathrm{V}}_\nam}\right],
\end{equation}
where $z_\alpha$ is the $(1-\alpha)$-th quantile of standard normal distribution and $\wh{\mathrm{V}}_\nam$ is an estimator of the asymptotic variance $\mathrm{V}_\nam(P)$. It is straightforward from \cref{thm:matching.adj.sets-uniform} that the above confidence intervals are uniformly asymptotically valid provided that the consistency of the variance estimator is uniform over the class. We state this formally in the following corollary.

\begin{corollary}[Uniform validity of specification-robust confidence intervals]\label{cor:CI-from-Theorem-uniform}
    Suppose that $\wh{\mathrm{V}}_\nam$ is a uniformly consistent estimator of the asymptotic variance $\mathrm{V}_\nam(P)$ in \cref{thm:matching.adj.sets-uniform}, in the sense that $$\sup_{P\in\sP} P(|\wh{\mathrm{V}}_\nam/\mathrm{V}_\nam(P)-1|>\eps)\to 0\ \text{ for every }\ \eps>0.$$ Assume that the class $\sP\subset\sP_\nam$ satisfies the conditions of \cref{thm:matching.adj.sets-uniform}. Then, the $(1-\alpha)$ confidence interval $\mathcal{I}_\alpha$ for $\tauR$ constructed in \eqref{eq:CI-from-AR-recall} satisfies
       $$\lim_{n\to\infty} \sup_{P\in\sP}\left|P(\tauR\in \mathcal{I}_\alpha)-(1-\alpha)\right|=0.$$
\end{corollary}

\section{Main proofs}

\subsection{Proof of \texorpdfstring{\cref{propo:reweighted-estimand-makes-sense}}{Lemma 1 (interpreting the new estimand)}}\label{app:proof-of-propo:reweighted-estimand-makes-sense}

\begin{proof}
    It follows from the construction of the (population) weights $w^*$ that $$\E_P[w^*(X_{S_\cap})\tau(X_{S_1})]=\E_P[w^*(X_{S_\cap})\tau(X_{S_2})]=\dots=\E_P[w^*(X_{S_\cap})\tau(X_{S_K})]=:\tauR.$$
    Now \cref{assump-at-least-one-is-valid} posits that at least one of the adjustment sets, say $S_{j}$, is valid. Using the ignorability condition $Y(0), Y(1)\indep A\mid X_{S_j}$, we deduce the following.
    \begin{align*}
    \tauR &= \E_P[w^*(X_{S_\cap})\tau(X_{S_{j}})]\\
    &= \E_{P} [w^*(X_{S_\cap}) \left(\E_P[Y\mid \trt=1, X_{S_j}]-\E_P[Y\mid \trt=0,X_{S_j}]\right)]\tag{by definition}\\
    &= \E_{P} [w^*(X_{S_\cap}) \left(\E_P[Y(1)\mid \trt=1, X_{S_j}]-\E_P[Y(0)\mid \trt=0,X_{S_j}]\right)]\tag{consistency/SUTVA}\\
    &= \E_{P}[w^*(X_{S_\cap})(Y(1)-Y(0))]\tag{using ignorability and tower property}\\
    &= \E_{w^*P}[Y(1)-Y(0)],
    \end{align*}
    as desired to show.
\end{proof}

\subsection{Proof of \texorpdfstring{\cref{propo:exist-n-unique}}{Lemma (existence and uniqueness of optimal weights)}}\label{app:proof-of-propo:exist-n-unique}

\begin{proof}
Define $C:=\{\lambda:\psi_P(\lambda)\le 1\}$ where $\psi_P(\lambda)=\E_P[\exp(\lambda^\top g(X_{S_\cap}))]$. We first show that $C$ is compact. Note that $\psi_P$ is lower semicontinuous, since for any $\lambda_n\to \lambda$, Fatou's lemma gives
\[\psi_P(\lambda)=\E_P\left[\liminf_{n\to\infty} \exp(\lambda_n^\top g)\right]\le \liminf_{n\to\infty} \psi_P(\lambda_n).\]
Therefore $C$ is a closed set. 
We next show that $C$ is convex. Fix $\lambda_0,\lambda_1\in\R^{K-1}$ and $s\in [0,1]$. Convexity of the exponential map $x\mapsto \exp(x)$ yields that \[\exp\left(((1-s)\lambda_0+s\lambda_1)^\top g(X_{S_\cap})\right)\le (1-s)\exp(\lambda_0^\top g(X_{S_\cap})) + s\exp(\lambda_1^\top g(X_{S_\cap})),\]
pointwise. Both sides are nonnegative, so taking expectations, we get $\psi_P((1-s)\lambda_0+s\lambda_1)\le (1-s)\psi_P(\lambda_0)+s\psi_P(\lambda_1)$ for the possibly extended-real valued function $\psi_P$. This proves that $\psi_P$ is convex. The same argument shows that the effective domain $D=\{\lambda:\psi_P(\lambda)<\infty\}$ is also convex.

We next prove that $C$ is bounded, by contradiction.
If $C$ is not bounded, we can take $\lambda_n\in C$ such that $\|\lambda_n\|\to\infty$. Passing to a subsequence, we have $${\lambda_n}/{\|\lambda_n\|}\to v, \text{ for some $v\in\R^{K-1}$ with $\|v\|=1$}.$$ Fix $t>0$. Since $\|\lambda_n\|\to\infty$, we have $t<\|\lambda_n\|$ for all sufficiently large $n$. Since $C$ is convex and contains both $0$ and $\lambda_n$, we can say that $\frac{t}{\|\lambda_n\|}\lambda_n\in C$. This combined with $C$ being closed implies that $tv\in C$, i.e.,
$$\E_P[\exp(tv^\top g)]\le 1, \text{ for every }t>0.$$
On the other hand, we know that for some $\eps>0$,  $P(v^\top g(X_{S_\cap})> \eps)>0$ (since otherwise $v^\top g\le 0$ a.s.~under $P$). As a result, $$\E_P[\exp(tv^\top g(X_{S_\cap}))]\ge \exp(t\eps)\cdot P(v^\top g(X_{S_\cap})>\eps) \to \infty, \text{ as }t\to\infty.$$
This contradiction implies that $C$ must be bounded, and thus compact.

\smallskip

We now turn to the main proof. Since $\lambda\mapsto\psi_P(\lambda)$ is lower semicontinuous, it attains a minimum on $\{\lambda:\psi_P(\lambda)\le 1\}$. Moreover, $0\in C$ and $\psi_P(0)=1$, so we conclude that the infimum of $\psi_P$ lies in the closed set $C$. In other words, $\psi_P$ attains a global minimum at some $\lambda^*\in C$. Since $C$ is compact, $\lambda^*$ is finite.

To prove the uniqueness, assume for the sake of contradiction that there are two minimizers $\lambda_0\neq \lambda_1$. Since $\psi_P(0)=1<\infty$, both $\lambda_0$ and $\lambda_1$ lie in the effective domain $D=\{\lambda:\psi_P(\lambda)<\infty\}$. Invoking condition (i) with $v=\lambda_0-\lambda_1$ and $-v$, we can say that $P(v^\top g\neq 0)>0$. The strict convexity of $x\mapsto \exp(x)$ implies that $$\exp\left(\frac12(\lambda_0+\lambda_1)^\top g(X_{S_\cap})\right)\le \frac{1}{2}\left(\exp(\lambda_0^\top g(X_{S_\cap})) + \exp(\lambda_1^\top g(X_{S_\cap})) \right),$$
with strict inequality on the event $\{\lambda_0^\top g(X_{S_\cap})\neq \lambda_1^\top g(X_{S_\cap})\}=\{v^\top g\neq 0\}$. Since this event has positive probability, we conclude by taking expectations in the above display that
\[\psi_P\left(\frac{1}{2}(\lambda_0+\lambda_1)\right)<\frac{1}{2}\left(\psi_P(\lambda_0)+\psi_P(\lambda_1)\right)=\inf_{\lambda\in\R^{K-1}}\psi_P(\lambda), \]
which contradicts the minimality of $\lambda_0$ and $\lambda_1$. This proves that the minimizer must be unique.

\medskip

We now show why condition (i) in \cref{assump:heterogeneity} is also necessary for the existence of a unique, finite minimizer of $\psi_P$. Let, if possible, there be $v\neq 0$ such that $v^\top g(X_{S_\cap})\le 0$ a.s.~under $P$. If $P(v^\top g < 0)>0$, then $t\mapsto \psi_P(\lambda+tv)$ is strictly decreasing for $t\ge 0$ and therefore no finite minimizer exists.
When $v^\top g = 0$ a.s., note that for any minimizer $\lambda^*$,  $\psi_P(\lambda^*+tv)$ takes the same value for each $t$, showing the minimizer is not unique. We thus conclude that \cref{assump:heterogeneity} (i) is both sufficient and necessary for the existence of a unique, finite minimizer $\lambda^*$ of the function $\lambda\mapsto\psi_P(\lambda)$.

\medskip

The second conclusion follows from a standard result from information theory \citep{Csiszar1975,CoverThomas2005MaximumEntropy}; see also \citet[Theorem 5.2]{DonskerVaradhan1976}. For completeness, we provide a proof below. Define
\begin{equation}
    \label{wstar-def-temp}
    w^*(X_{S_\cap})=\exp\big(g(X_{S_\cap})^\top\lambda^*\big)/\psi_P(\lambda^*).
\end{equation}
By definition, $w^*\ge 0$ and $\E_P[w^*]=1$. 
We now show that $\E_P[w^* \|g\|]<\infty$ and $\E_P[w^* g]=0$. By \cref{assump:heterogeneity} (ii), 
 $\lambda^*\in \text{int }D$ where $D = \{\lambda:\psi_P(\lambda)<\infty\}$. Therefore, for any unit vector $u$, there exists $t_0>0$ sufficiently small such that $\lambda^*\pm t_0u\in \text{int }D$. For $t\in (0, t_0]$, define
 \[f_t(X_{S_\cap}) := \exp(g(X_{S_\cap})^\top \lambda^{*}) \frac{\exp(tu^\top g(X_{S_\cap}))-1}{t}.\]
Convexity of the map $t\mapsto \exp(tu^\top g)$ implies that $f_t$ is non-decreasing in $t$. As $t\downarrow 0$, $$f_t \,\downarrow\, f_0:=\exp(g^\top \lambda^*) u^\top g, \quad \text{pointwise}.$$ Since $\lambda^*+t_0u\in \text{int }D$, we can deduce that $f_{t_0}(X_{S_\cap})$ is integrable. Moreover, $f_t \le f_{t_0}$ by the monotonicity. Therefore we can apply the monotone convergence theorem for a non-increasing family bounded above by an integrable function to conclude that 
\[\E_P[f_0]=\lim_{t\,\downarrow\, 0}\E_P[f_t]  = \lim_{t\,\downarrow\, 0} \frac{\psi_P(\lambda^*+tu)-\psi_P(\lambda^*)}{t}\ge 0,\]
 where in the last step we used the minimality of $\lambda^*$. The above, combined with $f_0^+\le f_{t_0}^+$ and $f_{t_0}\in L^1$ implies that $f_0\in L^1$, i.e., $\E_P[w^* \|g\|]<\infty$. Applying the same arguments to $-u$ yields $\E_P[f_0]\le 0$. This gives $\E_P[f_0]=0$, i.e., $\E_P[w^* g]=0$.

 \medskip

 Thus, $w^*$ lies in the feasible set for the optimization problem \eqref{adj-sets-popln-opt}, given by $$\mathcal{S}=\left\{w=w(X_{S_\cap})\ge 0,\,\E_P[w]=1,\,\E_P[w\|g\|]<\infty,\,\E_P[wg]=0\right\}.$$ Fix any $w$ in this feasible set, and note that
\begin{align*}
    \dkl(wP\|P) &= \E_P[w\log w] \\
    &= \E_P\left[w\log \frac{w}{w^*}\right]+ \E_P[w\log w^*]\\
    &=\dkl(wP\|w^*P)+ \E_P[wg]^\top \lambda^* -\log \psi_P(\lambda^*)\cdot \E_P[w]\tag{using \eqref{wstar-def-temp}}\\
    &=\dkl(wP\|w^*P)+ \E_P[w^*g]^\top \lambda^* -\log \psi_P(\lambda^*)\cdot \E_P[w^*]\tag{since both $w,w^*\in\mathcal{S}$}\\
    &=\dkl(wP\|w^*P)+ \E_P[w^*\log w^*]\tag{again using \eqref{wstar-def-temp}}\\
    &=\dkl(wP\|w^*P)+\dkl(w^*P\|P)\\
    &\ge \dkl(w^*P\|P),
\end{align*}
with equality if and only if $wP=w^*P$. This finishes the proof.
\end{proof}

\subsection{Proof of \texorpdfstring{\cref{thm:matching.adj.sets}}{Theorem 1 (asymptotic distribution of our specification-robust estimator)}}\label{proof-of-main-CLT-with-bias-corr}

\begin{proof} The proof proceeds in the following steps.
First, we analyze how close the reweighted estimators $\wh{\tau}_k^{\,\mathtt{R}}$ computed in \cref{algo:rewt-adj-sets} are to their oracle versions; see \eqref{emp-minus-oracle-tau-k} for a quantitative statement. The next step is to consider the aggregated estimators and analyze the difference $\sum_k\wh\nu_k\wh\tau_k^{\,\mathtt{R}}-\sum_k\nu_k^*\wh\tau_k^{\,\mathtt{R},*}$. This difference alone is not asymptotically negligible---we need the bias-correction term to account for first-order bias from estimating the transfer weights. With the additional bias-correction in place, we obtain an asymptotic linear expansion of the (bias-corrected and aggregated) specification-robust estimator, which yields the desired limiting distribution.

\paragraph{Difference of $\wh\tau_k^{\,\mathtt{R}}$ and their oracle versions.}
Define oracle estimators of the reweighted ATE $\tauR$ as 
$$\wh{\tau}_k^{\,\mathtt{R},*}:=\tauR+\P_n \left[w^*(X_{S_\cap})\left( \aipw{k}-\tauR\right)\right],\quad k=1,\dots,K,$$ where $w^*$ is the population weight function that solves \eqref{adj-sets-popln-opt}, $\tauR$ is the ATE for the reweighted population as defined in \eqref{eq:def-tauR}, and the AIPW functionals are as defined in \cref{sec:neyman-orthogonality}. 
Since $\P_n [\wh{w}(X_{S_\cap})]=1$, we can deduce using algebra that
\begin{align*}
	\wh{\tau}_k^{\,\mathtt{R}}-\wh{\tau}_k^{\,\mathtt{R},*} &= \P_n \left[\wh{w}(X_{S_\cap})\left(\haipw{k}-\tauR\right)-w^*(X_{S_\cap})\left( \aipw{k}-\tauR\right)\right]\\[1mm]
	&=\P_n \left[\left(\wh{w}(X_{S_\cap})-w^*(X_{S_\cap})\right) \left( \aipw{k}-\tauR\right)\right]\\[1mm]
    &\qquad\qquad+\P_n \left[{w}^*(X_{S_\cap})\left(\haipw{k}- \aipw{k}\right)\right]\\[1mm]
	&\qquad\qquad+\P_n \left[\left(\wh{w}(X_{S_\cap})-w^*(X_{S_\cap})\right) \left(\haipw{k}- \aipw{k}\right)\right]\\[1mm]
    &=: \mathbf{I}_n + \mathbf{II}_n + \mathbf{III}_n.
\end{align*}
It follows from standard proof techniques for AIPW estimators that $\mathbf{II}_n=\o(n^{-1/2})$ as well as $\mathbf{III}_n=\o(n^{-1/2})$; we delegate these proofs to \cref{lemma:AIPW-approx}. It only remains to tackle the first term.
We apply the tower property to deduce that
\begin{align*}
    &\E_P\left[\aipw{k}-\tau(X_{S_k}) \mid X_{S_\cap}\right]\\[1mm]
    &=\E_P\left[\E_P\left[\frac{A(Y-\mu_1(X_{S_k}))}{e(X_{S_k})}-\frac{(1-A)(Y-\mu_0(X_{S_k}))}{1-e(X_{S_k})}\;\Big|\; X_{S_k}\right]\;\Big|\; X_{S_\cap}\right]
    \tag{since $S_{\cap}\subseteq S_k$}\\
    &=0.\tag{by definition}
\end{align*}
We now apply \cref{lemma.one} with $U_i=\aipw{k}$, $V_i=X_{i,S_\cap}$, $\wh{a}=\wh{w}$ and $a=w^*$. Note that $\var(U_i\mid V_i)\le 6C/\xi$ by \cref{cvarlemma}, and $\|\wh{a}-a\|_{L^2(\P_n)}=\o(1)$ by \cref{lemma:rate-of-weights}. Therefore, we can apply \cref{lemma.one} to conclude that 
$$\P_n\left[\left(\wh{w}(X_{S_\cap})-w^*(X_{S_\cap})\right) \left( \aipw{k}-\E_P[\tau(X_{S_k})\mid X_{S_\cap}]\right)\right]=\o(n^{-1/2}),$$
which allows us to replace the $\aipw{k}$ term in $\mathbf{I}_n$ with $\E_P[\tau(X_{S_k})\mid X_{S_\cap}]$, i.e., 
\begin{align*}
    \wh{\tau}_k^{\,\mathtt{R}}-\wh{\tau}_k^{\,\mathtt{R},*} &=\mathbf{I}_n+\o(n^{-1/2})\\
    &=\P_n \left[\left(\wh{w}(X_{S_\cap})-w^*(X_{S_\cap})\right) \left(\E_P[\tau(X_{S_k})\mid X_{S_\cap}]-\tauR\right)\right] + \o(n^{-1/2}).\label{emp-minus-oracle-tau-k}\numberthis
\end{align*}

\paragraph{Difference of aggregated $\wh\tau_k^{\,\mathtt{R}}$ and its oracle version.} Recall that $\nu^*_1:=1-\sum_{k=2}^K\nu^*_k$ and
$$\nu^*_{2:K} := \left(\E_P\left[ w^*(X_{S_\cap})\,g(X_{S_\cap})\,g(X_{S_\cap})^\top\right]\right)^{-1}\E_P\left[ w^*(X_{S_\cap})\,g(X_{S_\cap})\left(\tau(X_{S_1})-\tauR \right)\right].$$ To quantify how close $\wh\tau^\nam_{bc}-B_n=\sum_{k=1}^K \wh{\nu}_k \wh{\tau}_k^{\,\mathtt{R}}$ is to its oracle $\sum_{k=1}^K\nu^*_k\wh{\tau}_k^{\,\mathtt{R},*}$, we begin by writing
\begin{align*}
   \sum_{k=1}^K  \wh\nu_k \wh{\tau}_k^{\,\mathtt{R}}- \sum_{k=1}^K \nu_k^*\,\wh{\tau}_k^{\,\mathtt{R},*}
   &= \sum_{k=1}^K(\wh\nu_k - \nu_k^*)   \wh{\tau}_k^{\,\mathtt{R},*}+ \sum_{k=1}^K\nu_k^*(\wh{\tau}_k^{\,\mathtt{R}}-\wh{\tau}_k^{\,\mathtt{R},*}) + \sum_{k=1}^K(\wh\nu_k - \nu_k^*)   (\wh{\tau}_k^{\,\mathtt{R}}-\wh{\tau}_k^{\,\mathtt{R},*}).
   \numberthis\label{AR-gap-AR-star-1}
\end{align*}
For the first term in the above display, use $\sum_{k=1}^K\wh\nu_k=\sum_{k=1}^K\nu_k^*=1$ to deduce that 
\begin{align*}
    \sum_{k=1}^K  (\wh\nu_k - \nu_k^*)\, \wh{\tau}_k^{\,\mathtt{R},*} &= \sum_{k=1}^K(\wh\nu_k - \nu_k^*)\,   \P_n \left[w^*(X_{S_\cap})\left( \aipw{k}-\tauR\right)\right]\\
    &=\o(1)\O(n^{-1/2})=\o(n^{-1/2}), \numberthis\label{AR-gap-AR-star-2}
\end{align*}
since $$\E_P \left[w^*(X_{S_\cap})\left( \aipw{k}-\tauR\right)\right]\stackrel{\text{Tower}}{=}\E_P \left[w^*(X_{S_\cap})\left( \tau(X_{S_k})-\tauR\right)\right]\stackrel{\eqref{eq:def-tauR}}{=}0.$$ 
Next, we show that the last sum in \eqref{AR-gap-AR-star-1} is asymptotically negligible. To do this, we start from \eqref{emp-minus-oracle-tau-k}, and apply the Cauchy-Schwarz inequality and \cref{lemma:rate-of-weights} (which shows that $\|\wh{w}-w^*\|_{L^2(\P_n)}=\o(n^{-1/4})$) to deduce that
\begin{equation}\label{rate-emp-minus-oracle-tau-k}
|\wh{\tau}_k^{\,\mathtt{R}}-\wh{\tau}_k^{\,\mathtt{R},*} |\le \|\wh{w}(X_{S_\cap})-w^*(X_{S_\cap})\|_{L^2(\P_n)}\|\E_P[\tau(X_{S_k})\mid X_{S_\cap}]-\tauR\|_{L^2(\P_n)}+\o(n^{-1/2})=\o(n^{-1/4}),
\end{equation}
where $\|\E_P[\tau(X_{S_k})\mid X_{S_\cap}]-\tauR\|_{L^2(\P_n)}=\O(1)$ followed from the fact that $\E_P[\tau(X_{S_k})\mid X_{S_\cap}]$ has finite second moment (since its moment generating function is finite on an open neighborhood of the origin). On the other hand, \cref{lemma:rate-of-nu} shows that $\wh\nu_k-\nu_k^*=\o(n^{-1/4})$---which combined with \eqref{rate-emp-minus-oracle-tau-k} yields the following.
\begin{equation}
   \sum_{k=1}^K (\wh\nu_k - \nu_k^*)  (\wh{\tau}_k^{\,\mathtt{R}}-\wh{\tau}_k^{\,\mathtt{R},*})=\o(n^{-1/4})\o(n^{-1/4})=\o(n^{-1/2}).\label{AR-gap-AR-star-3}
\end{equation}
Having shown the first and third sums in \eqref{AR-gap-AR-star-1} are asymptotically negligible (\cref{AR-gap-AR-star-2,AR-gap-AR-star-3}), we arrive at
\begin{align*}
        \sum_{k=1}^K \wh\nu_k\,\wh{\tau}_k^{\,\mathtt{R}}- \sum_{k=1}^K \nu_k^*\,\wh{\tau}_k^{\,\mathtt{R},*} &=\sum_{k=1}^K\nu_k^*(\wh{\tau}_k^{\,\mathtt{R}}-\wh{\tau}_k^{\,\mathtt{R},*})+\o(n^{-1/2})\label{AR-gap-second-last}\numberthis
\end{align*}
Denote by $\phi(X_{S_\cap};\nu)$ the conditional mean of contrasts aggregated using affine weights $\nu$ and centered by $\tauR$, i.e., $$\phi(X_{S_\cap};\nu):=\sum_{k=1}^K\nu_k\,\E_P[\tau(X_{S_k})\mid X_{S_\cap}]-\tauR.$$
Using this notation and \eqref{emp-minus-oracle-tau-k}, we can write the right-hand side of \eqref{AR-gap-second-last} as
\begin{equation}\label{gap-between-agg-AR-and-oracle}
   \sum_{k=1}^K \wh\nu_k\,\wh{\tau}_k^{\,\mathtt{R}}- \sum_{k=1}^K \nu_k^*\,\wh{\tau}_k^{\,\mathtt{R},*} =\P_n \left[\left(\wh{w}(X_{S_\cap})-w^*(X_{S_\cap})\right) \phi(X_{S_\cap};\nu^*)\right]+\o(n^{-1/2}).
\end{equation}

\paragraph{Using the bias-correction.} In view of the last display, it remains to handle the first-order bias from estimating the transfer weights, which we achieve by using a bias-correction term $B_n$ such that $ \wh\tau^\nam_{bc} = \sum_{k=1}^K\wh\nu_k\,\wh\tau_k^{\,\mathtt{R}} + B_n$; see \cref{algo:rewt-adj-sets} for the definition of $B_n$. We resume from \eqref{gap-between-agg-AR-and-oracle}, and apply \cref{lemma:Taylor-expand-weights}---which gives a Taylor expansion for $(\wh{w}(X_{S_\cap})-w^*(X_{S_\cap}))$---to deduce that
\begin{align*}
    \wh\tau^\nam_{bc} - \sum_{k=1}^K \nu_k^*\,\wh{\tau}_k^{\,\mathtt{R},*} &= B_n+\P_n\left[ w^*(X_{S_\cap})\left(\wh{g}(X_{S_\cap})^\top\,\wh{\lambda} -g(X_{S_\cap})^\top\lambda^*\right)\,\phi(X_{S_\cap};\nu^*)\right]+\o(n^{-1/2})\\
    &=: B_n + \Delta^\lambda_n+\Delta^g_n + \o(n^{-1/2}),\label{AR-minus-oracle}\numberthis
\end{align*}
 where
 \begin{align*}
     \Delta^\lambda_n&:=(\wh{\lambda}-\lambda^*)^\top \P_n\left[ w^*(X_{S_\cap})\,g(X_{S_\cap})\,\phi(X_{S_\cap};\nu^*)\right],\\[2mm]
     \Delta^g_n&:=\wh{\lambda}^{\,\top} \,\P_n\left[ w^*(X_{S_\cap})\left(\wh{g}(X_{S_\cap})-g(X_{S_\cap})\right)\,\phi(X_{S_\cap};\nu^*)\right].
 \end{align*}
We first show that $\Delta_n^\lambda$ is asymptotically negligible (i.e., $\o(n^{-1/2})$).
Use $\nu^*_1=1-\sum_{k=2}^K\nu^*_k$ to write 
\begin{align*}
\phi(X_{S_\cap};\nu^*)&=\sum_{k=1}^K\nu^*_k\,\E_P[\tau(X_{S_k})\mid X_{S_\cap}]-\tauR=\E_P[\tau(X_{S_1})\mid X_{S_\cap}]-\tauR -g(X_{S_\cap})^\top \nu^*_{2:K}.
\end{align*} On the other hand, the tower property gives
\begin{align*}
   \E_P\left[ w^*(X_{S_\cap})\,g(X_{S_\cap})\,g(X_{S_\cap})^\top \nu^*_{2:K}\right]  &= \E_P\left[ w^*(X_{S_\cap})\,g(X_{S_\cap})\left(\tau(X_{S_1})-\tauR \right)\right]\\
&=\E_P\left[ w^*(X_{S_\cap})\,g(X_{S_\cap})\left(\E_P[\tau(X_{S_1})\mid X_{S_\cap}]-\tauR \right)\right]\end{align*}
This implies that
$\E_P[ w^*(X_{S_\cap})\,g(X_{S_\cap})\left(\E_P[\tau(X_{S_1})\mid X_{S_\cap}]-\tauR -g(X_{S_\cap})^\top \nu^*_{2:K}\right)]=0$,
which we can restate as
$$\E_P[ w^*(X_{S_\cap})\,g(X_{S_\cap})\,\phi(X_{S_\cap};\nu^*)]=0.$$
We combine the above display with the fact that $\wh{\lambda}\Pto \lambda^*$ (from~\cref{lemma:rate-of-lambda}) to deduce that
 \begin{align*}
     \Delta^\lambda_n&:=(\wh{\lambda}-\lambda^*)^\top \P_n\left[ w^*(X_{S_\cap})\,g(X_{S_\cap})\,\phi(X_{S_\cap};\nu^*)\right]\\[1mm]
&=\o(1)\O(n^{-1/2})=\o(n^{-1/2}).\label{Delta-lambda}\numberthis
 \end{align*}
It only remains to handle the sum $B_n+\Delta^g_n$ in \eqref{AR-minus-oracle}. Recall that  $g(X_{S_\cap})=\E_P[\Delta\tau(X)\mid X_{S_\cap})]=\E_P[\Delta\aipw{}\mid X_{S_\cap}]$ where $$\Delta\aipw{}:=(\aipw{1}-\aipw{2},\,\dots,\, \aipw{1}-\aipw{K}).$$
We show in \cref{lemma:bias-corr-is-close-to-oracle} that 
$B_n=\wh{B}_n^* + \o(n^{-1/2})$ where 
$$\wh{B}_n^*=\wh{\lambda}^\top \,\P_n\left[w^*(X_{S_\cap})\,(\Delta\aipw{}-\wh{g}(X_{S_\cap}))\,\phi(X_{S_\cap};\nu^*)\right].$$ (The rationale behind this notation involving both ${\,}\hat{\,}\,$ and $*$ is that $\wh{B}_n^*$ combines both estimates and oracle components, and denotes a quantity defined only for the purposes of this proof.)
Using $\wh{B}_n^*$ as a replacement for $B_n$ in \eqref{AR-minus-oracle} and invoking \eqref{Delta-lambda}, we can continue from \eqref{AR-minus-oracle} to write
\begin{align*}
    \wh\tau^\nam_{bc}-\sum_{k=1}^K \nu_k^*\,\wh{\tau}_k^{\,\mathtt{R},*}  &= \wh{B}_n^* + \Delta^g_n + \o(n^{-1/2})\\
    &= \wh{\lambda}^{\,\top} \,\P_n\left[ w^*(X_{S_\cap})\,\phi(X_{S_\cap};\nu^*)\left(\Delta\aipw{}-g(X_{S_\cap})\right)\right]+\o(n^{-1/2})\\[1mm]
    &=\lambda^{*,\top}\P_n\left[ w^*(X_{S_\cap})\,\phi(X_{S_\cap};\nu^*)\left(\Delta\aipw{}-g(X_{S_\cap})\right)\right]+\o(n^{-1/2}),\label{pre-ALR}\numberthis
\end{align*}
where the last equality is due to the fact that
$$ (\wh{\lambda}-\lambda^*)^\top \,\P_n\left[ w^*(X_{S_\cap})\,\phi(X_{S_\cap};\nu^*)\left(\Delta\aipw{}-g(X_{S_\cap})\right)\right]=\o(1)\O(n^{-1/2})=\o(n^{-1/2}),$$
which follows from $\E_P[\Delta\aipw{}\mid X_{S_\cap}]=g(X_{S_\cap})$ and $\wh\lambda-\lambda^*=\o(1)$ (\cref{lemma:rate-of-lambda}).

\paragraph{Asymptotic linear expansion.} We deduce from \eqref{pre-ALR} that  
\begin{align*}
    \wh\tau^\nam_{bc}-\tauR &= \sum_{k=1}^K \nu_k^*\,\wh{\tau}_k^{\,\mathtt{R},*} -\tauR + (\lambda^*)^\top\P_n\left[ w^*(X_{S_\cap})\left(\Delta\aipw{}-g(X_{S_\cap})\right)\,\phi(X_{S_\cap};\nu^*)\right]+\o(n^{-1/2})
    \\
    &=\P_n \left[w^*(X_{S_\cap})\left(\sum_{k=1}^K\nu^*_k\left( \aipw{k}-\tauR\right)\right.\right.\\
    &\qquad\qquad\left.\left.+\lambda^{*,\top} \left(\Delta\aipw{}-g(X_{S_\cap})\right)\left(\sum_{k=1}^K\nu^*_k\,\E_P[\tau(X_{S_k})\mid X_{S_\cap}]-\tauR\right)\right)\right]+\o(n^{-1/2}),
\end{align*}
which gives us the desired result via the classical central limit theorem.
\end{proof}

\subsection{Proof of \texorpdfstring{\cref{lem:why-shared-covs-only}}{Lemma on admissible reweightings}}\label{app:proof-of-why-shared-covs-only}

\begin{proof}
Suppose that $w=w(X_{S_\cap})$ and that $S_k$ is a valid adjustment set under $P$. Then, since $S_\cap\subseteq S_k$, the
conditional law of $(Y(0),Y(1),A)$ given $X_{S_k}$ is unchanged under $wP$, and hence $Y(a)\indep A\mid X_{S_k}$ also holds under $wP$.

For the converse, consider \cref{example-2}:~Draw $U_1,U_2,X_1\iid\normal(0,1)$, $A=\ind{U_1+X_1>0}$, $X_2 =  U_1+U_2$, and $Y=A\cdot(1 + \frac{1}{2} X_1 - 6U_2)+5U_2+ \eps_Y$ where $\eps_{Y}\sim\normal(0,1)$ independent of $X_1,U_1,U_2$. In this example, $S_1=\{X_1\}$ is valid, but $S_2=\{X_1,X_2\}$ is not. To illustrate that an arbitrary reweighting can turn both of these candidate adjustments invalid, consider the following reweighting: \[w(X_1,X_2)=\sqrt{1-4t}\cdot \exp\left(tX_2^2\right), \quad 0<t<1/{4}.\]
Clearly, $w\ge 0$ and $\E_P[w]=1$. Under $wP$, $(U_1, U_2)$ remain jointly Gaussian and independent of $X_1$, with \[\cov_{wP}(U_1,U_2)=\frac{2t}{1-4t}, \quad \var_{wP}(U_1)=\frac{1-2t}{1-4t}.\]
Consequently, $A= \ind{U_1 + X_1 >0}$ and $U_2$ are dependent conditional on $X_1$. To see this, note that by Gaussian regression,
$$
\E_{w P}\left[U_2 \mid U_1, X_1\right] =\frac{\cov_{w P}\left(U_1, U_2\right)}{\var_{w P}\left(U_1\right)} U_1 =\frac{2 t}{1-2 t} U_1 .
$$
Since $A=1\left\{U_1+X_1>0\right\}$, the tower property gives
$$
\cov_{w P}\left(U_2, A \mid X_1=x\right)=\frac{2 t}{1-2 t} \cov_{w P}\left(U_1, 1\left\{U_1>-x\right\}\right)>0,
$$
where the inequality is strict because 
$
\E_{w P}\left[U_1 \mid U_1>-x\right]>\E_{w P}\left[U_1 \mid U_1 \leq-x\right]
$ for every $x\in\R$.
Finally, since $Y(0)=5 U_2+\eps_Y$, with $\eps_Y$ independent of $(A, X_1, U_2)$ under $w P$,
$$
\cov_{w P}\left(Y(0), A \mid X_1=x\right)=5 \cov_{w P}\left(U_2, A \mid X_1=x\right)>0 .
$$
This implies that $S_1$ is invalid under $wP$, while $S_2$ remains invalid because the conditional law of $(A, Y(0), Y(1))$ given $(X_1,X_2)$ is unchanged under the reweighting.

We next show that solving the KL optimization problem with arbitrary weights can produce a noncausal common estimand. Since $\tau(X_1)=1+\tfrac12X_1$ (see \cref{proof:eg2-details} for a proof), the difference in contrasts defined in \cref{assump:heterogeneity} is given by
\[
\Delta\tau_2(X)=\E_P[U_2\mid A=1,X_1,X_2]+5\E_P[U_2\mid A=0,X_1,X_2].
\]
Indeed, $U_2\mid X_2\sim N(X_2/2,1/2)$ and, conditional on $(X_1,X_2)$, $A=1$ is equivalent to $U_2<X_1+X_2$. Thus, writing $z=\sqrt2X_1+X_2/\sqrt2$ and denoting the standard normal density and cdf by $\varphi$ and $\Phi$,
\[
\begin{aligned}
\E_P[U_2\mid A=1,X_1,X_2]
 &=\frac{X_2}{2}-\frac{\varphi(z)}{\sqrt2\,\Phi(z)},\quad
\E_P[U_2\mid A=0,X_1,X_2]
 &=\frac{X_2}{2}+\frac{\varphi(z)}{\sqrt2\,\Phi(-z)}.
\end{aligned}
\]
Both conditional means are strictly increasing in $X_2$ for fixed $X_1$, since $\varphi(z)/\Phi(z)$ is decreasing and $\varphi(z)/\Phi(-z)$ is increasing. Therefore, $\Delta\tau_2(X)$ is strictly increasing in $X_2$. Moreover, symmetry gives
\[
\E_P[\Delta\tau_2(X)]
 =\frac4{\sqrt2}\E_P\!\left[\frac{\varphi(z)}{\Phi(-z)}\right]>0.
\]
The difference $\Delta\tau_2(X)$ takes both signs (let $X_2\to\pm\infty$ with $X_1=0$), and has finite exponential moments of every order. Hence the unrestricted version of (3) has the unique KL minimizer
\[
 w^\dagger(X)=\frac{\exp\{\lambda^\dagger\Delta\tau_2(X)\}}
 {\E_P[\exp\{\lambda^\dagger\Delta\tau_2(X)\}]},
 \qquad \E_P[w^\dagger\Delta\tau_2]=0,
 \qquad\lambda^\dagger<0.
\]
For every feasible $w$, the identity $D_{\mathrm{KL}}(wP\|P)=D_{\mathrm{KL}}(wP\|w^\dagger P)+D_{\mathrm{KL}}(w^\dagger P\|P)$ proves optimality and uniqueness. As $X_2\indep X_1$ under $P$ and $w^\dagger$ is strictly decreasing in $X_2$ for fixed $X_1$, we have $\E_P[w^\dagger X_2]<0$. Finally, $\E_P[U_2\mid X_1,X_2]=X_2/2$, so
\[
\begin{aligned}
\E_{w^\dagger P}[Y(1)-Y(0)]
 &=\E_P[w^\dagger\tau(X_1)]-3\E_P[w^\dagger X_2]\\
 &>\E_P[w^\dagger\tau(X_1)]
 =\E_P[w^\dagger\tau(X_1,X_2)],
\end{aligned}
\]
where the last equality follows from the constraint in \eqref{adj-sets-popln-opt} when solved with arbitrary weights $w=w(X_1,X_2)$. Thus, the common weighted contrast is not the ATE of the reweighted population. This completes the proof.
\end{proof}

\subsection{Proof of \texorpdfstring{\cref{lem:minimax-reweighted-ATE}}{Lemma on minimax lower bound}}\label{app:proof-of-minimax-reweighted-ATE}

\begin{proof}
We begin by recalling the notation: $\overline\CC(P):=\bigcup_{k=1}^K \CC_k(P)$, where $\CC_k(P)$ collects all probability distributions of $(Y(0), Y(1), A, X)$ that induce the same observed data law as $P$, and imply ignorability for the adjustment set $S_k$. 
Note that each $\CC_k(P)$ is non-empty:~Draw $(A,X)$ from their marginal under $P$. If $A=a$, draw $Y(a)$ from the observed conditional law of $Y$ given $(A=a,X)$, and independently draw $Y(1-a)$ from the observed conditional law of $Y$ given $\smash{(A=1-a,X_{S_k})}$. With $Y=Y(A)$, the resulting $(Y,A,X)$ has the same law as under $P$, and satisfies $(Y(0), Y(1))\indep A$ conditional on $\smash{X_{S_k}}$.

For any $Q\in \CC_k(P)$, the set $S_k$ satisfies ignorability. Therefore, by the same argument as in the proof of \cref{propo:reweighted-estimand-makes-sense} given in \cref{app:proof-of-propo:reweighted-estimand-makes-sense}, we can say that
\[\E_{wQ}[Y(1)-Y(0)]=\E_{Q}[w(X_{S_\cap})\,\tau(X_{S_k})]=\E_{P}[w(X_{S_\cap})\,\tau(X_{S_k})]=:\overline\tau_k(wP),\]
where in the middle step we used the fact that $Q$ induces the same observed data law as $P$.
Since $\overline\CC(P)=\bigcup_{k=1}^K \CC_k(P)$ and each $\CC_k(P)$ is non-empty, we deduce that
\[\sup_{Q\in\overline\CC(P)}\left(a-\E_{wQ}[Y(1)-Y(0)]\right)^2\ge \max_{1\le k\le K}\left(a-\overline\tau_k(wP)\right)^2.\]
Finally, taking infimum over $a$ and using algebra,
\begin{align*}
  &\inf_{a\in\R}  \sup_{Q\in\overline\CC(P)}\left(a-\E_{wQ}[Y(1)-Y(0)]\right)^2\\
  &\ge \min_{a\in\R}\max_{1\le k\le K}\left(a-\overline\tau_k(wP)\right)^2\\
  &\ge \min_{a\in\R}\max\left\{\left(a-m\right)^2 , \left(a-M\right)^2\right\}\quad  \left(m:=\min_{k}\overline\tau_k(wP),\ M:=\max_{k}\overline\tau_k(wP)\right)\\
  &\ge \min_{a\in\R}\frac{1}{2}\left\{\left(a-m\right)^2 + \left(a-M\right)^2\right\}\\
  &= \min_{a\in\R} \left\{\left(a-\frac{m+M}{2}\right)^2 + \frac{1}{4}\left(M-m\right)^2\right\}\\
  &=\frac{1}{4}\left(M-m\right)^2 = \frac{1}{4}\left(\max_{k}\overline\tau_k(wP)-\min_{k}\overline\tau_k(wP)\right)^2.
\end{align*}
This establishes the desired bound. Next, suppose $w=w(X_{S_\cap})\ge 0$, $\E_P[w]=1$ is such that  $\max_{k}\overline\tau_k(wP)=\min_{k}\overline\tau_k(wP)=\vartheta$. Then, for each $k$, 
\begin{equation}
    \label{eq:wQcommon}
    \E_{wQ}[Y(1)-Y(0)]=\E_{Q}[w(X_{S_\cap})\,\tau(X_{S_k})]=\E_{P}[w(X_{S_\cap})\,\tau(X_{S_k})]=\overline\tau_k(wP)=\vartheta.
\end{equation}
Consequently,
\begin{align*}
    &\arginf_{a\in\R}  \sup_{Q\in\overline\CC(P)}\left(a-\E_{wQ}[Y(1)-Y(0)]\right)^2= \arginf_{a\in\R}  \left(a-\vartheta\right)^2 = \vartheta.
\end{align*}
Moreover, this unique minimizer $\vartheta$ coincides with the ATE of the population $wQ$ since \eqref{eq:wQcommon} holds, in particular, for the valid adjustment set under \cref{assump-at-least-one-is-valid}.
\end{proof}

\subsection{Proof of \texorpdfstring{\cref{propo:Neyman-orthogonality}}{Neyman orthogonality}}
\label{proof-of-neyman-orthogonality}

\begin{proof} All interchanges of differentiation and expectation below are justified by the dominated convergence theorem. Indeed, for $|t|<t_0$ we have $e_t(X_{S_k})\in[\xi/2,1-\xi/2]$ and $|g_t^\top\lambda_t-g^\top\lambda^*|\le|t|(H\|\lambda^*\|+H\|h_\lambda\|)+(\eps/2)\|g\|$. The integrand $w_{g_t,\lambda_t}H_t$ and its derivative with respect to $t$ are bounded by $e^{(\|\lambda^*\|+\eps/2)\|g\|}(1+\|g\|)^3(1+|m_1|)(1+|Y|+\sum_{k,a}|\mu_a(X_{S_k})|)$ times a constant, and the same bound with $(1+\|g\|)$ in place of the last three factors controls the normalizing constant $\E_P[e^{g_t^\top\lambda_t}]$ and its derivative with respect to $t$. Since $(1+\|g\|)^3e^{(\eps/2)\|g\|}\lesssim e^{\eps\|g\|}$, these bounds are integrable under the conditions of \cref{propo:Neyman-orthogonality}.

\medskip

We begin by writing down two derivatives, one for the AIPW part and one for the reweighting function.
    Note that $\E_P[\aipw{\mu,e}\mid X_{S_k}]=\tau(X_{S_k})$ by the tower property. Along the nuisance path $(\mu_t,e_t)$, algebra gives the exact identity
    \begin{equation}
        \label{eq:aipw-nyo}
        \E_P[\aipw{\mu_t, e_t}\mid X_{S_k}]-\tau(X_{S_k})=t^2h_{e,k}(X_{S_k})\left[\frac{h_{\mu,1,k}(X_{S_k})}{e_t(X_{S_k})}+\frac{h_{\mu,0,k}}{1-e_t(X_{S_k})}\right].
    \end{equation}
    Indeed, suppressing $(X_{S_k})$ to reduce notational clutter,
    \begin{equation*}
    \begin{split}
       \E_P[\aipw{\mu_t, e_t}\mid X_{S_k}]&= \mu_1+th_{\mu,1,k}-\mu_0-th_{\mu,0,k}-\frac{te}{e_t}h_{\mu,1,k}+\frac{t(1-e)}{1-e_t}h_{\mu,0,k}\\
        &=\tau+th_{\mu,1,k}\left(1-\frac{e}{e_t}\right)-th_{\mu,0,k}\left(1-\frac{1-e}{1-e_t}\right),
    \end{split} 
    \end{equation*}
and now use $e_t -e = th_{e,k}$ to obtain \eqref{eq:aipw-nyo}. It follows from \eqref{eq:aipw-nyo} that
\begin{equation}
    \label{eq:aipw-path-derivative}
    \frac{d}{dt}\E_P\left[\aipw{\mu_t, e_t}\mid X_{S_k}\right]\bigg|_{t=0}=0.
\end{equation}
    Taking conditional expectations, the same conclusion holds conditional on $X_{S_\cap}$.

    \medskip

    We next differentiate the normalized transfer weight $w_{g,\lambda}$. Suppressing $(X_{S_\cap})$ to reduce notational clutter, we can write the exponent of $w_{g_t, \lambda_t}$ as
    \[g_t^\top \lambda_t = g^\top \lambda^* + t(h_g^\top \lambda^* +g^\top h_\lambda)+t^2 h_g^\top h_\lambda.\]
    Therefore, by ordinary differentiation and dominated convergence,
    \begin{multline}\label{deriv-w}
 \frac{d}{dt} w_{g_t, \lambda_t}(X_{S_\cap})\Big|_{t=0}=w^*(X_{S_\cap})\Big[\left\{ h_g(X_{S_\cap})^\top \lambda^* +g(X_{S_\cap})^\top h_\lambda \right\}\\
            \qquad\qquad- \E_P\left[w^*\left\{ h_g^\top \lambda^* +g^\top h_\lambda \right\}\right]\Big].
    \end{multline}

We now turn to the main proof and write the influence function as $$\psi(D;\tauR,\mu_t,e_t,m_{1,t},g_t,\lambda_t,\nu_t):= w_{g_t,\lambda_t}(X_{S_\cap})\cdot H_t(D),$$
    where
    \begin{multline*}
        H_t(Y,A,X)=\left(\nu_t^\top \aipw{\mu_t,e_t}-\tauR\right.\\
+\left.\lambda_t^\top (\Delta\aipw{\mu_t,e_t}-g_t(X_{S_\cap}))\left(m_{1,t}(X_{S_\cap})-g_t(X_{S_\cap})^\top \nu_{t,2:K}-\tauR\right)\right).
    \end{multline*}
Differentiating using the product rule and dominated convergence, we obtain
\begin{multline}\label{eq:main-deriv}
    \frac{d}{dt}\E_P\left[\psi(D;\tauR,\mu_t,e_t,m_{1,t},g_t,\lambda_t,\nu_t)\right]\bigg|_{t=0}\\=\E_P\left[\frac{d}{dt}w_{g_t,\lambda_t}(X_{S_\cap})\bigg|_{t=0}H_0(D)\right]
    +\E_P\left[w^*(X_{S_\cap})\frac{d}{dt} H_t(D)\bigg|_{t=0}\right].
\end{multline}
For the remaining calculation, define $$q_t (X_{S_\cap}) := m_{1,t}(X_{S_\cap}) - g_t(X_{S_\cap})^\top \nu_{t,2:K}-\tauR.$$ 
    It follows from tower property and algebra that 
    \begin{equation}
        \label{eq:delta-aipw-mean-zero}
        \E_P[\Delta\aipw{\mu,e}-g(X_{S_\cap})\mid X_{S_\cap}]=0,
    \end{equation} and therefore $$\E_P[H_0(D)\mid X_{S_\cap}]=\E_P[\nu^{*,\top} \aipw{\mu,e}-\tauR\mid X_{S_\cap}]=q_0(X_{S_\cap}).$$
    Moreover, by construction of the transfer weights $w^*$ and the affine weights $\nu^*$, 
    \[\E_P\left[w^*(X_{S_\cap}) H_0(D)\right]=0,\qquad \E_P\left[w^*(X_{S_\cap}) g(X_{S_\cap}) q_0(X_{S_\cap})\right]=0.\]
    We now work with the first term on the right hand side of \eqref{eq:main-deriv} and use \eqref{deriv-w} to write
    \begin{align}
       & \E_P\left[\frac{d}{dt}w_{g_t,\lambda_t}(X_{S_\cap})\bigg|_{t=0}H_0(D)\right]\nonumber\\ 
       &= \E_P\left[w^* \{h_g^\top \lambda^*+g^\top h_\lambda\}\E_P[H_0(D)\mid X_{S_\cap}]\right] - \E_P[w^*H_0(D)]\cdot \E_P[w^*\{h_g^\top \lambda^*+g^\top h_\lambda\}]\nonumber\\
       &= \E_P\left[w^* \{h_g^\top \lambda^*+g^\top h_\lambda\} q_0(X_{S_\cap})\right]\qquad \tag{since $\E_P[w^* H_0(D)]=0$}\\
       &= \E_P[w^* q_0 h_g]^\top\lambda^*,\label{what-remains}
    \end{align}
    since $\E_P[w^* g q_0]=0$.
    It remains to differentiate $H_t$.  Note that
\[
\left.\frac{d}{dt}q_t(X_{S_\cap})\right|_{t=0}
=h_{m_1}(X_{S_\cap})
-h_g(X_{S_\cap})^\top\nu^*_{2:K}
-g(X_{S_\cap})^\top h_{\nu,2:K}.
\]
Therefore, differentiating $H_t$ using the product rule gives
\begin{align}
\label{eq:H-derivative}
\left.\frac{d}{dt}H_t(D)\right|_{t=0}
={}&
\sum_{k=1}^K h_{\nu,k}\{\aipw{\mu,e}-\tau^R\}
+
\sum_{k=1}^K\nu_k^*
\left.\frac{d}{dt}\aipw{\mu_t,e_t}\right|_{t=0}
\nonumber\\
&+h_\lambda^\top\{\Delta\aipw{\mu,e}-g(X_{S_\cap})\}q_0(X_{S_\cap})
\nonumber\\
&+(\lambda^*)^\top
\left\{
\left.\frac{d}{dt}\Delta\aipw{\mu_t,e_t}\right|_{t=0}
-h_g(X_{S_\cap})
\right\}q_0(X_{S_\cap})
\nonumber\\
&+(\lambda^*)^\top\{\Delta\aipw{\mu,e}-g(X_{S_\cap})\}
\{h_{m_1}-h_g^\top\nu^*_{2:K}-g^\top h_{\nu,2:K}\}.
\end{align}
We now take the expectation of \eqref{eq:H-derivative} after multiplying by $w^*$. 
\begin{enumerate}\itemsep2mm
    \item The first sum contributes zero, since $\E_P[\aipw{\mu,e}\mid X_{S_k}]=\tau(X_{S_k})$, $\E_P[w^* m_1]=\tauR$ and $\E_P[w^* g]=0$. 

    \item The second contributes zero as well, by \eqref{eq:aipw-path-derivative} and tower property. 

    \item The second line of \eqref{eq:H-derivative} has conditional mean zero, due to \eqref{eq:delta-aipw-mean-zero}. 

    \item  In the third line of \eqref{eq:H-derivative}, the derivative of the AIPW contrast has conditional mean zero by \eqref{eq:aipw-path-derivative}; therefore the only contribution from that line is
\(
-(\lambda^*)^\top \E_P[w^* h_g\,q_0].
\)

\item 
Finally, the last line of \eqref{eq:H-derivative} has conditional mean zero due to \eqref{eq:delta-aipw-mean-zero}.  
\end{enumerate}

 Consequently,
\begin{equation}
\label{eq:H-part}
\E_P\!\left[w^*(X_{S_\cap})\left.\frac{d}{dt}H_t(D)\right|_{t=0}\right]
=-\E_P[w^* h_g\,q_0]^\top \lambda^*,
\end{equation}
which exactly cancels with \eqref{what-remains}, and thus completes the proof.
\end{proof}

\subsection{Proof of \texorpdfstring{\cref{global-bound}}{global bound on ATE shift}}\label{proof-global-bound}

\begin{proof} Under \cref{assump-at-least-one-is-valid}, there exists one valid adjustment set $S_j$. The Donsker-Varadhan inequality gives  $$\E_Q[\varphi]\le \dkl(Q\,\|\,P)+\log \E_P\left[e^\varphi\right].$$ Using this for $Q=w^*P$ and $\varphi=\pm\, t\,(m_j-\overline\tau_j)$, we obtain $$\left|\tauR-\overline\tau\right|=\left|\E_{w^*P}\left[m_j -\overline\tau_j\right]\right|\le \frac{\dkl(w^*P\,\|\,P)}{t}+\frac{\sigma^2}{2}t.$$
    Optimizing over $t>0$ gives us the desired bound.

    \medskip 

    To show that the bound is sharp, consider the special case where $m_j(X_{S_\cap})$ are jointly Gaussian and $\var_P(m_j)=\sigma^2$ for each $j$. We show in the proof of \cref{local-bound} that $\tauR-\overline\tau = -\Delta^\top \Sigma_g^{-1}c_j$ where $\Delta = \E_P[g]$, $\Sigma_g = \var_P(g)$ and $c_j = \cov(g, m_j)$ where $S_j$ is a valid adjustment set (exists under \cref{assump-at-least-one-is-valid}). Moreover, it follows that $$\dkl(w^*P\| P)=\frac{1}{2}\lambda^{*,\top} \Sigma_{g}\lambda^* = \frac{1}{2}\Delta^\top \Sigma_g^{-1}\Delta.$$ 
    Thus, $$\frac{|\tauR-\overline\tau|}{\sigma\sqrt{2\dkl(w^*P\|P)}}=\frac{|\Delta^\top \Sigma_{g}^{-1}c_j|}{\sigma \sqrt{\Delta^\top\Sigma_g^{-1}\Delta}}.$$
    By the Cauchy-Schwarz inequality, this ratio is at most $1$, with equality when $m_j-\overline\tau_j$ is perfectly correlated with the tilting direction $\lambda^{*,\top} (g-\Delta)$. 

    For a more concrete example, consider $X_1\sim \normal(0,\sigma^2)$, $A\sim \text{Ber}(1/2)$ independent of $X_1$, and let $X_2 $ be distributed as $X_2\mid A, X_1\sim \normal ((2A-1)(\delta-Z)/2, 1)$. Consider potential outcomes $Y(0)=X_2$ and $Y(1)=X_1+X_2$, and adjustment sets $S_1=\{X_1,X_2\}$ and $S_2=\{X_1\}$ with $S_\cap=\{X_1\}$. It is straightforward to verify that $m_1(X_1)=X_1$ and $m_2(X_1)=\delta$. These conditional contrasts are jointly Gaussian (allowing degeneracy) and also $\sigma^2$-subGaussian. It is also straightforward to derive that the optimal weighting is given by $w^*(X_1)=\exp(\delta X_1/\sigma^2 - \delta^2/2\sigma^2)$, which shifts the marginal of $X_1$ from $\normal(0, \sigma^2)$ to $\normal(\delta,\sigma^2)$. Finally, $\tauR=\delta$ while $\overline\tau = 0$, and $\dkl(w^*P\|P)=\delta^2/2\sigma^2$. So, in this example, $|\tauR-\overline\tau|=\sigma\sqrt{2\dkl(w^*P\|P)}$.
\end{proof}

\subsection{Proof of \texorpdfstring{\cref{local-bound}}{local bound on ATE shift}}\label{proof-local-bound}

\begin{proof} We begin with the special case where the conditional contrasts $m_k(X_{S_\cap})=\E_P[\tau(X_{S_k})\mid X_{S_\cap}]$ are jointly Gaussian. The proof for this simpler case provides intuition for the more general case.
  Using a standard formula, the log-MGF of $g=(m_1-m_2,\dots,m_1-m_K)^\top$ can be written as 
  $$\Psi(\lambda):=\log\E_P[\exp(\lambda^\top g(X_{S_\cap}))]=\Delta^\top\lambda+\frac{1}{2}\lambda^\top \Sigma_g\lambda,$$
  where $\Delta=\E_P[g(X_{S_\cap})]$ and $\Sigma_g =\var_P(g(X_{S_\cap}))$. Since $\Sigma_g\succ 0$, the unique minimizer of $\Psi(\lambda)$ is given by $\lambda^*=-\Sigma_g^{-1}\Delta$. Now define $$U:=(g-\Delta)^\top \lambda^*,\qquad\text{and}\qquad V:=m_j-\overline\tau_j.$$ Since $m_j$ are jointly Gaussian, it follows that $(U,V)$ is jointly Gaussian with moment generating function $M(s,t)=\exp(\frac{1}{2}(s^2\sigma^2_U+2st\sigma_{UV}+t^2\sigma^2_V)).$ Differentiating this w.r.t.~$t$ and setting $t=0$ and $s=1$, we obtain $$\E_P[e^U V]=\E_P[e^U]\cdot\cov_P(U,V).$$
Consequently, 
\begin{equation*}
    \begin{split}
        \tauR-\overline\tau=\tauR-\overline\tau_j&=\E_{w^*P}\left[m_j-\overline\tau_j\right]\\
        &=\E_P\left[\frac{e^U}{\E_P[e^U]}V\right]=\cov_P(U,V)=-\Delta^\top\Sigma_g^{-1}c_j,
    \end{split}
\end{equation*}
as desired to show.

\medskip

  We now turn to the general case. To make the sequence $\{P_t\}$ explicit in the notation, we use an additional subscript $t$ for any quantity that is computed under $P_t$, e.g.,
  \[g_t = g(X_{S_\cap}; P_t),\qquad \Delta_t=\E_{P_t}[g_t],\qquad \Sigma_{g,t} =\var_{P_t}(g_t),\qquad \Psi_t(\lambda)=\log\E_{P_t}[\exp(\lambda^\top g_t)].\] We work along a sequence $\{P_t\}_{t\in\N}$ such that $\|\Sigma_{g,t}^{-1/2}\,\Delta_t\|\to 0$ as $t\to\infty$. 
   Our first step is to make the approximation $\Psi_t(\lambda)\approx (\lambda^\top \Delta_t+\frac{1}{2}\lambda^\top \Sigma_{g,t}\lambda)$ precise. 
  Using the Cauchy-Schwarz inequality, \[\big|\lambda^\top \Delta_t\big|\le \big\|\Sigma_{g,t}^{1/2}\lambda\big\|\sqrt{ \left\|\Delta_t^\top \Sigma_{g,t}^{-1}\Delta_t\right\|}.\]
  Combining the above with that facts that the MGFs of $\{\Sigma_{g,t}^{-1/2}\,g_t\}$ are uniformly bounded on a ball of radius $\rc$ around the origin ($\rc>0$ independent of $t$) and $\Delta_t^\top \Sigma_{g,t}^{-1}\Delta_t\to 0$, we conclude that
  \begin{align*}
  &\sup_{t} \E_{P_t}\left[\exp\left(\pm \rc \frac{\lambda^\top(g_t-\Delta_t)}{\|\Sigma_{g,t}^{1/2}\lambda\|}\right) \right]\\
  &=\sup_{t} \E_{P_t}\left[\exp\left(\pm \rc \frac{(\Sigma_{g,t}^{1/2}\lambda)^\top \Sigma_{g,t}^{-1/2} g_t - \lambda^\top \Delta_t}{\|\Sigma_{g,t}^{1/2}\lambda\|}\right) \right]\\
  &\le \sup_t\left\{\sup_{\|\lambda'\|\le \rc} \E_{P_t}\left[\exp\left((\lambda')^\top \Sigma_{g,t}^{-1/2}g_t\right)\right]\left(e^{r\|\Sigma_{g,t}^{-1/2}\Delta_t\|}+e^{-r\|\Sigma_{g,t}^{-1/2}\Delta_t\|}\right)\right\} <\infty,\label{supsup}\numberthis
  \end{align*}
  for any $\lambda\neq 0$. Next, using the integral form of Taylor's theorem,  \[\left|e^x-1-x-\frac{x^2}{2}\right|\le \left|\frac{x^3}{2}\int_0^1(1-s)^2e^{sx}\,ds\right|\le \frac{|x|^3}{2} e^{|x|}\int_0^1 (1-s)^2\,ds= \frac{1}{6}|x|^3 e^{|x|},\]
  for any $x\in\R$. Using this, we derive that following, for any fixed $\lambda$ such that $0<\|\Sigma_{g,t}^{1/2}\lambda\|\le \rc/2$.
\begin{align*}
    &\left|\E_{P_t}\exp(\lambda^\top (g_t-\Delta_t))-1-\frac{1}{2}\lambda^\top \Sigma_{g,t}\lambda\right|\\
    &\le \frac{1}{6}\E_{P_t}\left[\left|\lambda^\top (g_t-\Delta_t)\right|^3 \exp\left(\left|\lambda^\top (g_t-\Delta_t)\right|\right)\right]\\
    &= \frac{1}{6}\|\Sigma_{g,t}^{1/2}\lambda\|^3\, \E_{P_t}\left[b_t^3 \exp\left(\|\Sigma_{g,t}^{1/2}\lambda\| b_t\right)\right] \qquad\qquad\qquad \left(  b_t := \frac{\left|\lambda^\top (g_t-\Delta_t)\right|}{\|\Sigma_{g,t}^{1/2}\lambda\|}\right)\\
    &\le \frac{1}{6}\|\Sigma_{g,t}^{1/2}\lambda\|^3\,\frac{48}{r^3}\E_{P_t}\left[\exp\left(
\frac{\rc b_t}{2}+\|\Sigma_{g,t}^{1/2}\lambda\| b_t\right)\right]\tag{using that $b^3\le \frac{48}{r^3}\exp(\frac{rb}{2})$ since $e^u\ge \frac{u^3}{6}$}\\
&\le \frac{8}{r^3}\,\|\Sigma_{g,t}^{1/2}\lambda\|^3\,\E_{P_t}\left[\exp\left(
\rc \frac{\left|\lambda^\top (g_t-\Delta_t)\right|}{\|\Sigma_{g,t}^{1/2}\lambda\|}\right)\right].\tag{using $\|\Sigma_{g,t}^{1/2}\lambda\|\le \rc/2$ and the definition of $b_t$}
\end{align*}
  Invoking \eqref{supsup}, we conclude from the last display that 
  \begin{equation}
      \label{eq:lam3by2}
      \sup_t\left|\E_{P_t}\left[\exp(\lambda^\top (g_t-\Delta_t))\right]-1-\frac{1}{2}\lambda^\top \Sigma_{g,t}\lambda\right|\le C \|\Sigma_{g,t}^{1/2}\lambda\|^{3},
  \end{equation}
  for any $\lambda$ such that $0<\|\Sigma_{g,t}^{1/2}\lambda\|\le \rc/2$. We now turn to the log-MGF $\Psi_t$, and impose further that $\|\Sigma_{g,t}^{1/2}\lambda\|\le \min\{\rc/2,1\}$. Note that $$\Psi_t(\lambda)-\lambda^\top \Delta_t = \log \E_{P_t}\left[\exp\left(\lambda^\top (g_t-\Delta_t)\right)\right].$$
  It follows by calculus that $0\le u - \log(1+u)\le u^2/2$ for any $u\ge 0$. We apply this inequality with $u=\E_{P_t}[\exp(\lambda^\top(g_t-\Delta_t))]-1$, which is nonnegative by Jensen's inequality. The triangle inequality then gives
\begin{align*}
&    \left|\Psi_t(\lambda)-\lambda^\top \Delta_t-\frac{1}{2}\lambda^\top \Sigma_{g,t}\lambda\right|\\
&\le \left|\E_{P_t}\left[\exp(\lambda^\top (g_t-\Delta_t))\right]-1-\frac{1}{2}\lambda^\top \Sigma_{g,t}\lambda\right|+\frac{1}{2}\left(\E_{P_t}\left[\exp\left(\lambda^\top (g_t-\Delta_t)\right)\right] -1\right)^2  \\
&\le C \|\Sigma_{g,t}^{1/2}\lambda\|^3+\frac{1}{2}\left(\E_{P_t}\left[\exp\left(\lambda^\top (g_t-\Delta_t)\right)\right] -1\right)^2\tag{using \eqref{eq:lam3by2}}\\ 
&\le C \|\Sigma_{g,t}^{1/2}\lambda\|^3+\frac{1}{2}\left(\frac{1}{2}\|\Sigma_{g,t}^{1/2}\lambda\|^2 + C \|\Sigma_{g,t}^{1/2}\lambda\|^3\right)^2\tag{by Jensen and \eqref{eq:lam3by2}}\\ &\le C\|\Sigma_{g,t}^{1/2}\lambda\|^3 + \frac{1}{2}\left(\frac{1}{2}+C\right)^2 \|\Sigma_{g,t}^{1/2}\lambda\|^4\le C_1 \|\Sigma_{g,t}^{1/2}\lambda\|^3. \tag{using $\|\Sigma_{g,t}^{1/2}\lambda\|^3\le 1$ }
\end{align*}
Now choose\footnote{The definition of the constants can be traced back: $C_1=C+\frac{1}{2}(\frac{1}{2}+C)^2$, and $C$ is $8\rc^{-3}$ multiplied by two times the last bound in \eqref{supsup}, since we apply $e^{\rc|b|}\le e^{\rc b}+e^{-\rc b}$.} $\rho$ such that $0\le \rho\le \min\{\rc/2,1\}$ and $C_1 \rho \le 1/4$.
  Then, we continue from the last display to write the following.
  \begin{equation}\label{lifeofPsi}
     \frac{1}{4}(\lambda^\top \Sigma_{g,t}\lambda) \le \Psi_t(\lambda)-\lambda^\top \Delta_t \le \frac{3}{4}(\lambda^\top \Sigma_{g,t}\lambda), \quad\text{for every}\quad\|\Sigma_{g,t}^{1/2}\lambda\|\le \rho.
  \end{equation}
  Next, we turn to $\E_{P_t}[(w^*_t-1)^2]$, where $w^*_t$ is the population-level minimizer of the problem \eqref{adj-sets-popln-opt} under $P_t$, and denote the corresponding exponential tilt parameter by $\lambda^*_t$. Note that 
  \begin{equation}
      \label{eminus1sq}
      \E_{P_t}[(w^*_t-1)^2]=\E_{P_t}[(w^*_t)^2]-1=\exp(\Psi_t(2\lambda^*_t) - 2\Psi_t(\lambda^*_t))-1.
  \end{equation}
  Assume for now that $\Delta_t\neq 0$. Pick $t$ sufficiently large such that $\|\Sigma_{g,t}^{-1/2}\Delta_t\|\le \rho/16$. Using \eqref{lifeofPsi} and Cauchy-Schwarz inequality, we conclude that on the boundary $\|\Sigma_{g,t}^{1/2}\lambda\|=8\|\Sigma_{g,t}^{-1/2}\Delta_t\|$, 
\[\Psi_t(\lambda) \ge -\|\Sigma_{g,t}^{1/2}\lambda\|\cdot\|\Sigma_{g,t}^{-1/2}\Delta_t\| +\frac{1}{4}\|\Sigma_{g,t}^{1/2}\lambda\|^2=8\|\Sigma_{g,t}^{-1/2}\Delta_t\|^2>0.\]  
On the other hand, $\Psi_t(\lambda^*_t)\le \Psi(0)=0$. Therefore, using a standard convexity argument, $\lambda^*_t$ lies in the ellipsoid $\{\lambda\,:\,\|\Sigma_{g,t}^{1/2}\lambda\|< 8\|\Sigma_{g,t}^{-1/2}\Delta_t\|\}$. This, in turn, implies that both $\lambda^*_t$ and $2\lambda_t^*$ lie in the ellipsoid where \eqref{lifeofPsi} holds. Applying the upper bound in \eqref{lifeofPsi} at $2\lambda^*_t$ and the lower bound at $\lambda^*_t$ yields
\[\Psi_t(2\lambda^*_t)-2\Psi_t(\lambda^*_t)\le \frac{5}{2}\|\Sigma_{g,t}^{1/2}\lambda^*_t\|^2\le 160\,\|\Sigma_{g,t}^{-1/2}\Delta_t\|^2.\]
For all sufficiently large $t$, the above bound is at most $1$. Therefore, using $e^x-1=\int_0^x e^t\,dt\le ex$ for $0\le x\le 1$, we continue from \eqref{eminus1sq} to write the following.
\begin{equation}
    \label{lem4.4last}
    \E_{P_t}[(w^*_t-1)^2]\le 160e\big\|\Sigma_{g,t}^{-1/2}\Delta_t\big\|^2,
\end{equation}
for all sufficiently large $t$. The above inequality trivially holds for $\Delta_t=0$.

\medskip

Finally, suppose that $S_{j,t}$ is a valid adjustment set under $P_t$ (exists by \cref{assump-at-least-one-is-valid}). Then, with $m_{j,t}=\E_{P_t}[\tau(X_{S_j};P_t)\mid X_{S_\cap}]$, we can write
\[\tauR(P_t)-\overline\tau(P_t)=\E_{P_t}[(w^*_t -1)(m_{j,t}-\overline\tau(P_t))].\]
Then, using \eqref{lem4.4last} and the Cauchy-Schwarz inequality,
\[\left|\tauR(P_t)-\overline\tau(P_t)\right|\le \sqrt{160e} \big\|\Sigma_{g,t}^{-1/2}\Delta_t\big\|\var_{P_t}^{1/2}(m_{j,t}).\]
Since $\sup_t \var_{P_t}(m_{j,t})<\infty$, this completes the proof.
\end{proof}

\subsection{Proof of \texorpdfstring{\cref{thm:matching.adj.sets-protect-covariates}}{Theorem 2 (protecting covariates)}}
\label{proof:thm:matching.adj.sets-protect-covariates}

\begin{proof}
The proof closely mirrors the proof of \cref{thm:matching.adj.sets} in \cref{proof-of-main-CLT-with-bias-corr}; the key differences are that, the vector of differences-in-contrasts $\Delta\tau(X)\in\R^{K-1}$ is now replaced with the constraint vector $\delta(X)\in\R^{K-1+d}$ as defined in \eqref{eq:new-delta}, and $g(X_{S_\cap})=\E_P[\Delta\tau(X)\mid X_{S_\cap}]$ is replaced with $g(X_{S_\cap})=\E_P[\delta(X)\mid X_{S_\cap}]\in \R^{K-1+d}$. Note that here $\lambda^*,\wh\lambda\in\R^{K-1+d}$,  $\nu^*,\wh\nu\in\R^{K-1+d}$, and that only the weights $\wh\nu_1,\dots,\wh\nu_K$ are used to aggregate the candidate estimators while the bias-correction involves the entire vectors $\wh\nu$ and $\wh{g}$.

\paragraph{Difference between reweighted estimators and their oracle.} Define the oracle estimators of the reweighted ATE $\tauRp$ as 
$$\wh{\tau}_k^{\,\mathtt{R},*}:=\tauRp+\P_n \left[w^*(X_{S_\cap})\left( \aipw{k}-\tauRp\right)\right],\quad k=1,\dots,K.$$
    The arguments used in \cref{proof-of-main-CLT-with-bias-corr} to show \eqref{gap-between-agg-AR-and-oracle} apply here verbatim, and yield the following.
    \begin{align*}
    \sum_{k=1}^K \wh\nu_k\,\wh{\tau}_k^{\,\mathtt{R}} - \sum_{k=1}^K \nu_k^*\,\wh{\tau}_k^{\,\mathtt{R},*} &=\sum_{k=1}^K\nu_k^*(\wh{\tau}_k^{\,\mathtt{R}}-\wh{\tau}_k^{\,\mathtt{R},*})+\o(n^{-1/2})\\[2mm]
    &=\P_n \left[\left(\wh{w}(X_{S_\cap})-w^*(X_{S_\cap})\right) \phi(X_{S_\cap};\nu^*)\right]+\o(n^{-1/2}),\label{Thm2.eq1}\numberthis
\end{align*}
where $\phi(X_{S_\cap};\nu^*)$ denotes the conditional mean of contrasts aggregated using  the affine weights $\nu_1^*,\dots,\nu_K^*$ and centered by $\tauRp$, i.e., 
\begin{equation}
    \label{def-f-x-nu}
    \begin{split}
        \phi(X_{S_\cap};\nu^*)&:=\sum_{k=1}^K\nu_k^*\,\E_P[\tau(X_{S_k})\mid X_{S_\cap}]-\tauRp\\[2mm]
        &=\E_P[\tau(X_{S_1})\mid X_{S_\cap}]-\tauRp-g_{1:(K-1)}(X_{S_\cap})^\top \nu_{2:K}^*.
    \end{split}
\end{equation}
Now comes a point where this proof deviates from \cref{proof-of-main-CLT-with-bias-corr}:~Here we can no longer have $\nu^*$ such that $\E_P[w^*(X_{S_\cap})\, g(X_{S_\cap})\,\phi(X_{S_\cap}; \nu^*)]=0$ (this yields a system of $(K-1+d)$ equations in $(K-1)$ variables). We handle this issue by modifying the bias-correction so that the quantity $g_{1:(K-1)}^\top\nu_{2:K}^*$ in $\phi$ (in \eqref{def-f-x-nu}) is replaced with $g^\top \nu_{2:(K+d)}^*$; see \eqref{def-f-x-nu-tilde} below.

\paragraph{Extending the oracle part to incorporate  additional constraints.} We will use the short-hand $\wh\nu_f=\wh\nu_{(K+1):(K+d)}$ (and the same for $\nu^*$) to remove notational clutter. Observe that the proof of the supporting results
\cref{lemma:AIPW-approx,lemma:rate-of-double-conditional-mean,lemma:rate-of-nu,lemma:rate-of-lambda,lemma:rate-of-weights} all go through, with $\wh{g}$ replaced with $\wh{g}^0 = (\wh{g}_{1:(K-1)}, f - \E_P[f])$. Note also that $\wh{g}=\wh{g}^0 - c_n$ where $c_n =(0_{K-1}, \P_n[f]-\E_P[f])=\O(n^{-1/2})$. Thus, by triangle inequality, $\|\wh{g}-g\|_{L^{q}(\P_n)}=\o(n^{-1/4})$. As a result, the conclusions we need from \cref{lemma:AIPW-approx,lemma:rate-of-double-conditional-mean,lemma:rate-of-nu,lemma:rate-of-lambda,lemma:rate-of-weights} hold true with $\wh{g}$ as well.

\medskip

Next, using $\|\wh{w}-w^*\|_{L^2(\P_n)}=\o(n^{-1/4})$ from \cref{lemma:rate-of-weights}, $\wh{\nu}-\nu^*=\o(n^{-1/4})$ from \cref{lemma:rate-of-nu}, and algebra, we deduce that
\begin{align*}
    &\P_n\left[\wh{w}(f-\P_n [f])^\top\wh\nu_f\right]- \P_n\left[w^*(f-\E_P [f])^\top\nu^*_f\right]\\[1mm]
    &=\P_n[(\wh{w}-w^*)(f-\E_P [f])^\top \nu^*_f] - \P_n[w^*\nu^*_f]^\top (\P_n[f]-\E_P[f])\\[1mm]
    &\qquad\qquad\qquad\qquad\qquad\qquad\quad+(\wh\nu_f-\nu^*)^\top\P_n[w^*(f-\E_P [f])] + \o(n^{-1/2})\\[1mm]
    &=\P_n[(\wh{w}-w^*)(f-\E_P [f])^\top \nu^*_f] - \E_P[w^*\nu^*_f]^\top (\P_n[f]-\E_P[f]) + \o(n^{-1/2}),\label{Thm2.Eq2}\numberthis
\end{align*}
since $\E_P[w^*(f-\E_P[f])]=0$ by definition. We can now combine \eqref{Thm2.eq1} and \eqref{Thm2.Eq2} to deduce that
\begin{align*}
    &\left(\sum_{k=1}^K\wh\nu_k\,\wh\tau_k^{\,\mathtt{R} } -\P_n\left[\wh{w}(f-\P_n [f])^\top\wh\nu_f\right]\right)\\
    &\qquad\qquad\qquad- \left(\sum_{k=1}^K \nu_k^*\,\wh\tau_k^{\,\mathtt{R},*} - \P_n\left[w^*(f-\E_P [f])^\top\nu^*_f\right] +\nu^{*,\top}_f (\P_n[f]-\E_P[f])\right)\\[2mm]
    &=\P_n \left[\left(\wh{w}(X_{S_\cap})-w^*(X_{S_\cap})\right) \phi(X_{S_\cap};\nu^*)\right]-\P_n[(\wh{w}-w^*)(f-\E_P [f])^\top \nu^*_f]+\o(n^{-1/2})\\[2mm]
    &=\P_n \left[\left(\wh{w}(X_{S_\cap})-w^*(X_{S_\cap})\right) \wt\phi(X_{S_\cap};\nu^*)\right]+\o(n^{-1/2}),\label{gap-proof-potect}\numberthis
\end{align*}
where \begin{equation}
    \label{def-f-x-nu-tilde}
\wt\phi(X_{S_\cap};\nu^*):=\E_P[\tau(X_{S_1})\mid X_{S_\cap}]-\tauRp-g(X_{S_\cap})^\top \nu^*_{2:(K+d)}.
\end{equation}
The key difference between the above display and \eqref{def-f-x-nu} is that the above uses the full constraint vector $g(X_{S_\cap})=\E_P[(\Delta\tau, f-\E_P [f])\mid X_{S_\cap}]$, not just its first $K-1$ components.

\paragraph{Completing the proof.}
With $\wt\phi(X_{S_\cap},\nu^*)$ (as defined in \eqref{def-f-x-nu-tilde}) replacing $\phi(X_{S_\cap},\nu^*)$, proofs of the supporting results
\cref{lemma:Taylor-expand-weights,lemma:bias-corr-is-close-to-oracle} also go through. As a result, the same arguments used in the proof of \cref{thm:matching.adj.sets} in \cref{proof-of-main-CLT-with-bias-corr} (precisely, \eqref{Delta-lambda} and \eqref{pre-ALR}) yield the following.
\begin{align*}
    &\P_n \left[\left(\wh{w}(X_{S_\cap})-w^*(X_{S_\cap})\right) \wt\phi(X_{S_\cap};\nu^*)\right]\\[2mm]
    &+\P_n\left[\wh{w}(X_{S_\cap})\,\wh{\lambda}^{\,\top} (\wh\delta^{\,\mathtt{AIPW}}-\wh{g}(X_{S_\cap}))\left(\wh{\E}\left[\wh\tau(X_{S_1})\mid X_{S_\cap}\right]-\wh\tau^{\,\mathtt{R}}_1-\wh{g}(X_{S_\cap})^\top\wh\nu_{2:(K+d)}\right)\right]\\[2mm]
    &=\P_n\left[w^*(X_{S_\cap})\,{\lambda}^{*,\top} (\delta^{\,\mathtt{AIPW}}-{g}(X_{S_\cap}))\left(\E_P\left[\tau(X_{S_1})\mid X_{S_\cap}\right]-\tauRp-{g}(X_{S_\cap})^\top\nu_{2:(K+d)}^*\right)\right]+\o(n^{-1/2}).
\end{align*}
Combining the above display with \eqref{gap-proof-potect}, we obtain the following
asymptotic linear expansion.
\begin{align*}
    \htauRp - \tauRp &=\sum_{k=1}^K\wh\nu_k\,\wh\tau_k^{\,\mathtt{R}} -\tauRp -\P_n\left[\wh{w}(f-\P_n [f])^\top\wh\nu_f\right]\tag{\cref{algo:rewt-adj-sets-protect-covariates}}\\
    &\qquad+\P_n\left[\wh{w}(X_{S_\cap})\,\wh{\lambda}^{\,\top} (\wh\delta^{\,\mathtt{AIPW}}-\wh{g}(X_{S_\cap}))\left(\wh{\E}\left[\wh\tau(X_{S_1})\mid X_{S_\cap}\right]-\wh\tau^{\,\mathtt{R}}_1-\wh{g}(X_{S_\cap})^\top\wh\nu_{2:(K+d)}\right)\right]\\
    &=\P_n\bigg[w^*(X_{S_\cap})\bigg(\sum_{k=1}^K \nu_k^*\,\aipw{k} -\tauRp\bigg)\\
 &\qquad\;\;- (w^*(X_{S_\cap})-1)(f-\E_P [f])^\top\nu^*_{(K+1):(K+d)} \\
 &\qquad\;\;+w^*(X_{S_\cap})\,{\lambda}^{*,\top} (\delta^{\,\mathtt{AIPW}}-{g}(X_{S_\cap}))\left(\E_P\left[\tau(X_{S_1})\mid X_{S_\cap}\right]-\tauRp-{g}(X_{S_\cap})^\top\nu_{2:(K+d)}\right)\bigg]\\
 &\qquad\;\;+\o(n^{-1/2}),
\end{align*}
which yields the desired result.
\end{proof}

\subsection{Proof of \texorpdfstring{\cref{propo:reweighted-estimand-makes-sense-lm-based}}{Lemma A.1 (interpreting estimand)}}
\label{proof-of-propo3}

\begin{proof} The proof essentially follows along the lines of the proof of \cref{propo:reweighted-estimand-makes-sense}.
    We deduce from the construction of the (population) weights $w^*$ that $$\E_P[w^*(X_{S_\cap})\tau^*(X_{S_1})]=\E_P[w^*(X_{S_\cap})\tau^*(X_{S_2})]=\dots=\E_P[w^*(X_{S_\cap})\tau^*(X_{S_K})]=:\tauR.$$
    On the other hand,
    \begin{align*}
   \tauR &:= \E_P[w^*(X_{S_\cap})\tau^*(X_{S_{j}})]= \E_{P} [w^*(X_{S_\cap}) (\tau_j^*+X_{S_j}^\top \gamma_j^*)]\tag{by definition}\\
    &= \E_{P}[w^*(X_{S_\cap})(Y(1)-Y(0))]\tag{by \cref{assump:linear-model} and the tower property}\\
    &= \E_{w^*P}[Y(1)-Y(0)],
    \end{align*}
    as desired to show.
\end{proof}

\subsection{Proof of \texorpdfstring{\cref{propo:exist-n-unique-lm-based}}{Lemma A.2 (existence and uniqueness)}}\label{proof-of-propo4}

\begin{proof}
    The proof is completely analogous to the proof of \cref{propo:exist-n-unique} with $\tau(X_{S_k})$ replaced with $\tau^*(X_{S_k})$, hence omitted.
\end{proof}

\subsection{Proof of \texorpdfstring{\cref{thm:matching.adj.sets-lm-based}}{Theorem A.1 (CLT)}}\label{app:proof-of-thm-lm-based}

\begin{proof}
All the nuisance parameters in \cref{algo:rewt-adj-sets-lm-based} are finite dimensional. We can therefore view the entire procedure as a single $Z$-estimator and appeal to standard theory, instead of expanding each of its steps separately. We begin by recalling that the specification-robust estimator in this case takes the form $$\wh{\tau}^\mathtt{SR,lm} = \wh{a}_1 + \wh{b}_1^{\,\top} \P_n [\wh{w}(X_{S_\cap})X_{S_\cap}].$$ Observe that
\begin{align*}
    \tauR &= \E_P[w^*(X_{S_\cap})\tau^*(X_{S_1})] \tag{from \cref{propo:reweighted-estimand-makes-sense-lm-based}}\\
    &=\E_P[w^*(X_{S_\cap})\E_P[\tau^*(X_{S_1})\mid X_{S_\cap}]] \tag{using the tower property}\\
    &= \E_P[w^*(X_{S_\cap})(a_1+X_{S_\cap}^\top b_1)]\tag{using \cref{assump:linear-model-further}}\\
    &= a_1 + b_1^\top \E_P[w^*(X_{S_\cap})X_{S_\cap}]. \label{tauRsolvesZ4}\numberthis
\end{align*}
Next, we set up the entire procedure as a single $Z$-estimation problem.
For generic coefficients $u=(u_1,\dots,u_K)$ and $v=(v_1,\dots,v_K)$, we introduce the notation
\[
g(X_{S_\cap};u,v):=\big(u_1-u_2+X_{S_\cap}^\top(v_1-v_2),\dots, u_1-u_K+X_{S_\cap}^\top(v_1-v_K)\big),
\]
so that $g^*(X_{S_\cap})=g(X_{S_\cap};a,b)$ under \cref{assump:linear-model-further}, and $\wh g(X_{S_\cap})=g(X_{S_\cap};\wh a,\wh b)$ in \cref{algo:rewt-adj-sets-lm-based}. Let $\theta$ collect the regression coefficients $(\delta_k,\beta_k,\tau_k,\gamma_k)$ and $(a_k,b_k)$ for $k=1,\dots,K$, the tilt parameter $\lambda$, and a scalar $t$ (think of the last one as a placeholder for $\tau$). Consider the estimating equations
\begin{align}
    \E_P\Big[(1,X_{S_k}^\top,\trt,\trt X_{S_k}^\top)^\top\big(Y-\delta_k-X_{S_k}^\top\beta_k-\trt(\tau_k+X_{S_k}^\top\gamma_k)\big)\Big]&=0,\qquad k=1,\dots,K,\label{Z1}\\
    \E_P\Big[(1,X_{S_\cap}^\top)^\top\big(\tau_k+X_{S_k}^\top\gamma_k-a_k-X_{S_\cap}^\top b_k\big)\Big]&=0,\qquad k=1,\dots,K,\label{Z2}\\
    \E_P\Big[\exp\big(g(X_{S_\cap};a,b)^\top\lambda\big)\,g(X_{S_\cap};a,b)\Big]&=0,\label{Z3}\\
    \E_P\Big[\exp\big(g(X_{S_\cap};a,b)^\top\lambda\big)\big(a_1+X_{S_\cap}^\top b_1-t\big)\Big]&=0,\label{Z4}
\end{align}
and denote by $\psi_\theta(D)$ the vector obtained by stacking the functions inside the above expectations.

We treat \eqref{Z1}--\eqref{Z4} as a single system of equations, $\E_P[\psi_\theta(D)]=0$, in the single unknown $\theta$. This system is triangular, and can therefore be solved one block at a time: \eqref{Z1} involves only the coefficients $(\delta_k,\beta_k,\tau_k,\gamma_k)$; given these, \eqref{Z2} determines $(a_k,b_k)$; given those, \eqref{Z3} determines $\lambda$; and finally \eqref{Z4} determines $t$. Each block has a unique solution: \eqref{Z1} and \eqref{Z2} are normal equations of least-squares problems with nonsingular design (\cref{assump:linear-model} (iii)), \eqref{Z3} is the first-order condition of the strictly convex problem \eqref{eq:def-lambda-lm}, and \eqref{Z4} is linear in $t$. The unique solution of the above system is given by
\[
\theta^*:=\Big((\delta_k^*,\beta_k^*,\tau_k^*,\gamma_k^*)_{k=1}^K,\ (a_k,b_k)_{k=1}^K,\ \lambda^*,\ \tauR\Big).
\]
Indeed, the first block satisfies \eqref{Z1} due to \eqref{def:contrasts-lm}; the second block solves \eqref{Z2} by \cref{assump:linear-model-further} and the tower property; $\lambda^*$ solves \eqref{Z3} by the first-order condition in the proof of \cref{propo:exist-n-unique-lm-based} (which in turn follows from the proof of \cref{propo:exist-n-unique} given in \cref{app:proof-of-propo:exist-n-unique}); and $t=\tauR$ solves \eqref{Z4} by \eqref{tauRsolvesZ4}, since $w^*(X_{S_\cap})=\exp(g^*(X_{S_\cap})^\top\lambda^*)/\psi_P(\lambda^*)$.

\cref{algo:rewt-adj-sets-lm-based} solves the empirical version of the same system, $\P_n[\psi_\theta(D)]=0$, in exactly this order: Steps 1 and 2 of \cref{algo:rewt-adj-sets-lm-based} solve the empirical versions of \eqref{Z1} and \eqref{Z2}, Step 3 solves that of \eqref{Z3} (its first-order condition), and Step 5 solves that of \eqref{Z4}, since $$\wh w(X_{S_\cap})=\exp(\wh g(X_{S_\cap})^\top\wh\lambda)/\P_n[\exp(\wh g(X_{S_\cap})^\top\wh\lambda)].$$ We denote this solution by $\wh\theta$, and its last coordinate is the proposed estimator $\wh\tau^{\,\mathtt{SR,lm}}$. It therefore suffices to show that $\wh\theta$ is consistent for $\theta^*$, and that the conditions of a standard asymptotic normality theorem for $Z$-estimators hold. We do so in the following steps.

\paragraph{Integrability requirements.} For any generic value $(a',b',\lambda')$ of the block $((a_k,b_k)_{k=1}^K,\lambda^*)$ in $\theta^*$, the map $X_{S_\cap}\mapsto g(X_{S_\cap};a',b')^\top\lambda'$ is an affine function with slope $\sum_{k=2}^K\lambda_k'(b_1'-b_k')$. At $\theta^*$, this slope has norm smaller than $r/2$ by assumption. Therefore, by continuity, there exists a neighborhood $\Theta$ of $\theta^*$ and constants $c,\eps>0$ such that
\begin{equation}\label{Zint}
    \sup_{\theta\in\Theta}\exp\big(2\,g(X_{S_\cap};a',b')^\top\lambda'\big)\le e^{c}\exp\big((r-\eps)\|X_{S_\cap}\|\big).
\end{equation}
Since $\E_P[\exp(r\|X_{S_\cap}\|)]<\infty$, the right-hand side of \eqref{Zint}, multiplied by any polynomial in $\|X_{S_\cap}\|$, is integrable. Together with \cref{assump:linear-model} (iv), this provides all the integrability we need below.

\paragraph{Consistency.} We show consistency in the same order as in the triangular system in \cref{Z1,Z2,Z3,Z4}. First note that the least-squares estimators in Steps 1 and 2 of \cref{algo:rewt-adj-sets-lm-based} are consistent by the law of large numbers and \cref{assump:linear-model} (iii)--(iv). Next, the empirical criterion $\lambda\mapsto\P_n[\exp(\wh g(X_{S_\cap})^\top\lambda)]$ is convex, and it converges in probability to $\psi_P(\lambda)$ for every $\lambda$ near $\lambda^*$, using the consistency of $(\wh a,\wh b)$, the mean value theorem and \eqref{Zint}. Since $\lambda^*$ is the unique minimizer of $\psi_P$ (\cref{propo:exist-n-unique-lm-based}), it follows from \citet[Theorem 2.7]{neweyMcFadden1994} that, with probability tending to one, the empirical criterion has a local minimizer near $\lambda^*$ that converges to $\lambda^*$ in probability. By convexity, this local minimizer is a global minimizer, and it is unique with probability tending to one, because
\[
\P_n\big[\wh g(X_{S_\cap})\wh g(X_{S_\cap})^\top\big] \Pto \E_P\big[g^*(X_{S_\cap})g^*(X_{S_\cap})^\top\big]\succ0,
\]
where the last inequality (that $\E_P[g^*(g^*)^\top]$ is positive definite) follows from \cref{assump:heterogeneity-lm-based} (i).  Hence $\wh\lambda\to\lambda^*$ in probability. Finally, writing
\[
\wh\tau^{\,\mathtt{SR,lm}}=\frac{\P_n\big[\exp(\wh g(X_{S_\cap})^\top\wh\lambda)\,(\wh a_1+X_{S_\cap}^\top\wh b_1)\big]}{\P_n\big[\exp(\wh g(X_{S_\cap})^\top\wh\lambda)\big]},
\]
we obtain $\wh\tau^{\,\mathtt{SR,lm}}\to\tauR$ in probability from a uniform law of large numbers \citep[Lemma 2.4]{neweyMcFadden1994}, whose envelope condition holds by \eqref{Zint}. We therefore have all the pieces to conclude that $\wh\theta\to\theta^*$ in probability.

\paragraph{Asymptotic normality.} We will apply \citet[Theorem 5.41]{vdV}, which requires a few regularity conditions that we verify below.
\begin{enumerate}
    \item \emph{Smoothness.} The map $\theta\mapsto\psi_\theta(D)$ is smooth. This holds here because the functions in \eqref{Z1} and \eqref{Z2} are affine in $\theta$, and the second-order derivatives of those in \eqref{Z3} and \eqref{Z4} are bounded on $\Theta$ by a polynomial in $\|X_{S_\cap}\|$ times $\sup_{\theta\in\Theta}\exp(g(X_{S_\cap};a',b')^\top\lambda')$, which is integrable by \eqref{Zint}.
    \item \emph{Second moments.} $\E_P\|\psi_{\theta^*}(D)\|^2<\infty$, by \cref{assump:linear-model} (iv) and \eqref{Zint}.
    \item \emph{Invertibility.} By the triangular structure of \eqref{Z1}--\eqref{Z4} noted above, the Jacobian $J:=\E_P[\partial_\theta\psi_\theta(D)]\big|_{\theta=\theta^*}$ is therefore block lower-triangular, with diagonal blocks
    \begin{gather*}
    -\E_P\big[(1,X_{S_k}^\top,\trt,\trt X_{S_k}^\top)^\top(1,X_{S_k}^\top,\trt,\trt X_{S_k}^\top)\big],\qquad -\E_P\big[(1,X_{S_\cap}^\top)^\top(1,X_{S_\cap}^\top)\big],\\
    \E_P\big[\exp(g^*(X_{S_\cap})^\top\lambda^*)\,g^*(X_{S_\cap})g^*(X_{S_\cap})^\top\big],\qquad\quad -\E_P\big[\exp(g^*(X_{S_\cap})^\top\lambda^*)\big].
    \end{gather*}
    The two on the first line are nonsingular by \cref{assump:linear-model} (iii), the third is positive definite by \cref{assump:heterogeneity-lm-based} (i), and the last one is negative. Hence $J$ is nonsingular.
\end{enumerate}
Therefore, \citet[Theorem 5.41]{vdV} applies to the consistent estimator $\wh\theta$ and yields that
\[
\sqrt n\,(\wh\theta-\theta^*)=-J^{-1}\sqrt n\,\P_n[\psi_{\theta^*}(D)]+\o(1).
\]
The last coordinate of $\wh\theta$ is $\wh\tau^{\,\mathtt{SR,lm}}$, so that $\sqrt n\,(\wh\tau^{\,\mathtt{SR,lm}}-\tauR)=\sqrt n\,\P_n[\psi_\mathtt{SR,lm}(D)]+\o(1)$, where $\psi_\mathtt{SR,lm}(D)$ denotes the last coordinate of $-J^{-1}\psi_{\theta^*}(D)$. Consequently, the desired result follows from the classical central limit theorem, with
\begin{equation}\label{eqn:VARlm}
    \mathrm{V}_\mathtt{SR,lm}=\var_P\big(\psi_\mathtt{SR,lm}(D)\big),
\end{equation}
the last diagonal entry of $J^{-1}\,\E_P[\psi_{\theta^*}(D)\psi_{\theta^*}(D)^\top]\,(J^{-1})^\top$. This finishes the proof.
\end{proof}

\subsection{Proof of \texorpdfstring{\cref{lemma:simplify}}{location shift}} \label{proof:lemma:simplify}

\begin{proof}
    As in \eqref{eq:def-lambda-lm}, $\lambda^*$ minimizes 
$\psi_P(\lambda)=\E_P[\exp(\lambda(\alpha+\beta X_1))]
=\exp(\lambda\alpha+\lambda^{2}\beta^{2}/2)$. This is strictly convex in
$\lambda$, and setting the derivative to zero gives \(\alpha+\lambda^*\beta^{2}=0\),
i.e., \(\lambda^*=-\alpha/\beta^{2}\).
\cref{propo:exist-n-unique-lm-based}
then yields
\[
  w^*(X_1)=\frac{e^{\lambda^*(\alpha+\beta X_1)}}{\psi_P(\lambda^*)}
  =\exp\left(\lambda^*\beta X_1-\frac{1}{2}(\lambda^*\beta)^{2}\right).
\]
Recognizing this as the likelihood ratio of $\normal(c,1)$ w.r.t.~$\normal(0,1)$ with $c=\lambda^*\beta$, we conclude that
$X_1\sim \normal (c,1)$ under $w^*P$. Finally, invoking \cref{propo:reweighted-estimand-makes-sense-lm-based} and \cref{assump:linear-model-further}, 
\begin{align*}
    \tauR =\E_P[w^*\tau^*(X_{S_1})]=\E_{w^*P}[a_1+b_1X_1] =a_1+b_1c,
\end{align*}
which is the same as $a_2+b_2c$, since $\alpha+\beta c = 0$.
\end{proof}

\subsection{Proof of \texorpdfstring{\cref{eg1-details}}{reweighting in Eg1}} \label{proof:eg1-details}

\begin{proof}
Recall that in \cref{example-1}, we have $Y=1+X_1-X_2+\eps_Y$ on $\{A=1\}$ and $Y=4X_2+\eps_Y$ on $\{A=0\}$. We begin with the contrast that is easier to compute:~Since the linear regression with adjustment set $S_2$ is correctly specified,
\begin{align*}
   \tau^*(X_{S_2})&=\E_P[Y\mid A=1, X_1, X_2]-\E_P[Y\mid A=0, X_1, X_2]\\
   &=(1+X_1-X_2)-4X_2=1+X_1-5X_2.
\end{align*}
Thus, $\E_P[\tau^*(X_{S_2})\mid X_1]=1+X_1$ since
$\E_P[X_2\mid X_1]=0$. In the notation of \cref{assump:linear-model-further}, we showed that $a_2=b_2=1$.

To compute $\tau^*(X_1)$, we need the $L^2$ projection of $X_2$ on $X_1$ under $P(\cdot\mid A=a)$; denote it by $\E_P^*[X_2\mid X_1, A=a]$. We will use the following formula:
\begin{equation}\label{eq:L2-projection}
    \E_P^*[X_2\mid X_1, A=a]=\E_P[X_2\mid A = a]+\frac{\cov(X_1,X_2\mid A=a)}{\var(X_1\mid A=a)}(X_1-\E_P[X_1\mid A=a]).
\end{equation}

Let us now compute the moments involved in the above formula. Write $e(x)=1/(1+e^{5x_1+5x_2})$
for the propensity score. Using $e(-x)=1-e(x)$ and
$(X_1,X_2)\overset d=(-X_1,-X_2)$, we deduce that for any bounded measurable function $h$,
\[
  \E_P[h(X)\mathbf 1\{A=1\}]=\E_P[h(X)e(X)]
  =\E_P[h(-X)e(-X)]=\E_P[h(-X)\mathbf 1\{A=0\}].
\]
Taking $h\equiv1$ gives $P(A=1)=P(A=0)=\tfrac12$, and hence
$\E_P[h(X)\mid A=1]=\E_P[h(-X)\mid A=0]$ for all bounded $h$. When $h$ is an even function,
the two group means agree, and are equal to $\E_P[h(X)]$. If $h$ is an odd function, then the two group means are negatives of each other. Applying this to
$h(x)=x_1^2$ and $h(x)=x_1x_2$,
\begin{equation*}
  \E_P[X_1^2\mid A=a]=1,\qquad \E_P[X_1X_2\mid A=a]=0 \qquad (a=0,1).
\end{equation*}
Similarly, using $h(x)=x_1$ and $h(x)=x_2$ we obtain
$\E_P[X_1\mid A=0]=-\E_P[X_1\mid A=1]$ and likewise for $X_2$. Finally, since $e$ is
symmetric in its two arguments and $(X_1,X_2)$ is exchangeable, we have
$(X_1,X_2,A)\overset d=(X_2,X_1,A)$ and thus $\E_P[X_2\mid A=1]=\E_P[X_1\mid A=1]=:p$. 
To summarize,
\begin{equation*}
\begin{gathered}
  \E_P[X_1\mid A=a]=\E_P[X_2\mid A=a]=\pm p,\
  \var(X_1\mid A=a)=1-p^2,\
  \cov(X_1,X_2\mid A=a)=-p^2,
\end{gathered}
\end{equation*}
with $+$ for $a=1$ and $-$ for $a=0$.

Returning to the main proof, we use \eqref{eq:L2-projection} to conclude that the $L^2$ projection of $X_2$ on $X_1$ under $P(\cdot\mid A=a)$ is given by
\[
 \E_P^*[X_2\mid X_1, A=a]=\begin{cases}\displaystyle
      p-\frac{p^2}{1-p^2}\big(X_1- p\big)=\frac{p}{1-p^2}-\frac{p^2}{1-p^2}X_1 &\text{if }a=1,\\[2mm]
      \displaystyle
     - p-\frac{p^2}{1-p^2}\big(X_1+ p\big)=-\frac{p}{1-p^2}-\frac{p^2}{1-p^2}X_1 &\text{if }a=0.
 \end{cases} 
\]
Consequently,
\begin{align*}
  \tau^*(X_{S_1})
  &=\Big[1+X_1-\frac{p}{1-p^2}+\frac{p^2}{1-p^2}X_1\Big]
   -4\Big[-\frac{p}{1-p^2}-\frac{p^2}{1-p^2}X_1\Big]\\
  &=\Big(1+\frac{3p}{1-p^2}\Big)+\Big(1+\frac{5p^2}{1-p^2}\Big)X_1 .
\end{align*}
Since $\tau^*(X_{S_1})$ is already a function of $X_1$, and $\E_P[\tau^*(X_{S_2})\mid X_1]=1+X_1$, 
\[
  g^*(X_1)=\tau^*(X_{S_1})-\E_P[\tau^*(X_{S_2})\mid X_1]
  =\frac{3p}{1-p^2}+\frac{5p^2}{1-p^2}X_1 .
\]
Finally, applying \cref{lemma:simplify} with $\alpha=3p/(1-p^2)$ and $\beta=5p^2/(1-p^2)$ we deduce that $c=-\alpha/\beta=-3/(5p)$, and $w^*(X_1)=\exp(cX_1-c^2/2)$, i.e., $X_1\sim \normal(c,1)$ under $w^*P$. Moreover,
$\lambda^*=-\alpha/\beta^2=-3(1-p^2)/(25p^3)$ and
 using the last part of \cref{lemma:simplify} with $a_2=b_2=1$, we conclude that
$\tauR=1+c=1-3/(5p)$. 
\end{proof}

\subsection{Proof of \texorpdfstring{\cref{eg2-details}}{reweighting in Eg2}} \label{proof:eg2-details}

\begin{proof}
Recall that in \cref{example-2} we have $ Y=1+\tfrac12X_1-U_2+\eps_Y$ on $\{A=1\}$ and 
$Y=5U_2+\eps_Y$ on $\{A=0\}$, where $A=\ind{X_1+U_1>0}$, $X_2=U_1+U_2$, and $X_1,U_1,U_2\iid \normal(0,1)$.  We begin with the contrast that is easier to compute:~Since the linear regression with adjustment set $S_1$ is correctly specified,
\begin{align*}
   \tau^*(X_{S_1})&=\E_P[Y\mid A=1, X_1]-\E_P[Y\mid A=0, X_1]=1+\frac{1}{2}X_1,
\end{align*}
using that $U_2\indep \{A=a, X_1\}$ and $\E_P[U_2]=0$.
Consequently, $\E_P[\tau^*(X_{S_1})\mid X_1]=1+\frac{1}{2}X_1$. In the notation of \cref{assump:linear-model-further}, we showed that $(a_1,b_1)=(1,\frac{1}{2})$.

Next, let $X_{S_2}$ denote $(X_1, X_2)$. To compute $\tau^*(X_{S_2})$, we need the $L^2$ projection of $U_2$ on $X_{S_2}$ under $P(\cdot\mid A=a)$; denote it by $\E_P^*[U_2\mid X_{S_2}, A=a]$. We will use the following formula, which is the vector analog of \eqref{eq:L2-projection}.
\begin{multline}\label{eq:L2-projection-vector}
    \E_P^*[U_2\mid X_{S_2}, A=a]=\E_P[U_2\mid A = a]\\+{\cov(U_2,X_{S_2}\mid A=a)}{\var(X_{S_2}\mid A=a)}^{-1}(X_{S_2}-\E_P[X_{S_2}\mid A=a]).
\end{multline}
Let us now compute the moments used in the above formula. Since $X_1+U_1\sim N(0,2)$, and $A=\ind{X_1+U_1> 0}$, integration by parts yields
\[\E_P[(X_1+U_1)\ind{A=1}]=\int_0^\infty \frac{z}{\sqrt{2}}\,\phi\left(\frac{z}{\sqrt{2}}\right)dz=-\sqrt{2}\int_0^\infty \phi'(z)\,dz = \sqrt{2}\,\phi(0)=\frac{1}{\sqrt{\pi}},\]
where $\phi(\cdot)$ denotes the density function of standard Gaussian. On the other hand, using the symmetry of $X_1+U_1$ about zero, \[\E_P[(X_1+U_1)^2 \ind{A=1}]=1,\quad \text{and}\quad P(A=1)=\frac{1}{2}.\]
Next, denote $V:=(X_1,X_2,U_2)^\top$, $\vartheta:=\frac{1}{2}\cov(V, X_1+U_1)=\frac{1}{2}(1,1,0)^\top$. Observe that $V^\perp := V-\vartheta (X_1+U_1)$ is uncorrelated with $X_1+U_1$, and therefore independent of $X_1+U_1$ (and hence $V^\perp\indep A$). Using this, together with the above displays and $\E_P[V^\perp]=0$, we deduce that
\[\E_P[V\ind{A=1}]=\frac{\vartheta}{\sqrt{\pi}},\qquad \E_P[VV^\top \ind{A=1}]=\vartheta\vartheta^\top+\frac{1}{2}\var(V^\perp)=\frac{1}{2}\var(V),\]
where in the last step we used $\var(V^\perp)=\var(V)-2\vartheta\vartheta^\top$. Repeating the above with $1-\ind{A=1}$ in place of $\ind{A=1}$, and finally dividing with $P(A=1)=\frac{1}{2}$, we conclude that
\begin{equation}\label{eq:e2stein}
\begin{gathered}
  \E_P[V\mid A=a]=\pm\frac{2\vartheta}{\sqrt\pi}=\pm 
  \frac{1}{\sqrt{\pi}}(1, 1, 0)^\top,\\
  \var(V\mid A=a)=\var(V)
  -\frac{4\vartheta\vartheta^\top}{\pi}=\begin{pmatrix}
      1 & 0 & 0\\
      0 & 2 & 1\\
      0 & 1 & 1
  \end{pmatrix}-\frac{1}{\pi}\begin{pmatrix}
      1 & 1 & 0\\
      1 & 1 & 0\\
      0 & 0 & 0
  \end{pmatrix},
\end{gathered}
\end{equation}
with $+$ for $a=1$ and $-$ for $a=0$. After substituting these into \eqref{eq:L2-projection-vector} and simplifying, we obtain
\[\E_P^*[U_2\mid X_1, X_2, A=a]=\frac{ X_1 + (\pi-1) X_2 +(-1)^a \sqrt{\pi}}{2\pi-3}.\]
Consequently,
\[
  \tau^*(X_{S_2})
  =\left[1+\frac12X_1-\frac{X_1+(\pi-1)X_2-\sqrt{\pi}}{2\pi-3}\right]
  -5\,\frac{(X_1+(\pi-1)X_2+\sqrt\pi)}{2\pi-3},
\]
and thus, using $\E_P[X_2\mid X_1]=0$,
\[
  \E_P[\tau^*(X_{S_2})\mid X_1]
  =\Big(1-\frac{4\sqrt{\pi}}{2\pi-3}\Big)+\Big(\frac12-\frac{6}{2\pi-3}\Big)X_1 .
\]
Combining this with $\tau^*(X_{S_1})=\tau^*(X_1)=1+\frac{1}{2}X_1$, we have
\[
  g^*(X_1)=\tau^*(X_{S_1})-\E_P[\tau^*(X_{S_2})\mid X_1]
  =\frac{4\sqrt\pi+6X_1}{2\pi-3}.
\]
Finally, applying
\cref{lemma:simplify} with $\alpha=4\sqrt\pi/(2\pi-3)$ and $\beta=6/(2\pi-3)$, we deduce that $X_1\sim \normal(c,1)$ under $w^*P$ where $c=-\alpha/\beta=-2\sqrt\pi/3$. Moreover,
$\lambda^*=-\alpha/\beta^2=-\sqrt\pi(2\pi-3)/9$, and using the last part of \cref{lemma:simplify} with $(a_1,b_1)=(1,\frac{1}{2})$, we conclude that
$\tauR=1+\frac{1}{2}c=1-\sqrt\pi/3$. This finishes the proof.
\end{proof}

\subsection{Proof of \texorpdfstring{\cref{propo:exist-n-unique-uniform}}{existence and uniqueness (uniform)}}
\label{proof:propo:exist-n-unique-uniform}

\begin{proof}
Since \cref{assump:heterogeneity-uniform} implies \cref{assump:heterogeneity}(i) for every $P\in\sP$, the existence and uniqueness of $\lambda^*(P)$ directly follow from \cref{propo:exist-n-unique} (see \cref{app:proof-of-propo:exist-n-unique} for a proof). It remains to derive the bound on $\sup_{P\in\sP}\|\lambda^*(P)\|$. Fix any $P\in\sP$ and assume $\lambda^*(P)\neq0$ (otherwise the claim is trivial). Write $\lambda^*(P)=tv$ with $t=\|\lambda^*(P)\|$ and $\|v\|=1$. Since $\psi_P(0)=1$, we have $\psi_P(\lambda^*(P))\le1$ by definition, and therefore
\[
1\ge\psi_P(tv)\ge\E_P\left[\exp\left\{t\,v^\top g\right\}\ind{v^\top g\ge\zeta}\right]\ge \exp\left\{t\zeta\right\}P\big(v^\top g\ge\zeta\big)\ge c_0\,\exp\left(t\zeta\right).
\]
Rearranging gives $t\le\log(1/c_0)/\zeta$. Since this holds for every $P\in\sP$, we are through. 
\end{proof}

\subsection{Proof of \texorpdfstring{\cref{thm:matching.adj.sets-uniform}}{uniform CLT}}
\label{proof:thm:matching.adj.sets-uniform}

\begin{proof}
The proof mirrors that of \cref{thm:matching.adj.sets} in \cref{proof-of-main-CLT-with-bias-corr},
with the key modification being that we now extend the supporting results \cref{lemma:AIPW-approx,lemma:rate-of-double-conditional-mean,lemma:rate-of-nu,lemma:rate-of-lambda,lemma:rate-of-weights,lemma:Taylor-expand-weights,lemma:bias-corr-is-close-to-oracle} from \cref{app:supporting-results} 
uniformly over $\sP$.
We organize the proof in three steps:
(i)~uniform extensions of the supporting rate results,
(ii)~a uniform asymptotic linear expansion for
$\wh\tau^{\,\nam}_{bc}$,
and (iii)~a Berry--Esseen bound,
which together yield the desired conclusions. Note that each of the quantities $g$, $\lambda^*$, $w^*$, $\nu^*$, $\tauR$ and $\mathrm{V}_{\nam}$ will now include $P$ as an additional argument.

\paragraph{Uniform versions of the supporting results.}
Our proofs of \cref{lemma:AIPW-approx,lemma:rate-of-double-conditional-mean,lemma:rate-of-nu,lemma:rate-of-lambda,lemma:rate-of-weights,lemma:Taylor-expand-weights,lemma:bias-corr-is-close-to-oracle} primarily rely on three ingredients:
(a) H\"{o}lder's inequality,
(b) finite moment generating function of $g(X_{S_\cap};\,P)$ in a neighborhood of $\lambda^*(P)$, and (c) the $\o(n^{-1/4})$ rates on the nuisance functions. We can thus extend all of these supporting results to hold uniformly over $P\in\sP$ under the conditions of \cref{thm:matching.adj.sets-uniform}, as follows.
\begin{enumerate}
    \item \cref{lem:gwh} shows that
\[
  \sup_{P\in\sP}
  \|\wh g(X_{S_\cap}) - g(X_{S_\cap};\,P)\|_{L^q(\P_n)}
  \;=\; \o(n^{-1/4}).
\]
\item \cref{lem:lambdawh} shows that
\[
  \sup_{P\in\sP}
  \|\wh\lambda - \lambda^*(P)\|
  \;=\; \o(n^{-1/4}).
\]
\item \cref{lem:aipw} shows that
for each $k=1,\dots,K$,
\[
  \sup_{P\in\sP}
  \left|
    \P_n\!\bigl[w^*(X_{S_\cap};P)
      \bigl(\tauAIPWobs{k} - \tauAIPW{k}\bigr)
    \bigr]
  \right|
  \;=\; \o(n^{-1/2}),
\]
and the same holds with $(\wh w - w^*)$ replacing $w^*$.
\item \cref{lem:wwh} shows that
\[
  \sup_{P\in\sP}
  \|\wh w(X_{S_\cap}) - w^*(X_{S_\cap};\,P)\|_{L^2(\P_n)}
  \;=\; \o(n^{-1/4}).
\]
\item \cref{lem:nuwh} shows that
\[
  \sup_{P\in\sP}
  \|\wh\nu - \nu^*(P)\|
  \;=\; \o(n^{-1/4}).
\]
\item \cref{lem:taylor-weights} shows that
\begin{align*}
        &\P_n \left[\left(\wh{w}(X_{S_\cap})-w^*(X_{S_\cap};P)\right)\,\phi(X_{S_\cap};\nu^*(P),P) \right]\\
        &=\P_n\left[ w^*(X_{S_\cap};P)\left(\wh{g}(X_{S_\cap})^\top\,\wh{\lambda} -g(X_{S_\cap};P)^\top\lambda^*(P)\right)\,\phi(X_{S_\cap};\nu^*(P),P)\right]+R_n^{w}(P),
    \end{align*}
    where $\phi(X_{S_\cap};\nu,P):=\sum_{k=1}^K\nu_k\,\E_P[\tau(X_{S_k};P)\mid X_{S_\cap}]-\tauR(P)$, and the remainder term is uniformly asymptotically negligible: $\sup_{P\in\sP}|R_n^w(P)|=\o(n^{-1/2})$.
    \item \cref{lem:bias-corr-approx} shows that $$\sup_{P\in\sP}\left|B_n-\wh{B}_n^*(P)\right|=\o(n^{-1/2}),$$
    where
    $$\wh{B}_n^*(P)=\wh{\lambda}^\top \,\P_n\left[w^*(X_{S_\cap};P)(\Delta\tauAIPW{}-\wh{g}(X_{S_\cap}))\,\phi(X_{S_\cap};\nu^*(P),P)\right].$$ 
\end{enumerate}

\paragraph{Uniform asymptotic linear expansion.}
For each $P\in\sP$, define the influence function
\begin{equation}\label{eq:IF}
\begin{split}
     \psi^{\,\nam}(D;P)
  &:=
  w^*(X_{S_\cap};P)
  \,\Bigg[
    \sum_{k=1}^K \nu^*_k
    \bigl(\tauAIPW{k} - \tauR(P)\bigr)\\
    &\qquad\qquad\qquad\qquad\qquad+
    (\lambda^*)^\top
    \bigl(\Delta\tauAIPW{}- g(X_{S_\cap};P)\bigr)
    \phi(X_{S_\cap};\nu^*(P),P)
  \Bigg],
\end{split}
\end{equation}
where $\phi(X_{S_\cap};\nu,P):=\sum_{k=1}^K\nu_k\,\E_P[\tau(X_{S_k};P)\mid X_{S_\cap}]-\tauR(P)$.
It follows from \cref{assump-at-least-one-is-valid-uniform} that
$\E_P[\psi^{\,\nam}(D;\,P)] = 0$. 
The algebraic decompositions used in the proof of \cref{thm:matching.adj.sets} (cf.~\cref{proof-of-main-CLT-with-bias-corr}) carry over verbatim to give
\begin{equation}\label{eq:ALE}
  \wh\tau^{\,\nam}_{bc} - \tauR(P)
  \;=\;
  \P_n[\psi^{\,\nam}(D;P)]
  \;+\;
  r_n(P),
\end{equation}
where $r_n(P)$ is the sum of all the remainder terms from
the proof in \cref{proof-of-main-CLT-with-bias-corr}.
We show below that $\sup_{P\in\sP}\left|r_n(P)\right|=\o(n^{-1/2})$ by
examining each of these remainder terms in turn.

\begin{enumerate}
    \item \textit{AIPW approximation errors.}
These are the terms $\P_n[w^*(X_{S_\cap})(\tauAIPWobs{k}-\tauAIPW{k})]$
and the cross-term with $(\wh w - w^*)$;
both are $\o(n^{-1/2})$ uniformly by \cref{lem:aipw}. This tells us that the $\o(n^{-1/2})$ term in \eqref{emp-minus-oracle-tau-k} is asymptotically negligible uniformly in $P\in \sP$.

\item \textit{Oracle estimators aggregated with $\wh{\nu}-\nu^*$.} The remainder term in \eqref{AR-gap-AR-star-2} is uniformly asymptotically negligible since $\sup_{P\in\sP}\|\wh{\nu}-\nu^*(P)\|=\o(n^{-1/4})$ from \cref{lem:nuwh} and for each $P\in\sP$, we have $\E_P[w^*(X_{S_\cap};P)(\Delta\tauAIPW{}-g(X_{S_\cap};P))]=0$ by definition.

\item \textit{Sum of cross-terms $(\wh\nu-\nu^*)(\wh\tau^\mathtt{R}_k - \wh\tau^{\mathtt{R},*}_k)$.} We apply \cref{lem:wwh} to deduce that the remainder term in \eqref{rate-emp-minus-oracle-tau-k} is uniformly asymptotically negligible. This, combined with \eqref{emp-minus-oracle-tau-k} and \cref{lem:nuwh} tells us that the sum of the cross-terms in \eqref{AR-gap-AR-star-3} is uniformly asymptotically negligible. 

Starting from the decomposition in \eqref{AR-gap-AR-star-1} and combining the  above three items, we arrive at
\begin{align*}
     \wh\tau^\nam_{bc}-B_n- \sum_{k=1}^K \nu_k^*(P)\,\wh{\tau}_k^{\,\mathtt{R},*}(P) 
    &=\P_n \left[\left(\wh{w}(X_{S_\cap})-w^*(X_{S_\cap};P)\right) \phi(X_{S_\cap};\nu^*(P),P)\right]+\o(n^{-1/2}),
\end{align*}
where the remainder term is uniformly asymptotically negligible over $P\in\sP$.
\item \textit{Taylor expansion of transfer weights.} Using \cref{lem:taylor-weights}, we can continue from the last display to write
\begin{align*}
     &\wh\tau^\nam_{bc}-B_n- \sum_{k=1}^K \nu_k^*(P)\wh{\tau}_k^{\,\mathtt{R},*}(P) \\
    &=\P_n\left[ w^*(X_{S_\cap};P)\left(\wh{g}(X_{S_\cap})^\top\,\wh{\lambda} -g(X_{S_\cap};P)^\top\lambda^*(P)\right)\,\phi(X_{S_\cap};\nu^*(P),P)\right]+R_n^{w}(P),
    \end{align*}
    where $\sup_{P\in\sP}|R_n^w(P)|=\o(n^{-1/2})$.
\item \textit{Bias-correction approximation.} It follows from \cref{lem:lambdawh} that the remainder term in \eqref{Delta-lambda} is uniformly asymptotically negligible. Then, using \cref{lem:bias-corr-approx}, we deduce from the last display that the remainder term in \eqref{pre-ALR} is uniformly $\o(n^{-1/2})$, which gives us the desired asymptotic linear expansion as in \eqref{eq:ALE}.
\end{enumerate}

\paragraph{Berry--Esseen bound.}
As in the statement of \cref{thm:matching.adj.sets-uniform}, we define
$\mathrm{V}_{\nam}(P):=\var_P(\psi^{\,\nam}(D;\,P))$. We apply the Berry--Esseen theorem in its $(2+\delta)$-moment form (see, e.g., \citet[Chapter V, Theorem 5]{petrov}): For i.i.d.\ mean-zero variables $\psi^{\,\nam}(D_i\,;P)$ with
finite $(2+\delta)$-th moment,
\begin{equation}\label{eq:BE}
  \sup_{t\in\R}
  \left|
    P\!\left(
      \frac{\sqrt{n}\,\P_n[\psi^{\,\nam}(D;P)]}{\sqrt{\mathrm{V}_{\nam}(P)}}
      \le t
    \right)
    - \Phi(t)
  \right|
  \;\le\;
  \frac{C_\delta\,\E_P[|\psi^{\,\nam}(D;P)|^{2+\delta}]}
       {\mathrm{V}_{\nam}(P)^{(2+\delta)/2}\,n^{\delta/2}}.
\end{equation}
Bounding the supremum over ${P\in\sP}$ of the right-hand side requires establishing a uniform upper bound on
$\E_P[|\psi^{\,\nam}|^{2+\delta}]$, since we have $\inf_{P\in\sP}\mathrm{V}_{\nam}(P)\ge v_0>0$ by assumption. Write
$\Lambda := \sup_{P\in\sP}\|\lambda^*(P)\| < \infty$ (finite by \cref{propo:exist-n-unique-uniform}). Consider the decomposition
\[
\psi^{\mathrm{SR}}(D;P)=w^*(X_{S_\cap};P)\big[\Phi_1(D;P)+\Phi_2(D;P)\,\phi(X_{S_\cap};\nu^*(P),P)\big],
\]
where $\Phi_1:=\sum_k\nu_k^*(\aipw{k}-\tauR)$ and $\Phi_2:=(\lambda^*)^\top(\Delta\aipw{}-g)$. First, we show that 
\begin{equation}
    \label{Phi1Phi2}
    \sup_P\E_P[|\Phi_1|^q+|\Phi_2|^q]<\infty.
\end{equation}
Towards showing the above, note that the strong overlap condition implies that $$|\aipw{k}|\le(1+\xi^{-1})(|\mu_1(X_{S_k})|+|\mu_0(X_{S_k})|)+2\xi^{-1}|Y|.$$  Also, strong overlap together with the conditional Jensen inequality yields that $$\E_P|\mu_1(X_{S_k})|^q\le\E_P[A|Y|^q/e(X_{S_k})]\le\xi^{-1}\E_P|Y|^q,$$ and similarly for $\mu_0$. Note also that $g$ has uniformly bounded moments due to condition (4) of \cref{thm:matching.adj.sets-uniform}. $\nu^*(P)$ and $\lambda^*(P)$ are uniformly bounded by \cref{lem:nuwh} and \cref{propo:exist-n-unique-uniform}, respectively. Finally,  $\tauR(P)$ is uniformly bounded due to the inequality $|\tauR(P)|\le\|w^*\|_{L^2(P)}\|m_1\|_{L^2(P)}$. This finishes the proof of \eqref{Phi1Phi2}.

Second, observe that $\phi(X_{S_\cap};\nu^*(P),P)=m_1-g^\top\nu^*_{2:K}-\tauR$ has uniformly bounded moments of every order due to condition (4) of \cref{thm:matching.adj.sets-uniform}. Moreover, Jensen's inequality yields $$\psi_P(\lambda^*(P))\ge\exp(\lambda^*(P)^\top\E_P[g])\ge\exp(-\Lambda\sup_P\E_P\|g\|),$$ and therefore, for every $s<r/\Lambda$, we have
\[
\sup_{P\in\mathcal P}\E_P\big[w^*(X_{S_\cap};P)^s\big]\le\sup_{P\in\mathcal P}\frac{\E_P[e^{s\Lambda\|g\|}]}{\psi_P(\lambda^*(P))^s}<\infty.
\]
Now choose $\delta\in(0,q-2)$ small enough such that $s_\delta:=(2+\delta)q/(q-2-\delta)<r/\Lambda$. Such a choice is possible because $s_\delta\to2q/(q-2)<r/\Lambda$ as $\delta\downarrow0$, by condition (4) of \cref{thm:matching.adj.sets-uniform}. Then, applying Hölder's inequality with exponents $q/(2+\delta)$ and $q/(q-2-\delta)$ gives
\[
\E_P|\psi^{\mathrm{SR}}|^{2+\delta}\le2^{1+\delta}\Big\{(\E_P|\Phi_1|^q)^{\frac{2+\delta}{q}}(\E_P[(w^*)^{s_\delta}])^{\frac{q-2-\delta}{q}}+(\E_P|\Phi_2|^q)^{\frac{2+\delta}{q}}(\E_P[(w^*|\phi|)^{s_\delta}])^{\frac{q-2-\delta}{q}}\Big\}.
\]
One more application of Hölder's inequality, with $\eta>0$ such that $s_\delta(1+\eta)<r/\Lambda$, yields $$\E_P[(w^*|\phi|)^{s_\delta}]\le(\E_P[(w^*)^{s_\delta(1+\eta)}])^{1/(1+\eta)}(\E_P|\phi|^{s_\delta(1+\eta)/\eta})^{\eta/(1+\eta)}.$$The last two displays, taken together, imply that
\[
\sup_{P\in\mathcal P}\E_P\big[|\psi^{\mathrm{SR}}(D;P)|^{2+\delta}\big]\le M<\infty,
\]
for a constant $M$ independent of $P$.
Combining this with $\inf_{P\in\sP}\mathrm{V}_{\nam}(P)\ge v_0>0$, we deduce from~\eqref{eq:BE} that
\begin{equation}\label{eq:BE_unif}
  \sup_{P\in\sP}\,\sup_{t\in\R}
  \left|
    P\!\left(
      \frac{\sqrt{n}\,\P_n[\psi^{\,\nam}(D;P)]}{\sqrt{\mathrm{V}_{\nam}(P)}}
      \le t
    \right)
    - \Phi(t)
  \right|
  \;\le\;
  \frac{C_\delta M}{v_0^{(2+\delta)/2}\,n^{\delta/2}}
  \stackrel{\text{as }n\to\infty}{\longrightarrow} 0.
\end{equation}
To finish the proof, we apply \eqref{eq:ALE} and triangle inequality.
 For any $t\in\R$ and $P\in\sP$, 
\begin{align}
  &\left|
    P\!\left(
      \frac{\sqrt{n}\,(\wh\tau^{\,\nam}_{bc} - \tauR)}
           {\sqrt{\mathrm{V}_{\nam}(P)}}
      \le t
    \right)
    - \Phi(t)
  \right| \nonumber\\[2mm]
  &\le
  \left|
    P\!\left(
      \frac{\sqrt{n}\,\P_n[\psi^{\,\nam}]}{\sqrt{\mathrm{V}_{\nam}(P)}}
      \le t
    \right)
    - \Phi(t)
  \right|
     \,+\,
  \left|
    P\!\left(
      \frac{\sqrt{n}\,(\wh\tau^{\,\nam}_{bc} - \tauR)}
           {\sqrt{\mathrm{V}_{\nam}(P)}}
      \le t
    \right)
    -
    P\!\left(
      \frac{\sqrt{n}\,\P_n[\psi^{\,\nam}]}{\sqrt{\mathrm{V}_{\nam}(P)}}
      \le t
    \right)
  \right|
  \label{eq:triangle}\\[2mm]
  &=: \Delta_{n,1}(t,P) + \Delta_{n,2}(t,P).
  \nonumber
\end{align}
Note that \eqref{eq:BE_unif} implies
$\sup_{P\in \sP}\sup_t \Delta_{n,1}(t,P) \to 0$.
For the other term, we use \eqref{eq:ALE} and
$\inf_{P}\mathrm{V}_{\nam}(P) \ge v_0 > 0$ to deduce, for any fixed $\eps>0$, that
\[
  \sup_{P\in\sP} \Delta_{n,2}(t,P)
  \;\le\;
  \sup_{P\in\sP}
  P\!\left(
    \frac{\sqrt{n}\,|r_n(P)|}{\sqrt{v_0}} > \eps
  \right)
  + \sup_{t\in\R}\bigl[\Phi(t+\eps)-\Phi(t-\eps)\bigr]
  \;\to\; C'\eps,
\]
as $n\to\infty$, using $\sup_{P\in \sP}\sqrt{n}\,|r_n(P)|\Pto 0$
and the continuity of $\Phi$.
Since $\eps > 0$ is arbitrary, taking $\sup_{P\in\sP}\sup_{t\in\R}$
in~\eqref{eq:triangle} and sending $\eps\to 0$ gives
\[
  \lim_{n\to\infty}
  \sup_{P\in\sP}\,
  \sup_{t\in\R}
  \left|
    P\!\left(
      \frac{\sqrt{n}\,(\wh\tau^{\,\nam}_{bc}-\tauR(P))}
           {\sqrt{\mathrm{V}_{\nam}(P)}}
      \le t
    \right)
    - \Phi(t)
  \right|
  = 0,
\]
as desired to show.
\end{proof}

\section{Supporting results}\label{app:supporting-results}

\begin{lemma}[Rate of $\wh{g}$]\label{lemma:rate-of-double-conditional-mean}
    Under the conditions of \cref{thm:matching.adj.sets}, it holds that $$\|\wh{g}(X_{S_\cap}) -g(X_{S_\cap})\|_{L^q(\P_n)}=\o(n^{-1/4}).$$
\end{lemma}

\begin{proof}[Proof of \cref{lemma:rate-of-double-conditional-mean}]
The proof simply follows using triangle inequality, since we assume that $\|\wh{\E}_P[\wh{\tau}(X_{S_k})\mid X_{S_\cap}] - \E_P[\tau(X_{S_k})\mid X_{S_\cap}]\|_{L^q(\P_n)}=\o(n^{-1/4})$ for each $k=1,\dots,K$.
\end{proof}

\begin{lemma}[Rate of $\wh\lambda$]\label{lemma:rate-of-lambda}
    Under the conditions of \cref{thm:matching.adj.sets}, the empirical optimization problem \eqref{eq:lambda-hat-def} has a unique, finite minimizer $\wh\lambda$ with probability tending to $1$. Moreover, $\wh{\lambda}-\lambda^*=\O(n^{-1/2})+\O(\|\wh{g}-g\|_{L^2(\P_n)})=\o(n^{-1/4})$.
\end{lemma}

\begin{proof} Recall that $\wh{\lambda}:=\argmin_\lambda \wh{L}_n(\lambda)$ and ${\lambda}^*:=\argmin_\lambda {L}(\lambda)$ where $$\wh{L}_n(\lambda):=\P_n \left[\exp\left(\wh{g}(X_{S_\cap})^\top \lambda\right)\right],\quad L(\lambda)=\E_P \left[\exp\left(g(X_{S_\cap})^\top \lambda\right)\right].$$
Condition (iii) of Theorem 3.2 states that $g$ and $m_1$ have finite MGF in a radius $r>2q'\|\lambda^*\|$ around the origin, where $q'=q/(q-2)$, and the same holds for $\wh{g}$ in the conditional sense: $\E_P[\exp(\rc\|\wh{g}(X_{S_\cap})\|)\mid \wh{g}(\cdot)]=\O(1)$. Since $\wh{g}(\cdot)$ is fitted independent of the data, the conditional Markov
inequality and the Markov inequality give
\begin{equation}
    \label{patch:D2-moments}
    \P_n \left[\exp\left(r\|g(X_{S_\cap})\|\right)+\exp\left(r\|\wh{g}(X_{S_\cap})\|\right)\right]=\O(1).
\end{equation}
Throughout this proof, we fix $\eps_0 > 0$ such that $2(\|\lambda^*\| + \eps_0) < \rc$. Define
$$A:=\E_P[ e^{-(\|\lambda^*\|+\eps_0)\|g(X_{S_\cap})\|}g(X_{S_\cap})\,g(X_{S_\cap})^\top].$$
    Note that \cref{assump:heterogeneity} (i) implies that $\E_P[gg^\top]$ is finite and positive definite:~Finiteness follows from the finiteness of the moment generating function of $g$, and if there exists $v\neq 0$ such that $0=v^\top\E_P[gg^\top]v =\E_P[(v^\top g)^2]$ then $v^\top g = 0$ a.s., violating \cref{assump:heterogeneity} (i). This, in turn, implies that $A\succ 0$ because if $v\neq 0$ is such that $v^\top Av = 0$, then using that the exponential function is strictly positive, we have $g^\top v = 0$ a.s., which contradicts $\E_P[gg^\top]\succ 0$.  Throughout, we denote by $\kappa>0$ the smallest eigenvalue of $A$.

First, we prove that with probability tending to one, $\wh{L}_n$ has a unique minimizer $\wh\lambda$ and $\wh{\lambda}\Pto \lambda^*$. We claim that, for any $0 < \eps \le \eps_0$ and with $B_\eps:=\{\lambda:\|\lambda-\lambda^*\|\,\le\,\eps\}$,
\begin{enumerate}
    \item\label{cond1} $\sup_{\lambda\in B_\eps}|\wh{L}_n(\lambda)-L(\lambda)|=\o(1)$,
    \item\label{cond2}  $L(\lambda)\ge L(\lambda^*)+\kappa\eps^2/2$ holds for  every $\lambda$ on the boundary $\|\lambda-\lambda^*\|=\eps$.
    \item\label{cond3} $\nabla^2 L(\lambda)\succeq A$ for every $\lambda\in B_\eps$, and $\sup_{\lambda\in B_\eps} \|\nabla^2\wh{L}_n(\lambda)-\nabla^2L(\lambda)\|=\o(1)$. Consequently, with probability tending to one, $\nabla^2 \wh{L}_n(\lambda)\succeq (\kappa/2) I$ for every $\lambda\in B_\eps$.
\end{enumerate}
Using the above, we can show $\wh{\lambda}\Pto \lambda^*$ as follows. For any $\lambda'$ such that $\|\lambda'-\lambda^*\|=\eps$, we have
\begin{align*}
    \wh{L}_n(\lambda') \ge L(\lambda') - \sup_{\lambda\in B_\eps} |\wh{L}_n(\lambda)-L(\lambda)|
    &\ge L(\lambda^*) +\frac{1}{2}\kappa\eps^2 - \sup_{\lambda\in B_\eps} |\wh{L}_n(\lambda)-L(\lambda)|\tag{using \cref{cond2}}\\[2mm]
    &\ge \wh{L}_n(\lambda^*) + \frac{1}{2}\kappa\eps^2 - 2\sup_{\lambda\in B_\eps}|\wh{L}_n(\lambda)-L(\lambda)|.
\end{align*}
This in conjunction with \cref{cond1} implies that the event $$E_{n,\eps}:=\left\{\inf_{\lambda\,:\,\|\lambda-\lambda^*\|\,=\,\eps}\wh{L}_n(\lambda)>\wh{L}_n(\lambda^*)\right\}$$ has probability tending to one.
The empirical objective $\wh{L}_n$ is a finite sum of exponentials, hence continuous and convex. Therefore, on the event $E_{n,\eps}$, its minimum over the compact set $B_\eps$ is attained at an interior point. Since a local minimizer of a convex function is a global minimizer, we conclude that, on the event $E_{n,\eps}$,  $\wh{L}_n$ has a global minimizer in the interior of $B_\eps$.
On the event in \cref{cond3}, $\wh{L}_n$ is strictly convex on $B_{\eps_0}$ hence the minimizer is unique.

Thus, with probability tending to one, $\wh{\lambda}$ is well defined, unique, and satisfies $\|\wh{\lambda} - \lambda^*\| < \eps$.
Since $0<\eps\le \eps_0$ is arbitrary, the above proves the consistency. It remains to show \cref{cond1,cond2,cond3}.

\begin{enumerate}
    \item Proof of \cref{cond1}: First,  set $\wt{L}_n(\lambda)=\P_n [\exp(g(X_S)^\top \lambda)]$ and write $$\wh{L}_n(\lambda)-L(\lambda) =\wh{L}_n(\lambda)-\wt{L}_n(\lambda) +\wt{L}_n(\lambda)-L(\lambda).$$
The conclusion $\sup_{\lambda\in B_\eps} |\wt{L}_n(\lambda)-L(\lambda)|=\o(1)$ follows from \citet[Theorem II.1]{Anderson1982}; which applies here because $\wt{L}_n$ is convex and converges pointwise (by the law of large numbers) to the finite
function $L$ on the open ball of radius $\rc$ around the origin, which contains the compact set $B_{\eps}$. 
To show that $\sup_{\lambda\in B_\eps} |\wh{L}_n(\lambda)-\wt{L}_n(\lambda)|=\o(1)$, we rely on the following Taylor expansion:
\begin{equation}
    \label{taylorXi}
    e^{\wh{g}(X_{S_\cap})^\top \lambda}-e^{g(X_{S_\cap})^\top \lambda}=\lambda^\top (\wh{g}(X_{S_\cap})-g(X_{S_\cap}))e^{\Xi(X_{S_\cap})^\top \lambda},
\end{equation}
for some $\Xi(X_{S_\cap})$ that lies between $g(X_{S_\cap})$ and $\wh{g}(X_{S_\cap})$. For $\lambda \in B_\eps$, we have $\|\lambda\|\le\|\lambda^*\| + \eps$ and $\|\Xi\| \le \max\{\|g\|,\|\wh{g}\|\}$, hence 
\begin{equation}
    \label{Xitogghat}
    \exp(\lambda^\top \Xi)\le \exp((\|\lambda^*\|+\eps)\|g\|)+\exp((\|\lambda^*\|+\eps)\|\wh{g}\|)
\end{equation}
We therefore deduce that
\begin{align*}
  & \sup_{\lambda\in B_\eps}|\wh{L}_n(\lambda)-\wt{L}_n(\lambda)|=\sup_{\lambda\in B_\eps}\left| \P_n \left[e^{\wh{g}(X_{S_\cap})^\top \lambda}-e^{g(X_{S_\cap})^\top \lambda}\right] \right|\\
  &\le \sup_{\lambda\in B_\eps}\left| \P_n \left[ |\lambda^\top (\wh{g}(X_{S_\cap})-g(X_{S_\cap}))| \cdot e^{\lambda^\top \Xi(X_{S_\cap})} \right] \right|\tag{using \eqref{taylorXi} and triangle inequality}\\
   &\le (\|\lambda^*\|+\eps)\left\|\|\wh{g}-g\|_{L^2(\P_n)}\right\|_2 \left\|e^{ (\|\lambda^*\|+\eps)\|g\|_2}+e^{ (\|\lambda^*\|+\eps)\|\wh{g}\|_2}\right\|_{L^2(\P_n)}. \tag{using Cauchy-Schwarz inequality and \eqref{Xitogghat}}
\end{align*}
The above bound is $\o(n^{-1/4})$, since $\|\wh{g}(X_{S_\cap})-g(X_{S_\cap})\|_{L^2(\P_n)}=\o(n^{-1/4})$ from \cref{lemma:rate-of-double-conditional-mean}, and the last factor is $\O(1)$ by \eqref{patch:D2-moments}.
 \qed

    \item Proof of \cref{cond2}: Applying the dominated convergence theorem with the integrable envelope $e^{(\|\lambda^*\|+\eps_0)\|g\|}(1 + \|g\|)^2$ (this is integrable since $\|\lambda^*\| + \eps_0 < r$), we can say that $L$ is twice continuously differentiable on $B_{\eps}$, with 
    \[\nabla L(\lambda) = \E_P[\exp(g^\top \lambda)g],\quad \text{and}\quad \nabla^2 L(\lambda) = \E_P[\exp(g^\top \lambda)gg^\top].\]
    It follows from a Taylor expansion of $L(\lambda)$ around $\lambda^*$ that
    $$L(\lambda)=L(\lambda^*) + \underbrace{\nabla L(\lambda^*)^\top}_{=\,0} (\lambda-\lambda^*) + \frac12 (\lambda-\lambda^*)^\top \nabla^2 L(\xi) (\lambda-\lambda^*),$$
    for some $\xi$ on the segment joining $\lambda^*$ and $\lambda$. In the last display we used $\nabla L(\lambda^*) = 0$  from the proof of \cref{propo:exist-n-unique} in \cref{app:proof-of-propo:exist-n-unique}. Since $\|\xi-\lambda^*\|\le\|\lambda-\lambda^*\|\le \eps$, we have $\|\xi\| \le \|\lambda^*\| + \eps_0$, and thus $\exp(\xi^\top g) \ge \exp(-(\|\lambda^*\|+\eps_0)\|g\|)$ (pointwise). This implies that $\nabla^2 L(\xi) \succeq A$, where $A:=\E_P[ \exp(-(\|\lambda^*\|+\eps_0)\|g\|)\,gg^\top]\succeq \kappa I$ as shown earlier.
    Consequently, $$L(\lambda)- L(\lambda^*) \ge \frac12 (\lambda-\lambda^*)^\top A (\lambda-\lambda^*)\ge \frac12\kappa \|\lambda-\lambda^*\|^2=\frac12\kappa \eps^2,$$
    as desired to show.\qed

    \item Proof of \cref{cond3}:  The bound $\nabla^2 {L}(\lambda) \succeq {A}$ on $B_{\eps_0}$ follows exactly as in the proof of \cref{cond2} above. For the uniform convergence, we start by writing that
\begin{equation}
    \label{obviodecomp}
    \nabla^2 \wh{L}_n(\lambda)-\nabla^2 L(\lambda)=\left(\nabla^2 \wh{L}_n(\lambda)-\nabla^2 \wt{L}_n(\lambda)\right)+\left(\nabla^2 \wt{L}_n(\lambda)-\nabla^2 L(\lambda)\right) .
\end{equation}
We will show that each of the above differences is $\o(1)$.  For $\|\lambda\|\le \|\lambda^*\|+\eps_0$, the derivative of the map $v\mapsto \exp(v^\top \lambda)\,vv^\top$ in a direction $h$ equals $$\exp(v^\top \lambda)((h^\top \lambda)\,vv^\top +vh^\top + hv^\top),$$ and the norm of the above is bounded by $c(1 + \|v\|)^2\exp\{(\|\lambda^*\|+\eps_0)\|v\|\}\|h\|$ for a constant $c\in (0,\infty)$.
Therefore, by the mean value theorem and the Cauchy–Schwarz inequality,
$$
\sup_{\lambda \in B_{\eps_0}}\left\|\nabla^2 \wh{L}_n(\lambda)-\nabla^2 \wt{L}_n(\lambda)\right\| \lesssim \left\|(1+\|g\|\vee\|\wh{g}\|)^2 e^{(\|\lambda^*\|+\eps_0)\|g\|\vee\|\wh{g}\|}\right\|_{L^2(\P_n)}\|\wh{g}-g\|_{L^2(\P_n)} .
$$
The first factor above is $\O(1)$ by \eqref{patch:D2-moments} and the second factor is $\o(n^{-1/4})$ by \cref{lemma:rate-of-double-conditional-mean}.

Finally, the conclusion for the second term in \eqref{obviodecomp} follows from a uniform law of large numbers \citep[Lemma 2.4]{neweyMcFadden1994}. Indeed, observe that $B_{\eps_0}$ is compact, the map $\lambda\mapsto \exp(g^\top \lambda)\, gg^\top $ is continuous with an integrable envelope $\exp\{(\|\lambda^*\|+\eps_0)\|g\|\}\|g\|^2$. Therefore, \citet[Lemma 2.4]{neweyMcFadden1994} applies and yields that
$$\sup_{\lambda \in B_{\eps_0}}\left\|\nabla^2 \wt{L}_{n}(\lambda)-\nabla^2 {L}(\lambda)\right\|=\o(1).$$  
Finally, for every $\lambda \in B_{\eps_0}$, $\lambda_{\min}(\nabla^2\wh{L}_n(\lambda)) \ge \lambda_{\min}(\nabla^2 L(\lambda)) - \|\nabla^2\wh{L}_n(\lambda) - \nabla^2 L(\lambda)\| \ge \kappa - \o(1)$ uniformly in $\lambda$, which gives $\nabla^2\wh{L}_n(\lambda) \succeq (\kappa/2)I$ on $B_{\eps_0}$ with probability tending to one.
This finishes the proof of \cref{cond3} and thus of the consistency $\wh\lambda\Pto\lambda^*$. 
\qed
\end{enumerate}

We next show that $\wh{\lambda}-\lambda^*=\o(n^{-1/4})$. We work on the event where $\|\wh\lambda-\lambda^*\|\le \eps_0$ and the conclusion of \cref{cond3} holds; we have shown above that this event has probability tending to one.
Since $\nabla\wh{L}_n(\wh{\lambda}) = 0$ and the segment between $\lambda^*$ and $\wh\lambda$ lies in $B_{\eps_0}$, the integral form of Taylor's theorem gives
$$
0=\nabla \wh{L}_n(\wh{\lambda})=\nabla \wh{L}_n\left(\lambda^*\right)+\overline{H}_n\cdot(\wh{\lambda}-\lambda^*),\quad \overline{H}_n := \int_0^1 \nabla^2 \wh{L}_n(\lambda^*+s(\wh{\lambda}-\lambda^*))\, ds\succeq \frac{\kappa}{2} I.
$$
Rearranging, we get
$
\|\wh{\lambda}-\lambda^*\|\le ({2}/{\kappa}) \|\nabla \wh{L}_n(\lambda^*)\|
$.
Note that
\begin{equation*}
  \nabla \wh{L}_n\left(\lambda^*\right)=\P_n\left[\exp\big(g^{\top} \lambda^*\big) g\right]
+\P_n\left[\exp \big(\wh{g}^{\top} \lambda^*\big) \wh{g}-\exp \big(g^{\top} \lambda^*\big) g\right].
\end{equation*}
The first term above is $\O(n^{-1/2})$ since $\E_P [\exp({g}^\top \lambda^*){g}]=0$ and the variance is finite as $\rc> 2\|\lambda^*\|$. 
For the second term, we can show that the derivative of the map $v\mapsto \exp(v^\top \lambda^*)v$ has norm at most $(\|\lambda^*\| + 1)(1 + \|v\|)\exp(\|\lambda^*\|\|v\|)$. Therefore, using the mean value theorem and the Cauchy-Schwarz inequality, we can bound the second term above as
\begin{align*}
    &\left|\P_n\left[\exp \left(\wh{g}^{\top} \lambda^*\right) \wh{g}-\exp \left(g^{\top} \lambda^*\right) g\right]\right|\\
    &\le (\|\lambda^*\|+1)\left\|(1+\|g\|\vee\|\wh{g}\|)\exp\left\{\|\lambda^*\|\|g\|\vee\|\wh{g}\|\right\}\right\|_{L^2(\P_n)}\|\wh{g} - g\|_{L^2(\P_n)}.
\end{align*}
Again, the last factor above is $\o(n^{-1/4})$ by \cref{lemma:rate-of-double-conditional-mean}, and the rest is $\O(1)$ as in \eqref{patch:D2-moments}.
Consequently, $\nabla\wh{L}_n(\lambda^*)=\o(n^{-1/4})$ and hence $\|\wh\lambda-\lambda^*\|=\o(n^{-1/4})$.
\end{proof}

\begin{lemma}[Rate of estimated transfer weights $\wh{w}$]\label{lemma:rate-of-weights}
     Under the conditions of \cref{thm:matching.adj.sets},
         $$\|\wh{w}(X_{S_\cap})-w^*(X_{S_\cap})\|_{L^2(\P_n)} =\o(n^{-1/4}),\quad 
         \text{and}\quad \|\wh{w}(X_{S_\cap})-w^*(X_{S_\cap})\|_{L^{t}(\P_n)}=\O(1),$$
        for any $t\in [1,\, \rc/\|\lambda^*\|)$, with the convention $\rc/\|\lambda^*\|=\infty$ if $\lambda^*=0$. In particular, the last claim holds for  $t =2q'$ where $q'=q/(q-2)$, because $\rc>2q'\|\lambda^*\|$ by assumption in \cref{thm:matching.adj.sets}.
\end{lemma}
\begin{proof}
Throughout, we use the shorthand $\Delta_n(x):=\wh{g}(x)^\top \wh{\lambda} - g(x)^\top \lambda^*$ for $x\in\R^{K-1}$. 
We first record moment bounds that are used repeatedly below. Fix $t\in[1,\rc/\|\lambda^*\|)$, and choose $\eps_t>0$ such that $t(\|\lambda^*\|+\eps_t)\le \rc$. By \cref{lemma:rate-of-lambda}, the event $E_n$ on which $\wh L_n$ has a unique minimizer $\wh\lambda$ with $\|\wh\lambda-\lambda^*\|\le\eps_t$ has probability tending to one. On this event $E_n$, we have $t\|\wh\lambda\|\le \rc$ by triangle inequality, and therefore
\[
e^{t\,g^\top\lambda^*}+e^{t\,\wh g^\top\wh\lambda}\le e^{r\|g\|}+e^{r\|\wh g\|}.
\]
Moreover, if $\xi(x)$ lies between $0$ and $\Delta_n(x)$, then $g(x)^\top\lambda^*+\xi(x)$ lies between $g(x)^\top\lambda^*$ and $\wh g(x)^\top\wh\lambda$, which gives $e^{g^\top\lambda^*+\xi}\le e^{g^\top\lambda^*}+e^{\wh g^\top\wh\lambda}$. Combining these two observations with $(a+b)^t\le2^{t-1}(a^t+b^t)$ and with \eqref{patch:D2-moments} from the proof of \cref{lemma:rate-of-lambda}, we conclude that on the event $E_n$,
\begin{equation}\label{patch:D3-a}
\P_n\Big[e^{t\,(g(X_{S_\cap})^\top\lambda^*+\xi(X_{S_\cap}))}\Big]\le2^{t-1}\,\P_n\big[e^{r\|g(X_{S_\cap})\|}+e^{r\|\wh g(X_{S_\cap})\|}\big]=\O(1),
\end{equation}
for every $\xi$ in $[0, \Delta_n]$ (or $[\Delta_n, 0]$ if $\Delta_n<0$). Note that $\P_n[e^{\wh g^\top\wh\lambda}]\le\P_n[e^{\wh g^\top0}]=1$ by definition. Using Jensen's inequality and the inequality $u\le e^{ru}/r$ for $u\ge0$, we obtain
\[
\P_n\big[e^{\wh g^\top\wh\lambda}\big]\ge\exp\big(\wh\lambda^\top\P_n[\wh g]\big)\ge\exp\big(-\|\wh\lambda\|\,\P_n\|\wh g\|\big)\ge\exp\big(-r^{-1}\,\|\wh\lambda\|\,\P_n[e^{r\|\wh g\|}]\big).
\]
Using $\|\wh\lambda\|=\O(1)$ and \eqref{patch:D2-moments} we conclude that
\begin{equation}\label{patch:D3-b}
1\big/\P_n\big[e^{\wh g(X_{S_\cap})^\top\wh\lambda}\big]=\O(1).
\end{equation}
The second claim of this lemma now follows from the triangle inequality, because
\[
\|\wh w\|_{L^t(\P_n)}=\frac{(\P_n[e^{t\,\wh g^\top\wh\lambda}])^{1/t}}{\P_n[e^{\wh g^\top\wh\lambda}]},\qquad \|w^*\|_{L^t(\P_n)}=\frac{(\P_n[e^{t\,g^\top\lambda^*}])^{1/t}}{\psi_P(\lambda^*)},
\]
and both are $\O(1)$ by \eqref{patch:D3-a} and \eqref{patch:D3-b}. We now turn to the first claim. By Taylor expansion, there exists a (point-dependent) $\xi(x)$ lying between $0$ and $\Delta_n(x)$ such that    
\begin{equation}\label{taylor}
    e^{\wh{g}(x)^\top \wh{\lambda} }-e^{ g(x)^\top \lambda^*}=e^{g(x)^\top \lambda^*} \Delta_n(x)+R_n(x), \quad\text{where}\quad
R_n(x):=\frac{1}{2} e^{g(x)^\top \lambda^*+\xi(x)} \Delta_n(x)^2.
\end{equation}
First, we control the remainder term, showing it is asymptotically negligible. Applying Hölder's inequality with $1/q'+2/q=1$, we deduce that
\begin{equation}\label{patch:D3-rem}
\P_n[R_n(X_{S_\cap})]\le\tfrac12\big\|e^{g(X_{S_\cap})^\top\lambda^*+\xi(X_{S_\cap})}\big\|_{L^{q'}(\P_n)}\,\|\Delta_n(X_{S_\cap})\|^2_{L^q(\P_n)}.
\end{equation}
The first factor above is $\O(1)$ by \eqref{patch:D3-a} with $t=q'$. For the second factor, note that
\[
\|\Delta_n\|_{L^q(\P_n)}\le\|\wh\lambda-\lambda^*\|\,\|g\|_{L^q(\P_n)}+\|\wh\lambda\|\,\|\wh g-g\|_{L^q(\P_n)}=\o(n^{-1/4}),
\]
by \cref{lemma:rate-of-double-conditional-mean,lemma:rate-of-lambda}. Therefore, we can conclude from \eqref{patch:D3-rem} that $\P_n[R_n(X_{S_\cap})]=\o(n^{-1/2})$.

\medskip
Next, we control the first-order term in \eqref{taylor}. Write
\[
\P_n\big[e^{g^\top\lambda^*}\Delta_n\big]=(\wh\lambda-\lambda^*)^\top\P_n\big[e^{g^\top\lambda^*}g\big]+\wh\lambda^\top\P_n\big[e^{g^\top\lambda^*}(\wh g-g)\big],
\]
Here $\P_n[e^{g^\top\lambda^*}g]$ is $\O(n^{-1/2})$ by Chebyshev's inequality, since $\E_P[e^{g^\top\lambda^*}g]=0$ and $\E_P[e^{2g^\top\lambda^*}\|g\|^2]<\infty$. The second term is bounded above (in absolute value) by the following:
\[
\|\wh\lambda\|\,\|e^{g^\top\lambda^*}\|_{L^2(\P_n)}\|\wh g-g\|_{L^2(\P_n)}=\o(n^{-1/4}).
\]
Combining with \eqref{taylor} and the bound on the remainder, we obtain
\begin{equation}\label{patch:D3-first}
\P_n\big[e^{\wh g(X_{S_\cap})^\top\wh\lambda}-e^{g(X_{S_\cap})^\top\lambda^*}\big]=\o(n^{-1/4}).
\end{equation}

\medskip

We now turn to the main proof, showing that $\|\wh w-w^*\|_{L^2(\P_n)}=\o(n^{-1/4})$. The triangle inequality gives
\begin{equation}\label{patch:D3-split}
\|\wh w-w^*\|_{L^2(\P_n)}\le\frac{\|e^{\wh g^\top\wh\lambda}-e^{g^\top\lambda^*}\|_{L^2(\P_n)}}{\P_n[e^{\wh g^\top\wh\lambda}]}+\frac{|\E_P[e^{g^\top\lambda^*}]-\P_n[e^{\wh g^\top\wh\lambda}]|}{\P_n[e^{\wh g^\top\wh\lambda}]}\,\|w^*\|_{L^2(\P_n)}.
\end{equation}
For the second term in the last display, $\|w^*\|_{L^2(\P_n)}=\O(1)$ by \eqref{patch:D3-a}, the denominator is handled by \eqref{patch:D3-b}, and
\[
\big|\P_n[e^{\wh g^\top\wh\lambda}]-\E_P[e^{g^\top\lambda^*}]\big|\le\big|\P_n[e^{\wh g^\top\wh\lambda}-e^{g^\top\lambda^*}]\big|+\big|(\P_n-\E_P)[e^{g^\top\lambda^*}]\big|=\o(n^{-1/4})+\O(n^{-1/2}),
\]
using \eqref{patch:D3-first} and Chebyshev's inequality. For the first term in \eqref{patch:D3-split}, a first-order Taylor expansion gives $$e^{\wh g(x)^\top\wh\lambda}-e^{g(x)^\top\lambda^*}=e^{g(x)^\top\lambda^*+\eta(x)}\Delta_n(x),$$ for some $\eta(x)$ lying between $0$ and $\Delta_n(x)$. Applying Hölder's inequality with $1/2=1/(2q')+1/q$,
\[
\big\|e^{g^\top\lambda^*+\eta}\Delta_n\big\|_{L^2(\P_n)}\le\big\|e^{g^\top\lambda^*+\eta}\big\|_{L^{2q'}(\P_n)}\,\|\Delta_n\|_{L^q(\P_n)}=\O(1)\,\o(n^{-1/4}),
\]
where we used \eqref{patch:D3-a} with $t=2q'$ and the bound on $\|\Delta_n\|_{L^q(\P_n)}$ derived above.  Therefore, both terms in \eqref{patch:D3-split} are $\o(n^{-1/4})$. This, together with \eqref{patch:D3-b}, completes the proof.
\end{proof}

\begin{lemma}[Rate of the affine weights $\wh\nu$]\label{lemma:rate-of-nu}
    Under the  conditions of \cref{thm:matching.adj.sets}, the matrix $\P_n[\wh{w}\,\wh{g}\,\wh{g}^\top]$ is invertible with probability tending to one, and it holds that $\wh{\nu}-\nu^*=\o(n^{-1/4})$, where $$\nu_{2:K}^* = \left(\E_P\left[ w^*(X_{S_\cap})\,g(X_{S_\cap})\,g(X_{S_\cap})^\top\right]\right)^{-1}\E_P\left[ w^*(X_{S_\cap})\,g(X_{S_\cap})\left(\E_P[\tau(X_{S_1})\mid X_{S_\cap}]-\tauR \right)\right],$$
    and $\nu^*_1 = 1-\sum_{k=2}^K \nu^*_{k}$.
\end{lemma}

\begin{proof} 
Recall that we compute $\wh{\nu}$ as $\wh\nu_1 := 1-\sum_{k=2}^K \wh\nu_{k}$ and $$\wh\nu_{2:K} := \left(\P_n\left[\wh{w}(X_{S_\cap})\wh{g}(X_{S_\cap})\wh{g}(X_{S_\cap})^\top\right]\right)^{-1}\P_n\left[ \wh{w}(X_{S_\cap})\wh{g}(X_{S_\cap})\left(\wh{\E}_P[\wh\tau(X_{S_1})\mid X_{S_\cap}]-\wh\tau_1^\mathtt{R}  \right)\right].$$
Note that
$\|\wh\nu_{2:K}-\nu^*_{2:K}\|_2\le \|D_n^{-1} (C_n-C)\|_2 + \|(D_n^{-1}-D^{-1})C\|_2$, where 
\begin{align*}
    C_n&=\P_n\left[ \wh{w}(X_{S_\cap})\wh{g}(X_{S_\cap})\left(\wh{\E}_P[\wh\tau(X_{S_1})\mid X_{S_\cap}]-\wh\tau_1^\mathtt{R}  \right)\right],\qquad D_n = \P_n\left[ \wh{w}(X_{S_\cap})\wh{g}(X_{S_\cap})\wh{g}(X_{S_\cap})^\top\right],\\ 
    C &= \E_P\left[ w^*(X_{S_\cap})\,g(X_{S_\cap})\left(\E_P[\tau(X_{S_1})\mid X_{S_\cap}]-\tauR\right) \right],
    \quad\ \
    D = \E_P\left[ w^*(X_{S_\cap}){g}(X_{S_\cap}){g}(X_{S_\cap})^\top\right].
\end{align*}
Recall that 
 $\|\wh{\E}_P[\wh\tau(X_{S_k})\mid X_{S_\cap}]-\E_P[\tau(X_{S_k})\mid X_{S_\cap}]\|_{L^2(\P_n)}=\o(n^{-1/4})$ by assumption. On the other hand, it follows from \cref{lemma:rate-of-double-conditional-mean} that $\|\wh{g}(X_{S_\cap})-g(X_{S_\cap})\|_{L^2(\P_n)}=\o(n^{-1/4})$ and \cref{lemma:rate-of-weights} gives $\|\wh{w}(X_{S_\cap})-w^*(X_{S_\cap})\|_{L^2(\P_n)}=\o(n^{-1/4})$. Combining these rate results with the algebraic identity
     \begin{multline}
         \label{eq:abc}
         \hat{a}\hat{b}\hat{c}-abc = ab(\hat{c}-c)+bc(\hat{a}-a)+ca(\hat{b}-b)+ (\hat{a}-a)(\hat{b}-b)c\\
        +(\hat{b}-b)(\hat{c}-c)a+(\hat{c}-c)(\hat{a}-a)b+(\hat{a}-a)(\hat{b}-b)(\hat{c}-c),
     \end{multline}
 we can show that $\|C_n-C\|_2=\o(n^{-1/4})$ as well as $\|D_n^{-1}-D^{-1}\|_2 = \o(n^{-1/4})$. For instance, 
\begin{align*}
    &\left\|\P_n \left[w^*(X_{S_\cap})\,g(X_{S_\cap})(\wh{\E}_P[\wh\tau(X_{S_1})\mid X_{S_\cap}]-\wh\tau_1^\mathtt{R} -\E_P[\tau(X_{S_1})\mid X_{S_\cap}]-\tauR)\right]\right\|_2\\[2mm]
    &\le  \left\| \| w^*(X_{S_\cap}) g(X_{S_\cap})\|_{L^2(\P_n)}\right\|_2\|\wh{\E}_P[\wh\tau(X_{S_1})\mid X_{S_\cap}]-\E_P[\tau(X_{S_1})\mid X_{S_\cap}]\|_{L^2(\P_n)}\\[2mm]
    &\qquad\qquad+|\wh\tau_1^\mathtt{R} -\tauR| \|\P_n [w^*(X_{S_\cap})\,g(X_{S_\cap})]\|_2 \tag{by triangle inequality and Cauchy-Schwarz}\\[2mm]
    &=\O(1)\o(n^{-1/4})+\o(n^{-1/4})\O(n^{-1/2})\tag{using \eqref{rate-emp-minus-oracle-tau-k} and finite moments}\\[2mm]
    &=\o(n^{-1/4}).
\end{align*}
Similarly, using the notation $m_1(X_{S_\cap})=\E_P[\tau(X_{S_1})\mid X_{S_\cap}]$, 
\begin{align*}
     &\left\|\P_n \left[w^*(X_{S_\cap})(\wh{g}(X_{S_\cap})-g(X_{S_\cap}))(\E_P[\tau(X_{S_1})\mid X_{S_\cap}]-\tauR)\right]\right\|_2\\[2mm]
    &\qquad\le \| w^*(X_{S_\cap}) (m_1(X_{S_\cap})-\tauR)\|_{L^2(\P_n)}\|_2\|\wh{g}(X_{S_\cap})-g(X_{S_\cap})\|_{L^2(\P_n)}\|_2\tag{by Cauchy-Schwarz}\\[2mm]
    &\qquad=\O(1)\o(n^{-1/4})\tag{using \cref{lemma:rate-of-double-conditional-mean} and finite moments}\\[2mm]
    &\qquad=\o(n^{-1/4}),
\end{align*}
\begin{align*}
    &\left\|\P_n \left[(\wh{w}(X_{S_\cap})-w^*(X_{S_\cap}))g(X_{S_\cap})(\E_P[\tau(X_{S_1})\mid X_{S_\cap}]-\tauR)\right]\right\|_2\\[2mm]
    &\qquad\le \| \wh{w}(X_{S_\cap})-w^*(X_{S_\cap})\|_{L^2(\P_n)}\|_2\|{g}(X_{S_\cap})(m_1(X_{S_\cap})-\tauR)\|_{L^2(\P_n)}\|_2\tag{by Cauchy-Schwarz}\\[2mm]
    &\qquad=\o(n^{-1/4})\O(1)\tag{using \cref{lemma:rate-of-weights} and finite moments}\\[2mm]
    &\qquad=\o(n^{-1/4}).
\end{align*}
The last three displays and the identity \eqref{eq:abc} together imply that $\|C_n-C\|_2=\o(n^{-1/4})$. It remains to show that $D_n$ is invertible with probability tending to one, and $\|D_n^{-1}-D^{-1}\|_2=\o(n^{-1/4})$. Write
\[D_n - D = (\P_n-P)[w^*gg^\top]+\P_n[(\wh{w}-w^*)\wh{g}\wh{g}^\top] +\P_n[w^*(\wh{g}\wh{g}^\top -gg^\top)].\]
The first term above is $\O(n^{-1/2})$ by Chebyshev's inequality, because the variance is finite as $\E_P[(w^*)^2\|g\|^4] < \infty$ which in turn follows from the finiteness of the moment generating function of $g(X_{S_\cap})$ in a radius $r>2\|\lambda^*\|$ around the origin. The second term is bounded by the product of $\|\wh{w}-w^*\|_{L^2(\P_n)}$ which is $\o(n^{-1/4})$ by \cref{lemma:rate-of-weights}, and $\|\wh{g}\|_{L^2(\P_n)}$ which is $\O(1)$ due to the finiteness of the moment generating function of $\wh{g}(X_{S_\cap})$ conditional on $\wh{g}(\cdot)$. Finally, for the third term, we use the inequality $\|\wh{g}\wh{g}^\top - gg^\top\|\le (\|\wh{g}\|+\|g\|)\|\wh{g}-g\|$ and Hölder's inequality with $1/(2q')+1/q=1/2$ to obtain
\begin{align*}
   \left| \P_n[w^*(\wh{g}\wh{g}^\top -gg^\top)]\right|
  & \le \left\|w^*(\|\wh{g}\|+\|g\|)\right\|_{L^2(\P_n)}\|\wh{g}-g\|_{L^2(\P_n)}\\
  &\le\left\|w^*\right\|_{L^{2 q'}(\P_n)}\left\| \| \wh{g}\|+\| g\|\right\|_{L^q(\P_n)}\left\|\wh{g}-g \right\|_{L^2(\P_n)}\\
  &=\O(1)\O(1)\o(n^{-1/4}).\tag{using \cref{lemma:rate-of-double-conditional-mean}}
\end{align*}
Taken together, the above bounds tell us that $D_n-D=\o(n^{-1/4})$. 

Next, write $\lambda_{\min}(M)$ for the smallest eigenvalue of a matrix $M$. Observe that $\lambda_{\min}(D)>0$ because for any $u\neq 0$ we have $P(u^\top g\neq 0)>0$ and $w^*$, which imply that $u^\top D u = \E_P[w^* (u^\top g)^2]>0$. Since $\lambda_{\min}(D_n)\ge \lambda_{\min}(D)-\|D_n-D\|_2$ and $\|D_n-D\|_2=\o(1)$, the event $\{\lambda_{\min}(D_n)>\frac{1}{2}\lambda_{\min}(D)\}$ has probability tending to one. On this event, $D_n$ is invertible with $\|D_n^{-1}\|_2\le 1/\lambda_{\min}(D_n)\le 2/\lambda_{\min}(D)$. Therefore, using the identity $D_n^{-1}-D^{-1}=D_n^{-1}(D-D_n)D^{-1}$, we deduce that
\[\left\|D_n^{-1}-D^{-1}\right\|_2\le \frac{2}{\lambda_{\min}(D)^2} \|D_n-D\|_2=\o(n^{-1/4}).\]
The above, together with $\|C_n-C\|_2=\o(n^{-1/4})$ and the fact that $\lambda_{\min}(D_n)>\frac{1}{2}\lambda_{\min}(D)>0$ holds with probability tending to one, gives us the desired conclusion, since $\|\wh\nu_{2:K}-\nu^*_{2:K}\|_2\le \|D_n^{-1} (C_n-C)\|_2 + \|(D_n^{-1}-D^{-1})C\|_2$.
\end{proof}

\begin{lemma}[AIPW approximation]\label{lemma:AIPW-approx}
    Under the conditions of \cref{thm:matching.adj.sets}, it holds that
    $$\P_n \left[{w}^*(X_{S_\cap})\left(\haipw{k}- \aipw{k}\right)\right]=\o(n^{-1/2}),$$
    and  $$\P_n \left[\left(\wh{w}(X_{S_\cap})-{w}^*(X_{S_\cap})\right)\left(\haipw{k}- \aipw{k}\right)\right]=\o(n^{-1/2}).$$
\end{lemma}

\begin{proof}
    Define $\haipw{k}$ as in \cref{algo:rewt-adj-sets} and $\aipw{k}$ as in \cref{thm:matching.adj.sets}. We can write
    \begin{equation}
        \label{aipw-6parts}
        \haipw{k}-\aipw{k}=\sum_{l=1}^3 (\Delta_{l,1}(D) + \Delta_{l,0}(D)),
    \end{equation}
    where $D=(Y,\trt,X)$, and \begin{align*}
       \Delta_{1,1}(D)&:= \left(\frac{1}{\wh{e}(X_{S_k})}-\frac{1}{e(X_{S_k})}\right)\trt(Y-\mu_1(X_{S_k}))\\
       \Delta_{2,1}(D) &:=\left(1-\frac{\trt}{e(X_{S_k})}\right)(\wh{\mu}_1(X_{S_k})-\mu_1(X_{S_k}))\\
       \Delta_{3,1}(D) &:= -A\left(\frac{1}{\wh{e}(X_{S_k})}-\frac{1}{e(X_{S_k})}\right)(\wh{\mu}_1(X_{S_k})-\mu_1(X_{S_k})),
    \end{align*}
    and $\Delta_{l,0}(D)$ are defined in an analogous fashion. It follows from the tower property (and that $\wh{e}$ is externally fitted / using cross-fitting) that
    $$\E_P \left[w^*(X_{S_\cap}) \left(\frac{1}{\wh{e}(X_{S_k})}-\frac{1}{e(X_{S_k})}\right)\trt(Y-\mu_1(X_{S_k}))\;\Big|\; \wh{e}(\cdot),X_{S_k}\right]=0,$$
    since $\wh{e}(\cdot)$ is independent of the data and $\E_P\left[A\left(Y - \mu_1(X_{S_k})\right)\mid X_{S_k}\right]=0$.
    Consequently, \begin{align*}
       & \var\left(\P_n\left[w^*(X_{S_\cap}) \left(\frac{1}{\wh{e}(X_{S_k})}-\frac{1}{e(X_{S_k})}\right)\trt(Y-\mu_1(X_{S_k}))\right]\;\Big|\; \wh{e}(\cdot),\,X_{1,\,S_k}, \dots,\, X_{n,\,S_k}\right)\\
        &\le \frac{1}{n^2}\sum_{i=1}^n (w^*(X_{i,\,S_\cap}))^2\left(\frac{1}{\wh{e}(X_{i,S_k})}-\frac{1}{e(X_{i,S_k})}\right)^2 \var(Y_i\mid A_i=1, X_{i,\,S_k})\\[2mm]
        &\lesssim n^{-1} \,\P_n \left[(w^*(X_{S_\cap}))^2 (\wh{e}\,(X_{S_k})-e(X_{S_k}))^2\right] \tag{using $\var(Y\mid A=1, X_{S_k})\le C$ and strong overlap}\\[2mm] 
        &\le n^{-1} \,\|(w^*(X_{S_\cap}))^2\|_{L^{q/(q-2)}(\P_n)} \|(\wh{e}\,(X_{S_k})-e(X_{S_k}))^2\|_{L^{q/2}(\P_n)} \tag{Using Hölder's inequality with $1/{q'}+2/q=1$}\\
        &= n^{-1}\, \|w^*(X_{S_\cap})\|_{L^{2q'}(\P_n)}^2\|\wh{e}\,(X_{S_k})-e(X_{S_k})\|_{L^{q}(\P_n)}^2=\o(n^{-1}),
    \end{align*}
    The last rate is due to the facts that $\|\wh{e}\,(X_{S_k})-e(X_{S_k})\|_{L^{q}(\P_n)}=\o(1)$, and $\|w^*(X_{S_\cap})\|_{L^{2q'}(\P_n)}=\O(1)$ which follows from the assumption that $\E_P[\exp(\rc\|g(X_{S_\cap})\|)]<\infty$ for $\rc\ge 2q'\|\lambda^*\|$. 
    The above display combined with Chebyshev's inequality tell us that $$\P_n \left[w^*(X_{S_\cap})\Delta_{1,1}(D)\right]=\o(n^{-1/2}).$$ We can follow a similar path to prove that 
    \begin{align*}
        \P_n \left[w^*(X_{S_\cap})\Delta_{1,0}(D)\right]=\o(n^{-1/2}),\\ \P_n \left[w^*(X_{S_\cap})\Delta_{2,0}(D)\right]=\o(n^{-1/2}),\\
        \P_n \left[w^*(X_{S_\cap})\Delta_{2,1}(D)\right]=\o(n^{-1/2}).
    \end{align*}
    Showing $\P_n \left[w^*(X_{S_\cap})\Delta_{3,1}(D)\right]=\o(n^{-1/2})$ is slightly different. Note that by Hölder's inequality, with $q'>1$ such that $1/q' + 2/q=1$,
    \begin{align*}
     \left|   \P_n \left[w^*(X_{S_\cap})\Delta_{3,1}(D)\right]\right|
        &\le \|w^*(X_{S_\cap})\|_{L^{q'}(\P_n)}\left\|\frac{1}{\wh{e}(X_{S_k})}-\frac{1}{e(X_{S_k})}\right\|_{L^{q}(\P_n)}\left\|\wh{\mu}_1(X_{S_k})-\mu_1(X_{S_k})\right\|_{L^{q}(\P_n)}\\
    &=\O(1)\o(n^{-1/2}),
    \end{align*}
    where the second stochastic order follows from the assumption on the product of rates, and the first one follows from the assumption that $\E_P[\exp(\rc\|g(X_{S_\cap})\|)]<\infty$ for $\rc\ge q'\|\lambda^*\|$. 
    The proof of $\P_n \left[w^*(X_{S_\cap})\Delta_{3,0}(D)\right]=\o(n^{-1/2})$ is analogous to the above. This concludes the proof of the first claim, i.e., that $$\P_n \left[{w}^*(X_{S_\cap})\left(\haipw{k}- \aipw{k}\right)\right]=\o(n^{-1/2}).$$
To show the second claim, we use \cref{lemma:rate-of-weights} to say that $$\|\wh{w}(X_{S_\cap})-{w}^*(X_{S_\cap})\|_{L^2(\P_n)}=\o(n^{-1/4}),\quad\text{and}\quad\|\wh{w}(X_{S_\cap})-{w}^*(X_{S_\cap})\|_{L^{2q'}(\P_n)}=\O(1),$$ where $q'=q/(q-2)$. We then use the decomposition \eqref{aipw-6parts} to write
\begin{multline*}
    \P_n \left[\left(\wh{w}(X_{S_\cap})-{w}^*(X_{S_\cap})\right)\left(\haipw{k}- \aipw{k}\right)\right] \\= \sum_{l=1}^3 \P_n\left[\left(\wh{w}(X_{S_\cap})-{w}^*(X_{S_\cap})\right)(\Delta_{l,1}+\Delta_{l,0})\right].
\end{multline*}
We handle the first two terms in the same manner as earlier, since they constitute of summands that have conditional mean zero. The key observation is that
\begin{align*}
&\var\left(\P_n\left[\left(\wh{w}(X_{S_\cap})-w^*(X_{S_\cap})\right)\Delta_{1,1}\right]\;\Big|\; \wh{e},\, \wh{\lambda},\, \wh{g},\, X_{1,\,S_k}, \dots,\, X_{n,\,S_k}\right)\\
       &= \var\left(\P_n\left[\left(\wh{w}-w^*\right) \left(\frac{1}{\wh{e}(X_{S_k})}-\frac{1}{e(X_{S_k})}\right)\trt(Y-\mu_1(X_{S_k}))\right]\;\Big|\; \wh{e},\, \wh{\lambda},\, \wh{g},\, X_{1,\,S_k}, \dots,\, X_{n,\,S_k}\right)\\
        &\le \frac{1}{n^2}\sum_{i=1}^n \left(\wh{w}(X_{S_\cap})-w^*(X_{S_\cap})\right)^2\left(\frac{1}{\wh{e}(X_{i,S_k})}-\frac{1}{e(X_{i,S_k})}\right)^2 \var(Y_i\mid A_i=1, X_{i,\,S_k})\\
        &\lesssim n^{-1} \P_n\left[\left(\wh{w}(X_{S_\cap})-w^*(X_{S_\cap})\right)^2\left(\wh{e}\,(X_{S_k})-e(X_{S_k})\right)^2\right]\tag{using $\var(Y\mid A=1, X_{S_k})\le C$ and strong overlap}\\[2mm]
        &\le n^{-1} \|\wh{w}(X_{S_\cap})-w^*(X_{S_\cap})\|_{L^{2q'}(\P_n)}^2\|\wh{e}\,(X_{S_k})-e(X_{S_k})\|_{L^{q}(\P_n)}^2\tag{Hölder's inequality}\\[2mm]
        &=n^{-1}\O(1)\o(1) =\o(n^{-1}),
    \end{align*}
and likewise for $\Delta_{1,0}$, $\Delta_{2,1}$ and $\Delta_{2,0}$. For the remaining part we directly apply Hölder's inequality:
\begin{align*}
       & \P_n \left[\left(\wh{w}(X_{S_\cap})-w^*(X_{S_\cap})\right)\Delta_{3,1}(D)\right]\\
        &\lesssim \|\wh{w}(X_{S_\cap})-w^*(X_{S_\cap})\|_{L^{q'}(\P_n)}\left\|{\wh{e}(X_{S_k})}-{e(X_{S_k})}\right\|_{L^{q}(\P_n)}\left\|\wh{\mu}_1(X_{S_k})-\mu_1(X_{S_k})\right\|_{L^{q}(\P_n)}\\
    &=\O(1)\o(n^{-1/2}) =\o(n^{-1/2}),
    \end{align*}
and likewise for $\Delta_{3,0}$. This completes the proof.
\end{proof}
\begin{lemma}[Taylor expansion of transfer weights]\label{lemma:Taylor-expand-weights}
    Under the conditions of \cref{thm:matching.adj.sets},
    \begin{align*}
        &\P_n \left[\left(\wh{w}(X_{S_\cap})-w^*(X_{S_\cap})\right)\,\phi(X_{S_\cap};\nu^*) \right]\\
        &=\P_n\left[ w^*(X_{S_\cap})\left(\wh{g}(X_{S_\cap})^\top\,\wh{\lambda} -g(X_{S_\cap})^\top\lambda^*\right)\,\phi(X_{S_\cap};\nu^*)\right]+\o(n^{-1/2}),
    \end{align*}
    where $\phi(X_{S_\cap};\nu^*)=\sum_{k=1}^K\nu_k^*\,\E_P[\tau(X_{S_k})\mid X_{S_\cap}]-\tauR$.
\end{lemma}
\begin{proof}
    First we observe that $\P_n \left[w^*(X_{S_\cap})\,\phi(X_{S_\cap};\nu^*) \right]=\O(n^{-1/2})$, and therefore 
    \begin{multline*}
        \left|\P_n \left[\left(\wh{w}(X_{S_\cap})-w^*(X_{S_\cap})\right)\,\phi(X_{S_\cap};\nu^*) \right]\right|\\
        \le \sum_{k=1}^K |\nu^*_k| \|\wh{w}(X_{S_\cap})-w^*(X_{S_\cap})\|_{L^2(\P_n)}\|\E_P[\tau(X_{S_k})\mid X_{S_\cap}]-\tauR\|_{L^2(\P_n)}=\o(n^{-1/4}).
    \end{multline*}
Consequently, $\P_n \left[\wh{w}(X_{S_\cap})\,\phi(X_{S_\cap};\nu^*) \right]
=\o(n^{-1/4})$, and thus $$\left(\frac{\P_n \left[e^{\wh{g}(X_{S_\cap})^\top \wh\lambda}\right]}{\E_P\left[e^{g(X_{S_\cap})^\top \lambda^*}\right]}-1\right)\P_n \left[\wh{w}(X_{S_\cap})\,\phi(X_{S_\cap};\nu^*) \right]\stackrel{\eqref{patch:D3-first}}{=}\O(n^{-1/4})\o(n^{-1/4})=\o(n^{-1/2}).$$
This implies that
\begin{align*}
        &\P_n \left[\left(\wh{w}(X_{S_\cap})-w^*(X_{S_\cap})\right)\,\phi(X_{S_\cap};\nu^*) \right]-\P_n\left[ w^*(X_{S_\cap})\left(\wh{g}(X_{S_\cap})^\top\,\wh{\lambda} -g(X_{S_\cap})^\top\lambda^*\right)\,\phi(X_{S_\cap};\nu^*)\right]\\
        &=\P_n \left[\left(e^{\Delta_n(X_{S_\cap})}-1-\Delta_n(X_{S_\cap})\right)w^*(X_{S_\cap})\,\phi(X_{S_\cap};\nu^*) \right]+\o(n^{-1/2}),
    \end{align*}
    where $\Delta_n(x):=\wh{g}(x)^\top \wh\lambda-{g}(x)^\top \lambda^*$. It follows from the proof of \cref{lemma:rate-of-weights} that \(\|\Delta_n\|_{L^{q}(\P_n)}^2=\o(n^{-1/2})\). Consequently, using the Taylor expansion \eqref{taylor} from the proof of \cref{lemma:rate-of-weights}, we deduce the following.
    \begin{align*}
        &\P_n \left[\left(e^{\Delta_n(X_{S_\cap})}-1-\Delta_n(X_{S_\cap})\right)w^*(X_{S_\cap})\,\phi(X_{S_\cap};\nu^*) \right]\\[2mm]
        &=\frac{1}{2}\,\P_n [e^{\xi(X_{S_\cap})}\Delta_n^2(X_{S_\cap})w^*(X_{S_\cap})\,\phi(X_{S_\cap};\nu^*)]\\[2mm]
        &\lesssim \|e^{\xi(X_{S_\cap})}w^*(X_{S_\cap})\,\phi(X_{S_\cap};\nu^*)\|_{L^{q'}(\P_n)}\|\Delta_n\|_{L^{q}(\P_n)}^2 \tag{by Hölder's inequality}\\[2mm]
        &=\O(1)\o(n^{-1/2})=\o(n^{-1/2}),
    \end{align*} where $1/q'+2/q=1$. To see why the first factor in the above bound is $\O(1)$, note that $w^*e^{\Delta_n}=e^{\wh g^\top\wh\lambda}/\psi_P(\lambda^*)=\wh w\,\P_n[e^{\wh g^\top\wh\lambda}]/\psi_P(\lambda^*)\le\wh w/\psi_P(\lambda^*)$ on the event of \cref{lemma:rate-of-lambda}. Therefore, by another application of Hölder's inequality,
    \begin{align*}
&\big\|e^{\xi(X_{S_\cap})}w^*(X_{S_\cap})\phi(X_{S_\cap};\nu^*)\big\|_{L^{q'}(\P_n)}
\\
&\le \|w^*(X_{S_\cap})\phi(X_{S_\cap};\nu^*)\big\|_{L^{q'}(\P_n)}+\|e^{\Delta_n(X_{S_\cap})}w^*(X_{S_\cap})\phi(X_{S_\cap};\nu^*)\big\|_{L^{q'}(\P_n)} \tag{using $e^\xi\le 1+e^{\Delta_n}$}\\
&\le\Big(\|w^*\|_{L^{2q'}(\P_n)}+\psi_P(\lambda^*)^{-1}\|\wh w\|_{L^{2q'}(\P_n)}\Big)\|\phi(X_{S_\cap};\nu^*)\|_{L^{2q'}(\P_n)}=\O(1),
    \end{align*}
    using \cref{lemma:rate-of-weights}.
    This completes the proof. 
    \end{proof}

\begin{lemma}[An approximation of the bias-correction term]\label{lemma:bias-corr-is-close-to-oracle}
    Under the conditions of \cref{thm:matching.adj.sets}, the bias-correction term $B_n$ in \cref{algo:rewt-adj-sets} is asymptotically equivalent to the following
$$\wh{B}_n^*=\wh{\lambda}^\top \,\P_n\left[w^*(X_{S_\cap})(\Delta\aipw{}-\wh{g}(X_{S_\cap}))\,\phi(X_{S_\cap};\nu^*)\right],$$ in the sense that $B_n=\wh{B}_n^*+\o(n^{-1/2})$, where  $\phi(X_{S_\cap};\nu^*)=\sum_{k=1}^K\nu_k^*\,\E_P[\tau(X_{S_k})\mid X_{S_\cap}]-\tauR$.
\end{lemma}
\begin{proof} Recall that we define the bias-correction term $B_n$ in \cref{algo:rewt-adj-sets} as
$$B_n:=\wh{\lambda}^\top\,\P_n\left[\wh{w}(X_{S_\cap})\left(
\Delta\haipw{}-\wh{g}(X_{S_\cap})\right)\left(\sum_{k=1}^K \wh\nu_k\, \wh{\E}_P\left[\wh\tau(X_{S_k})\mid X_{S_\cap}\right]-\wh\tau_1^\mathtt{R} \right)\right],$$
where $\Delta\haipw{k}:=\haipw{1}-\haipw{k}$, and other quantities are as defined in \cref{algo:rewt-adj-sets}.  Also define 
$${B}_n^*=\wh{\lambda}^\top \,\P_n\left[w^*(X_{S_\cap})(\Delta\aipw{}-{g}(X_{S_\cap}))\,\phi(X_{S_\cap};\nu^*)\right],$$ where $\Daipw{k}=\aipw{1}-\aipw{k}$, and recall that $g(X_{S_\cap})=\E_P[\Delta\tau(X)\mid X_{S_\cap}]=\E_P[\Delta\aipw{}\mid X_{S_\cap}]$. \cref{cvarlemma} yields $\var(\Delta\aipw{k}\mid X_{S_\cap})\le C$ (for some constant) and thus \cref{lemma.one} tells us that 
\begin{align*}
T_1 &:=\P_n \left[(\wh{w}(X_{S_\cap})-w^*(X_{S_\cap}))(\Delta\aipw{}-g(X_{S_\cap}))\phi(X_{S_\cap};\nu^*)\right]=\o(n^{-1/2}),
\end{align*}
 provided we show that
 \begin{equation}\label{what-we-need-for-T1}
     \P_n\left[(\wh{w}(X_{S_\cap})-w^*(X_{S_\cap}))^2\cdot \phi^2(X_{S_\cap};\nu^*)\,\right]=\o(1).
 \end{equation}
Applying Hölder's inequality with exponents $s,s'\in(1,\infty)$, $1/s+1/s'=1$, 
\begin{align*}
    \left|\P_n\left[(\wh{w}(X_{S_\cap})-w^*(X_{S_\cap}))^2\cdot \phi^2(X_{S_\cap};\nu^*)\,\right]\right|\le\|\wh{w}-w^*\|_{L^{2s}(\P_n)}^2 \|\phi\|_{L^{2s'}(\P_n)}^2.
\end{align*}
 The second factor is $\O(1)$ for any $s'>1$, because $m_1$ and $g$ have finite moment generating function in an open neighborhood of zero. For the first factor, pick $t\in (2, \rc/\|\lambda^*\|)$ (which is possible to do, because $\rc\ge 2q'\|\lambda^*\|>2\|\lambda^*\|$) and pick $2<2s<t$.  Write $1/(2s) = \theta/2 + (1-\theta)/t$. Then, by the interpolation inequality for $L^p$ norms (which again follows from Hölder's inequality),
 \[\|\wh{w}-w^*\|_{L^{2s}(\P_n)}\le \|\wh{w}-w^*\|_{L^{2}(\P_n)}^\theta \|\wh{w}-w^*\|_{L^{t}(\P_n)}^{1-\theta}.\]
The first factor above is $\o(1)$ by \cref{lemma:rate-of-weights}, and the second factor is $\O(1)$ due to finite moment generating functions of $g$ and $\wh{g}$. This finishes the proof of \eqref{what-we-need-for-T1}, and thus proves via \cref{lemma.one} that $T_1=\o(n^{-1/2})$.

 \medskip

 On the other hand, it follows by an argument similar to our proof of \cref{lemma:AIPW-approx} (using Hölder's inequality) that
$$T_2:=\P_n \left[w^*(X_{S_\cap})(\Delta\haipw{}-\Delta\aipw{})\phi(X_{S_\cap};\nu^*)\right]=\o(n^{-1/2}).$$

\medskip

 Next, define
$$T_3:=\P_n \left[w^*(X_{S_\cap})(\Delta\aipw{}-g(X_{S_\cap}))\Delta_f(X_{S_\cap})\right],$$
where
\begin{align*}
    \Delta_f (X_{S_\cap}) &:= \sum_{k=1}^K \wh\nu_k\underbrace{ \wh{\E}_P\left[\wh\tau(X_{S_k})\mid X_{S_\cap}\right]}_{\wh{m}_k}- \sum_{k=1}^K \nu^*_k\underbrace{\E_P\left[\tau(X_{S_k})\mid X_{S_\cap}\right]}_{m_k}- (\wh\tau_1^\mathtt{R} -\tauR)\\
    &= \sum_{k=1}^K \wh\nu_k (\wh{m}_k -m_k) + \sum_{k=1}^K (\wh\nu_k -\nu_k^*)m_k - (\wh\tau_1^\mathtt{R} -\tauR).
\end{align*}
This gives
\begin{multline*}
    T_3 = \sum_{k=1}^K \wh{\nu}_k \P_n [w^* (\Delta\aipw{}-g)(\wh{m}_k-m_k)] \\+ \sum_{k=1}^K (\wh{\nu}_k-\nu_k^*)\, \P_n [w^* (\Delta\aipw{}-g)m_k]\\ - (\wh\tau_1^\mathtt{R} -\tauR)\,\P_n[w^*(\Delta\aipw{}-g)].
\end{multline*}
Applying \cref{lemma:rate-of-nu} and \eqref{rate-emp-minus-oracle-tau-k}, we can conclude that the last two sums are $\o(n^{-1/2})$. For the first sum, we can apply \cref{lemma.one}, because by Hölder's inequality,
\begin{align*}
    \P_n\left[(w^*)^2(\wh{m}_k-m_k)^2\right]\le \|w^*\|_{L^{2q'}(\P_n)}^2\,\|\wh{m}_k-m_k\|_{L^{q}(\P_n)}^2=\O(1)\o(1)=\o(1),
\end{align*}
and \cref{cvarlemma} gives $\var(\Delta\aipw{k}\mid X_{S_\cap})\le C$ for some finite constant $C$.
Therefore \cref{lemma.one} applies and yields that $T_3=\o(n^{-1/2})$.

\medskip

Finally, define 
$$T_2':=\P_n \left[w^*(X_{S_\cap})(\wh{g}(X_{S_\cap})-g(X_{S_\cap}))\left(\sum_{k=1}^K \nu^*_k {\E}\left[\tau(X_{S_k})\mid X_{S_\cap}\right]-\tauR\right)\right].$$
We can use the algebraic identity $\hat{a}\hat{b}\hat{c}-abc = ab(\hat{c}-c)+bc(\hat{a}-a)+ca(\hat{b}-b)+ (\hat{a}-a)(\hat{b}-b)c+(\hat{b}-b)(\hat{c}-c)a+(\hat{c}-c)(\hat{a}-a)b+(\hat{a}-a)(\hat{b}-b)(\hat{c}-c)$ to write
$$B_n-B_n^* = \wh{\lambda}^\top (T_1+T_2-T_2'+T_3)+\o(n^{-1/2}).$$
On the other hand,
$$B_n^* - \wh{B}_n^*=\wh{\lambda}^\top \,\P_n\left[w^*(X_{S_\cap})(\wh{g}(X_{S_\cap})-g(X_{S_\cap}))\,\phi(X_{S_\cap};\nu^*)\right]=\wh{\lambda}^\top T_2'.$$
Therefore, using \cref{lemma:rate-of-lambda} and the fact that $T_1$, $T_2$ and $T_3$ are all $\o(n^{-1/2})$ to conclude that
\begin{align*}
    B_n-\wh{B}_n^* &=B_n-B_n^*+B_n^*-\wh{B}_n^*\\ 
    &=(\lambda^*)^\top (T_1+T_2+T_3)+\o(T_1+T_2+T_3)\\
    &=\o(n^{-1/2}),
\end{align*}
as desired to show.
\end{proof}

\begin{lemma}[Uniform version of \cref{lemma:rate-of-double-conditional-mean}] \label{lem:gwh}
Under the conditions of \cref{thm:matching.adj.sets-uniform}, it holds that
\[
  \sup_{P\in\sP}
  \|\wh g(X_{S_\cap}) - g(X_{S_\cap};\,P)\|_{L^q(\P_n)}
  \;=\; \o(n^{-1/4}).
\]
\end{lemma}
\begin{proof}
Write $\wh m_k(X_{S_\cap}):=\wh{\E}_P[\wh\tau(X_{S_k})\mid X_{S_\cap}]$, and \(m_k(X_{S_\cap}):={\E}_P[\tau(X_{S_k})\mid X_{S_\cap}]\) for each $k$. Then, $\wh g_k=\wh m_1-\wh m_k$ by construction and $g_k=m_1-m_k$ by definition. Using the triangle inequality, $$\|\wh g_k-g_k\|_{L^q(\P_n)}\le\|\wh m_1-m_1\|_{L^q(\P_n)}+\|\wh m_k-m_k\|_{L^q(\P_n)},$$ and both the above terms are $\o(n^{-1/4})$ uniformly in $P\in\sP$, by assumption.
\end{proof}

\begin{lemma}[Uniform version of \cref{lemma:rate-of-lambda}]
\label{lem:lambdawh}
Under the conditions of \cref{thm:matching.adj.sets-uniform}, it holds that
\[
  \sup_{P\in\sP}
  \|\wh\lambda - \lambda^*(P)\|
  \;=\; \o(n^{-1/4}).
\]
\end{lemma}

\begin{proof}
We follow the proof of \cref{lemma:rate-of-lambda} closely and make every bound uniform in $P\in\mathcal P$. Set $\Lambda=\sup_{P\in\sP}\|\lambda^*(P)\|$, which is finite by \cref{propo:exist-n-unique-uniform}. Throughout, we fix $\eps_0>0$ with $2(\Lambda+2\eps_0)<r$, where $r$ is the radius as in \cref{thm:matching.adj.sets-uniform}. For $0<\eps\le\eps_0$, denote $B_\eps(P):=\{\lambda:\|\lambda-\lambda^*(P)\|\le\eps\}$. The triangle inequality yields $\|\lambda\|\le \Lambda+\eps_0$ for $\lambda\in B_{\eps_0}(P)$.

Using Markov's inequality and condition (4) of \cref{thm:matching.adj.sets-uniform},
\begin{equation}\label{patch:D10-moments}
\lim_{M\to\infty}\limsup_{n\to\infty}\sup_{P\in\mathcal P}P\Big(\P_n\big[e^{r\|g(X_{S_\cap};P)\|}+e^{r\|\wh g(X_{S_\cap})\|}\big]>M\Big)=0,
\end{equation}
which is the uniform version of \eqref{patch:D2-moments}. 

Next, consider the matrix $$A(P):=\E_P[e^{-(\Lambda+\eps_0)\|g\|}gg^\top].$$ As in the proof of \cref{lemma:rate-of-lambda}, we will show that $A(P)\succeq \kappa I$ for some $\kappa>0$; the additional challenge here is to find a choice of $\kappa>0$ that works uniformly over $P\in\sP$. Towards this, we write
$G:=\sup_{P\in\mathcal P}\E_P[e^{r\|g\|}]<\infty$. With $c_0$ and $\zeta$ as in \cref{assump:heterogeneity-uniform}, choose $M_0>0$ large enough such that $e^{-rM_0}G\le c_0/2$. We will show that the choice $\kappa:=\frac12c_0\zeta^2e^{-(\Lambda+\eps_0)M_0}$ fits the bill.
 For a unit vector $u$, we apply Markov's inequality and use \cref{assump:heterogeneity-uniform} to deduce that
 \begin{align*}
     u^\top A(P)u &\ge\zeta^2e^{-(\Lambda+\eps_0)M_0}\,P\big(u^\top g\ge\zeta,\ \|g\|\le M_0\big)\tag{Markov's inequality}\\
     &\ge\zeta^2e^{-(\Lambda+\eps_0)M_0}\big(P(u^\top g\ge \zeta) - P(\|g\|>M_0)\big)\\
     &\ge\zeta^2e^{-(\Lambda+\eps_0)M_0}\big(c_0 - P(\|g\|>M_0)\big)\tag{by \cref{assump:heterogeneity-uniform}}\\
     &\ge\zeta^2e^{-(\Lambda+\eps_0)M_0}\big(c_0-e^{-rM_0}\E_P[e^{r\|g\|}]\big)\tag{Markov's inequality}\\
     &\ge\zeta^2e^{-(\Lambda+\eps_0)M_0}\big(c_0-e^{-rM_0}G\big)\ge\kappa. \tag{by definition}
 \end{align*}
This shows that$A(P)\succeq\kappa I$ for every $P\in\mathcal P$. Then, as in the proof of \cref{cond2} of \cref{lemma:rate-of-lambda}, $\nabla L(\lambda;P)=0$ at $\lambda=\lambda^*(P)$ and $\nabla^2L(\lambda;P)\succeq A(P)\succeq\kappa I$ for $\lambda\in B_{\eps_0}(P)$.

\medskip

We next state the uniform versions of \cref{cond1,cond2,cond3} from the proof of \cref{lemma:rate-of-lambda}: For any $0<\eps\le\eps_0$,
\begin{enumerate}
\item\label{cond1u} $\sup_{P\in\mathcal P}\sup_{\lambda\in B_\eps(P)}|\wh L_n(\lambda)-L(\lambda;P)|=\o(1)$;
\item\label{cond2u} $L(\lambda;P)\ge L(\lambda^*(P);P)+\kappa\eps^2/2$ for every $P\in\mathcal P$ and every $\lambda$ with $\|\lambda-\lambda^*(P)\|=\eps$;
\item\label{cond3u} $\sup_{P\in\mathcal P}\sup_{\lambda\in B_{\eps_0}(P)}\|\nabla^2\wh L_n(\lambda)-\nabla^2L(\lambda;P)\|=\o(1)$; consequently, $$\inf_{P\in\mathcal P}P\big(\nabla^2\wh L_n(\lambda)\succeq(\kappa/2)I\text{ for all }\lambda\in B_{\eps_0}(P)\big)\to1.$$
\end{enumerate}
Given \cref{cond1u,cond2u,cond3u}, the consistency argument in the proof of \cref{lemma:rate-of-lambda}  applies verbatim with every probability replaced by its infimum over $P\in\mathcal P$, and yields events $E_n(P)$ with $\inf_PP(E_n(P))\to1$ such that, on the event $E_n(P)$ the minimizer $\wh\lambda$ exists, is unique, and satisfies $\|\wh\lambda-\lambda^*(P)\|<\eps$. 

\medskip

The proof of \cref{cond2u} closely follows the exact steps in the proof of \cref{cond2} in the proof of \cref{cond2} of \cref{lemma:rate-of-lambda}, thus omitted. We next show \cref{cond1u,cond3u}.
 To do this in one go, write $$\wh L_n=\P_n[F_0(\cdot,\wh g)],\qquad \nabla^2\wh L_n=\P_n[F_2(\cdot,\wh g)],$$ i.e., $F_0(\lambda, x)=\exp(x^\top \lambda)$ and $F_2(\lambda, x)=\exp(x^\top \lambda)\,xx^\top$.
 We decompose $\P_n[F_j(\lambda,\wh g)]-\E_P[F_j(\lambda,g)]$ into a plug-in part and a sampling part:
 \[\P_n[F_j(\lambda,\wh g)]-\E_P[F_j(\lambda,g)] = \P_n[F_j(\lambda,\wh g)-F_j(\lambda,g)]+(\P_n-\E_P)[F_j(\lambda,g)]\]
 The bounds for the plug-in part in the proofs of \cref{cond1,cond3} of \cref{lemma:rate-of-lambda} extend verbatim to this case, i.e., hold uniformly in $P$ over $\sP$. This is because those involved only $\|\wh g-g\|_{L^2(\P_n)}=\o(n^{-1/4})$ which holds uniformly by \cref{lem:gwh}, and $\P_n[e^{r\|g\|}+e^{r\|\wh g\|}]=\O(1)$ which holds uniformly by \eqref{patch:D10-moments}. Next, for the sampling part, note that for $\|\lambda\|,\|\lambda'\|\le \Lambda+\eps_0$, we have the Lipschitz bound
\[
\|F_j(\lambda,x)-F_j(\lambda',x)\|\le\|x\|^{j+1}e^{(\Lambda+\eps_0)\|x\|}\|\lambda-\lambda'\|.
\]
Fix $\delta>0$, and let $N_\delta(P)$ be a $\delta$-net of $B_{\eps_0}(P)$ with at most $N_\delta:=(1+2\eps_0/\delta)^{K-1}$ points (note, this covering number does not depend on $P$). Then,
\begin{equation}
    \label{coveringapprox}
    \sup_{\lambda\in B_{\eps_0}(P)}\big\|(\P_n-\E_P)F_j(\lambda,g)\big\|\le\max_{\lambda\in N_\delta(P)}\big\|(\P_n-\E_P)F_j(\lambda,g)\big\|+\delta\,(\P_n+\E_P)\big[\|g\|^{j+1}e^{(\Lambda+\eps_0)\|g\|}\big].
\end{equation}
Since $r>2(\Lambda+\eps_0)$, we can write 
\begin{equation}
    \label{randomexpineq}
e^{2(\Lambda+\eps_0)\|x\|}\|x\|^{2j}+\|x\|^{j+1}e^{(\Lambda+\eps_0)\|x\|}\le C e^{r\|x\|},
\end{equation}
for each $j\in\{0,2\}$, where $C=C({r,\Lambda})\in (0,\infty)$ is an absolute constant. Using \eqref{randomexpineq}, Chebyshev's inequality and the union bound, we arrive at
\[
\sup_{P\in\mathcal P}P\Big(\max_{\lambda\in N_\delta(P)}\big\|(\P_n-\E_P)F_j(\lambda,g)\big\|>\eta\Big)\le\frac{N_\delta C G}{n\eta^2}.
\]
On the other hand, using \eqref{randomexpineq} and Markov's inequality, we get 
$$\sup_PP(\delta(\P_n+\E_P)[\|g\|^{j+1}e^{(\Lambda+\eps_0)\|g\|}]>\eta)\le 2\delta C G/\eta.$$
Combining the last two displays with \eqref{coveringapprox}, we get
\[\sup_{P} P\left(\sup_{\lambda\in B_{\eps_0}(P)}\|(\P_n-\E_P)F_j(\lambda, g)\|>2\eta\right) \le \frac{N_\delta C G}{\eta^2}\cdot \frac{1}{n} + \frac{2CG}{\eta}\cdot \delta.\]
Choosing $\delta$ small first and then $n$ large proves the desired uniform convergence in \cref{cond1u,cond3u}. Finally, as in \cref{lemma:rate-of-lambda}, the smallest eigenvalue of $\nabla^2\wh L_n(\lambda)$ is at least $\kappa-\|\nabla^2\wh L_n(\lambda)-\nabla^2L(\lambda;P)\|$, which gives the last assertion of \cref{cond3u}.

\medskip

Finally, we bound $\|\wh{w}-w^*\|_{L^2(\P_n)}$ uniformly over $\sP$.
 We work with the event $E_n(P)$ intersected with the event in \cref{cond3u}, which holds with probability tending to one. On this event, the integral form of Taylor's theorem gives 
 \begin{align*}
 &\|\wh\lambda-\lambda^*(P)\|\le(2/\kappa)\|\nabla\wh L_n(\lambda^*(P))\|,\\
    &\nabla \wh{L}_n\left(\lambda^*\right)=\P_n\left[\exp\big(g^{\top} \lambda^*\big) g\right]
+\P_n\left[\exp \big(\wh{g}^{\top} \lambda^*\big) \wh{g}-\exp \big(g^{\top} \lambda^*\big) g\right].
 \end{align*}
 exactly as in the proof of \cref{lemma:rate-of-lambda}. The term $\P_n[e^{g^\top\lambda^*}g]$ is an i.i.d.~sum of mean zero summands with the variance bounded by $CG$, due to \eqref{randomexpineq} and the definition of $G$. Therefore, $\P_n[e^{g^\top\lambda^*}g]=\O(n^{-1/2})$ uniformly in $P$ by Chebyshev's inequality. The plug-in term $\O(1)\cdot\|\wh g-g\|_{L^2(\P_n)}$ can be bounded using the mean value theorem exactly as in the proof of \cref{lemma:rate-of-lambda}, by at most $\O(1)\cdot\|\wh g-g\|_{L^2(\P_n)}$ uniformly in $P$, due to \eqref{patch:D10-moments} and $\sup_{P\in\P}\|\wh g-g\|_{L^2(\P_n)}=\o(n^{-1/4})$ (\cref{lem:gwh}). This leads to the desired conclusion that $\sup_P\|\wh\lambda-\lambda^*(P)\|=\o(n^{-1/4})$.
\end{proof}

\begin{lemma}[Uniform version of \cref{lemma:rate-of-weights}]
\label{lem:wwh}
Under the conditions of \cref{thm:matching.adj.sets-uniform}, it holds that
$$\sup_{P\in\mathcal P}\|\widehat w(X_{S_\cap})-w^*(X_{S_\cap};P)\|_{L^2(\P_n)}=o_p(n^{-1/4}),$$ and for every $t\in[1,r/\Lambda)$, $$\sup_{P\in\mathcal P}\|\widehat w(X_{S_\cap})-w^*(X_{S_\cap};P)\|_{L^t(\P_n)}=O_p(1).$$
\end{lemma}
\begin{proof}
The proof of \cref{lemma:rate-of-weights} uses \cref{lemma:rate-of-double-conditional-mean,lemma:rate-of-lambda}, the moment bound \eqref{patch:D2-moments}, Chebyshev's inequality for averages of $e^{g^\top\lambda^*}$ and $e^{g^\top\lambda^*}g$, and the positivity of $\psi_P(\lambda^*)$. Under the conditions of \cref{thm:matching.adj.sets-uniform}, each of these holds uniformly in $P\in\mathcal P$. More precisely,
\begin{itemize}
    \item \cref{lem:gwh} extends \cref{lemma:rate-of-double-conditional-mean} uniformly over $\sP$.
    \item \cref{lem:lambdawh} extends \cref{lemma:rate-of-lambda} uniformly over $\sP$.
    \item  \eqref{patch:D10-moments} replaces \eqref{patch:D2-moments}.
    \item The variances in the Chebyshev bounds are now uniformly bounded by a constant (depending only on $\rc$, and $\Lambda$) times $\sup_P\E_P[e^{r\|g\|}]<\infty$.
    \item $\psi_P(\lambda^*(P))\ge\exp(-\Lambda\sup_P\E_P\|g\|)>0$ by Jensen's inequality.
\end{itemize}
Therefore, replacing $\|\lambda^*\|$ by $\Lambda$ throughout the proof of \cref{lemma:rate-of-weights}  yields both the desired claims uniformly in $P$ over $\sP$.

\end{proof}

\begin{lemma}[Uniform version of \cref{lemma:rate-of-nu}]
\label{lem:nuwh}
Under the conditions of \cref{thm:matching.adj.sets-uniform}, it holds that
\[
  \sup_{P\in\sP}
  \|\wh\nu - \nu^*(P)\|
  \;=\; \o(n^{-1/4}).
\]
\end{lemma}
\begin{proof}
We start by writing, from the proof of \cref{lemma:rate-of-nu}, that
$$\wh\nu_{2:K} - \nu^*_{2:K} = D_n^{-1}(C_n - C) + (D_n^{-1}-D^{-1})C,$$
where $D = \E_P[w^* g g^\top]$ and
$C = \E_P[w^* g\,(m_1(X_{S_\cap})-\tauR)]$, $m_k(X_{S_\cap})=\E_P[\tau(X_{S_k})\mid X_{S_\cap}]$ (here we suppress the additional dependence on $P$ to reduce notational clutter).

\medskip

The decompositions of $C_n-C$ and $D_n-D$ in the proof of \cref{lemma:rate-of-nu} and the corresponding bounds hold uniformly in $P\in\mathcal P$, by \cref{lem:gwh,lem:wwh} and conditions (1) and (4) of \cref{thm:matching.adj.sets-uniform}. Thus,
$$\sup_{P\in\sP}\|C_n-C\|=\o(n^{-1/4}),\quad{ and
}\quad\sup_{P\in\sP}\|D_n^{-1}-D^{-1}\|=\o(n^{-1/4}).$$
The vector $C=C(P)$ is uniformly bounded by condition (4) of \cref{thm:matching.adj.sets-uniform}. Moreover, $D(P)\succeq\kappa I$ uniformly in $P$, with $\kappa$ as in the proof of \cref{lem:lambdawh}. More precisely, since $\psi_P(\lambda^*(P))\le1$, we have $w^*(\cdot\,;P)\ge e^{-\Lambda\|g\|}\ge e^{-(\Lambda+\eps_0)\|g\|}$, hence $D(P)\succeq A(P)\succeq\kappa I$. 
Therefore the arguments in the proof of \cref{lemma:rate-of-nu}, with $\lambda_{\min}(D)$ replaced by $\kappa$, shows that $D_n$ is invertible with probability tending to one uniformly in $P$, and that $\sup_P\|\widehat\nu-\nu^*(P)\|=o_p(n^{-1/4})$.
\end{proof}

\begin{lemma}[Uniform version of \cref{lemma:AIPW-approx}]
\label{lem:aipw}
Under the conditions of \cref{thm:matching.adj.sets-uniform}, it holds 
for each $k=1,\dots,K$ that
\[
  \sup_{P\in\sP}
  \left|
    \P_n\!\bigl[w^*(X_{S_\cap};P)
      \bigl(\tauAIPWobs{k} - \tauAIPW{k}\bigr)
    \bigr]
  \right|
  \;=\; \o(n^{-1/2}),
\]
and the same holds with $(\wh w - w^*)$ replacing $w^*$.
\end{lemma}
\begin{proof}
The proof of \cref{lemma:AIPW-approx} bounds each of the six terms in \eqref{aipw-6parts} either by the conditional Chebyshev's inequality or by Hölder's inequality. We check here that each of those bound holds uniformly in $P\in\mathcal P$. 

 Recall that $\rc >2q'\Lambda$ where $q'=q/(q-2)$ and $\Lambda=\sup_{P\in\sP}\|\lambda^*(P)\|<\infty$. Therefore, invoking condition (4) of \cref{thm:matching.adj.sets-uniform} and applying Markov's inequality, we get
 \begin{equation}
     \label{wstar2qprimeunif}
     \sup_P\|w^*(X_{S_\cap};P)\|_{L^{2q'}(\P_n)}=\O(1),
 \end{equation}
 and \cref{lem:wwh} yields the same for $\wh w-w^*$.

 Next, we consider the terms in \eqref{aipw-6parts} in order.
 For the term $\Delta_{1,1}$ defined in the proof of \cref{lemma:AIPW-approx}, we have
 \begin{align*}
&\sup_{P\in\sP} \var_P\left(\P_n[w^*\Delta_{1,1}]\mid \wh e,X_{1,S_k},\dots,X_{n,S_k})\right)\\
& \le n^{-1}C\xi^{-4} \sup_{P\in\sP}\|w^*\|^2_{L^{2q'}(\P_n)}\sup_{P\in\sP}\|\wh e-e\|^2_{L^q(\P_n)}
     \tag{by conditions (2) and (3) of \cref{thm:matching.adj.sets-uniform}}\\
     &=\O(1)\o(n^{-1}), \tag{using \eqref{wstar2qprimeunif} and condition (1) of \cref{thm:matching.adj.sets-uniform}}
 \end{align*}
 and hence the desired conclusion follows by Chebyshev's inequality. The proof for $\Delta_{1,0}$ is exactly similar.

 For $\Delta_{2,a}$, $\var(1-A/e(X_{S_k})\mid X_{S_k})\le\xi^{-1}$ gives the bound $n^{-1}\xi^{-1}\|w^*\|^2_{L^{2q'}(\P_n)}\|\wh\mu_a-\mu_a\|^2_{L^q(\P_n)}$. Since $\sup_P\|\wh e-e\|_{L^q(\P_n)}$ and $\sup_P\|\wh\mu_a-\mu_a\|_{L^q(\P_n)}$ are $\o(1)$ by condition (1) of \cref{thm:matching.adj.sets-uniform}, the conditional Chebyshev inequality (as in the proof of \cref{lemma.one}) shows that these terms are $\o(n^{-1/2})$ uniformly in $P$.

 For $\Delta_{3,a}$, Hölder's inequality with $1/q'+2/q=1$ and strong overlap give $$|\P_n[w^*\Delta_{3,a}]|\le\xi^{-2}\|w^*\|_{L^{q'}(\P_n)}\|\wh e-e\|_{L^q(\P_n)}\|\wh\mu_a-\mu_a\|_{L^q(\P_n)},$$ which is $\o(n^{-1/2})$ uniformly by the product-rate part of condition (1) of \cref{thm:matching.adj.sets-uniform}. The same arguments apply with $\wh w-w^*$ in place of $w^*$, using \cref{lem:wwh}.
\end{proof}

\begin{lemma}[Uniform version of \cref{lemma:Taylor-expand-weights}]\label{lem:taylor-weights}
    Under the conditions of \cref{thm:matching.adj.sets-uniform},
    \begin{align*}
        &\P_n \left[\left(\wh{w}(X_{S_\cap})-w^*(X_{S_\cap};P)\right)\,\phi(X_{S_\cap};\nu^*(P),P) \right]\\
        &=\P_n\left[ w^*(X_{S_\cap};P)\left(\wh{g}(X_{S_\cap})^\top\,\wh{\lambda} -g(X_{S_\cap};P)^\top\lambda^*(P)\right)\,\phi(X_{S_\cap};\nu^*(P),P)\right]+R_n^{w}(P),
    \end{align*}
    where $\phi(X_{S_\cap};\nu,P)=\sum_{k=1}^K\nu_k\,\E_P[\tau(X_{S_k};P)\mid X_{S_\cap}]-\tauR(P)$, and the remainder term is uniformly asymptotically negligible: $$\sup_{P\in\sP}|R_n^w(P)|=\o(n^{-1/2}).$$
\end{lemma}

\begin{proof}
    Putting together the algebraic decompositions used in the proof of \cref{lemma:Taylor-expand-weights}, the remainder terms are controlled using \cref{lemma:rate-of-weights}, the moment bounds from the proof of \cref{lemma:rate-of-weights} and Hölder's inequality. Since \cref{lemma:rate-of-weights} generalizes uniformly as in \cref{lem:wwh} and the moment bounds also hold uniformly under the conditions of \cref{thm:matching.adj.sets-uniform}, the rest of the proof of \cref{lemma:Taylor-expand-weights} carries through uniformly over $P\in\sP$.
\end{proof}

\begin{lemma}[Uniform version of \cref{lemma:bias-corr-is-close-to-oracle}]\label{lem:bias-corr-approx}
    Under the conditions of \cref{thm:matching.adj.sets-uniform}, the bias-correction term $B_n$ in \cref{algo:rewt-adj-sets} is uniformly asymptotically equivalent to the following
$$\wh{B}_n^*(P)=\wh{\lambda}^\top \,\P_n\left[w^*(X_{S_\cap};P)(\Delta\aipw{}-\wh{g}(X_{S_\cap}))\,\phi(X_{S_\cap};\nu^*(P),P)\right],$$ where  $\phi(X_{S_\cap};\nu,P):=\sum_{k=1}^K\nu_k\,\E_P[\tau(X_{S_k};P)\mid X_{S_\cap}]-\tauR(P)$, in the sense that $$\sup_{P\in\sP}\left|B_n-\wh{B}_n^*(P)\right|=\o(n^{-1/2}).$$
\end{lemma}

\begin{proof}
  The proof of \cref{lemma:bias-corr-is-close-to-oracle} decomposes $B_n - \wh B_n^*$ into terms
$T_1, T_2, T_2', T_3$, each handled by \cref{lemma.one} together with the supporting results \cref{lemma:AIPW-approx,lemma:rate-of-double-conditional-mean,lemma:rate-of-nu,lemma:rate-of-lambda,lemma:rate-of-weights,lemma:Taylor-expand-weights}. Since we already have shown uniform versions of these results (namely, \cref{lem:lambdawh,lem:gwh,lem:aipw,lem:wwh,lem:nuwh,lem:taylor-weights}),
all of the bounds in the proof of \cref{lemma:bias-corr-is-close-to-oracle} carry through uniformly over $P\in\sP$, thereby completing the proof.
\end{proof}

\begin{lemma}\label{lemma.one} 
Let $\mathcal F$ be a $\sigma$-field, and assume that $\{(U_i,V_i)\}_{i=1}^n$ are independent of each other and mutually independent of $\mathcal F$. Let $\wh a(\cdot)$ be such that $\wh a(V_1),\dots,\wh a(V_n)$ are measurable with respect to $\mathcal G:=\mathcal F\vee\sigma(V_1,\dots,V_n)$ (i.e., $\wh a$ may depend on $\mathcal F$ and on $V_1,\dots,V_n$, but not on $U_1,\dots,U_n$). Assume further that
$$\frac{1}{n}\sum_{i=1}^n (\wh{a}(V_i)-a(V_i))^2 \cdot \var(U_i\mid V_i) = \o(1).$$ Then, 
    $$\frac{1}{n}\sum_{i=1}^n (\wh{a}(V_i)-a(V_i))(U_i-\E\,[U_i\mid V_i])=\o(n^{-1/2}).$$
    Two special cases are worth spelling out:~(a) $\var(U_i\mid V_i)\le C<\infty$ and $\|\wh{a}-a\|_{L^2(\P_n)}=\o(1)$, and (b) $\{(U_i, V_i)\}$ are i.i.d.~with $\E[\var(U\mid V)^{\beta'}<\infty$ and $\|\wh{a}-a\|_{L^{2\beta}(\P_n)}=\o(1)$ where $\beta,\beta'\in(1,\infty)$ such that $1/\beta+1/\beta'=1$.
\end{lemma}

\begin{proof}[Proof of \cref{lemma.one}]
Write $S_n=n^{-1}\sum_{i=1}^n (\wh{a}(V_i)-a(V_i))(U_i-\E\,[U_i\mid V_i])$. Since $\{(U_i,V_i)\}_{i=1}^n$ are independent of each other and independent of $\mathcal F$, conditionally on $\mathcal G$ the variables $U_1,\dots,U_n$ are independent, and each $U_i$ has the conditional law of $U_i$ given $V_i$. Then, because $\wh a(V_1),\dots,\wh a(V_n)$ are $\mathcal G$-measurable, it follows that $\E_P[S_n\mid\mathcal G]=0$, and
\[
\E_P[S_n^2\mid\mathcal G]=\frac{1}{n^2}\sum_{i=1}^n\big(\wh a(V_i)-a(V_i)\big)^2\cdot\var(U_i\mid V_i)=:\frac{R_n}{n}.
\]
        We have $R_n = \o(1)$ by assumption. Hence for any $\eps>0$ and $\eta>0$, Chebyshev's inequality yields
        \[\P\left(\sqrt{n}|S_n|>\eps\right)\le \E\left[\min\{1,\,R_n/\eps^2\}\right]\le \eta/\eps^2 + \P\left(R_n>\eta\right).\]
        Letting $n\to\infty$ and then $\eta\to 0$ completes the proof of the main claim.

\medskip

        The special case (a) is immediate. For case (b), note that Hölder's inequality with exponents $1/\beta+1/\beta'=1$ yields $$R_n\le \|\wh{a}(V_i)-a(V_i)\|_{L^{2\beta}(\P_n)}^2 \|\var(U_i\mid V_i)\|_{L^{\beta'}(\P_n)}.$$
        The second factor above is $\O(1)$ by Markov's inequality, while the first factor is $\o(1)$. This proves $R_n=\o(1)$, as desired to show.
\end{proof}

\begin{lemma}[Conditional variances and integrability]\label{cvarlemma}
Suppose that \cref{assump-at-least-one-is-valid} holds, and that $\var_P[Y\mid A=a,X_{S_\cap}]\le C$ and $\var_P[Y\mid A=a,X_{S_k}]\le C$ almost surely for $a=0,1$ and $k=1,\dots,K$. Then, almost surely, for each $k=1,\dots,K$,
\begin{enumerate}[label=(\alph*)]
\item $\var_P(\tau(X_{S_k})\mid X_{S_\cap})\le4C/\xi$,
\item $\var_P(\aipw{k}\mid X_{S_\cap})\le6C/\xi$, and $\var_P(\Delta\aipw{k}\mid X_{S_\cap})\le24C/\xi$ for $k\ge2$.
\end{enumerate}
If, in addition, \cref{assump:heterogeneity} and condition (iii) of \cref{thm:matching.adj.sets} hold, then
\begin{enumerate}[label=(\alph*),start=3]
\item $\E_P[w^*(X_{S_\cap})|\tau(X_{S_k})|]<\infty$ for each $k$.
\end{enumerate}
\end{lemma}

\begin{proof}
(a) Fix any nonnegative measurable function $h$ of $X_{S_k}$. Using the tower property and $\E_P[A\mid X_{S_k}]=e(X_{S_k})\ge\xi$, we get
\[
\begin{aligned}
\E_P[h(X_{S_k})\mid X_{S_\cap}]&=\E_P\Big[\frac{A}{e(X_{S_k})}h(X_{S_k})\,\Big|\,X_{S_\cap}\Big]\\
&\le\xi^{-1}\E_P[A\,h(X_{S_k})\mid X_{S_\cap}]\le\xi^{-1}\E_P[h(X_{S_k})\mid A=1,X_{S_\cap}],
\end{aligned}
\]
where in the last step we used $\E_P[Ah(X_{S_k})\mid X_{S_\cap}]=P(A=1\mid X_{S_\cap})\,\E_P[h\mid A=1,X_{S_\cap}]$. Using the above inequality with $h:=(\mu_1(X_{S_k})-\E_P[Y\mid A=1,X_{S_\cap}])^2$,  we obtain that
\[
\var_P(\mu_1(X_{S_k})\mid X_{S_\cap})\le\xi^{-1}\var_P(\mu_1(X_{S_k})\mid A=1,X_{S_\cap})\le\xi^{-1}\var_P(Y\mid A=1,X_{S_\cap})\le C/\xi,
\]
where the middle inequality follows from the law of total variance conditionally on $(A=1,X_{S_\cap})$, using $\E_P[Y\mid A=1,X_{S_k}]=\mu_1(X_{S_k})$. The same argument with $1-A$ in place of $A$ gives $\var_P(\mu_0(X_{S_k})\mid X_{S_\cap})\le C/\xi$. Since $\tau=\mu_1-\mu_0$, we conclude that $$\var_P(\tau(X_{S_k})\mid X_{S_\cap})\le2(C/\xi+C/\xi)=4C/\xi,$$
which proves (a). \qed

\medskip
(b) Write the AIPW functional as $\aipw{k}=\tau(X_{S_k})+\epsilon_k$, where $\epsilon_k:=\frac{A}{e(X_{S_k})}(Y-\mu_1(X_{S_k}))-\frac{1-A}{1-e(X_{S_k})}(Y-\mu_0(X_{S_k}))$ satisfies $\E_P[\epsilon_k\mid X_{S_k}]=0$. Hence, by tower property, $\epsilon_k$ is uncorrelated with every function of $X_{S_k}$ conditionally on $X_{S_\cap}$, and
\[
\var_P(\aipw{k}\mid X_{S_\cap})=\var_P(\tau(X_{S_k})\mid X_{S_\cap})+\E_P[\epsilon_k^2\mid X_{S_\cap}].
\]
Since $A(1-A)=0$,
\[
\E_P[\epsilon_k^2\mid X_{S_k}]=\frac{\var_P(Y\mid A=1,X_{S_k})}{e(X_{S_k})}+\frac{\var_P(Y\mid A=0,X_{S_k})}{1-e(X_{S_k})}\le\frac{2C}{\xi}.
\]
Combining the above with (a) gives the first bound in (b). The second bound simply follows from  the inequality $\var(U-V\mid X_{S_\cap})\le2\var(U\mid X_{S_\cap})+2\var(V\mid X_{S_\cap})$. \qed

\medskip

(c) Let $m_k(X_{S_\cap})=\E_P[\tau(X_{S_k})\mid X_{S_\cap}]$. Using the tower property, conditional Jensen's inequality and part (a),
\[
\begin{aligned}
\E_P[w^*|\tau(X_{S_k})|]&\le\E_P[w^*|m_k|]+\E_P\big[w^*\,\E_P[|\tau(X_{S_k})-m_k|\mid X_{S_\cap}]\big]\\[2mm]
&\le(\E_P[(w^*)^2])^{1/2}(\E_P[m_k^2])^{1/2}+(4C/\xi)^{1/2},
\end{aligned}
\]
since $\E_P[w^*]=1$. Condition (iii) of \cref{thm:matching.adj.sets} implies that $\E_P[(w^*)^2]=\psi_P(2\lambda^*)/\psi_P(\lambda^*)^2<\infty$, and $\E_P[m_k^2]<\infty$ follows from the fact that $m_1$ and $g$ have finite moment generating functions in an open neighborhood of the origin and $m_k=m_1-g_k$. This completes the proof.
\end{proof}

\end{document}